\documentclass[11pt]{article}

\usepackage{etoolbox}

\newtoggle{proofs} % toggle if to pust important proofs in appendix or paper: use True to put proofs in paper

\togglefalse{proofs}

\usepackage{macros}

\usemedgeometry

\usepackage{pgfplots}
\usepackage{tikz}
\usepackage{overpic}
\usepackage[square,numbers]{natbib}

\definecolor{darkblue}{rgb}{0,0,.75}
\usepackage[colorlinks=true,linkcolor=darkblue,%
citecolor=darkblue,urlcolor=darkblue]{hyperref}
\usepackage{xcolor}
\usepackage{comment}

\newcounter{storecounter}

\begin{document}

\begin{center}
  {\Large Near Instance-Optimality in Differential Privacy}
  
  \vspace{.3cm}
  
  \begin{tabular}{cc}
    Hilal Asi & John C.\ Duchi \\
    \texttt{asi@stanford.edu} &
    \texttt{jduchi@stanford.edu} \\
    \multicolumn{2}{c}{
      Stanford University
    } \\
  \end{tabular}
\end{center}

\begin{abstract}
  We develop two notions of instance optimality in differential privacy,
  inspired by classical statistical theory: one  by defining a local
  minimax risk and the other by considering unbiased mechanisms and
  analogizing the Cram\'{e}r-Rao bound, and we show that the
  local modulus of continuity of the estimand of interest completely
  determines these quantities. We also develop a complementary
  collection mechanisms, which we term the \emph{inverse sensitivity}
  mechanisms, which are instance optimal (or nearly instance optimal)
  for a large class of estimands. Moreover, these mechanisms uniformly
  outperform the smooth sensitivity framework---on each instance---for
  several function classes of interest, including $\R$-valued continuous
  functions. We carefully present two
  instantiations of the mechanisms for median and
  robust regression estimation with corresponding experiments.
  %%
  %% We investigate instance-specific optimality in differential
  %% privacy. Inspired by classical statistical theory, we develop two notions
  %% of instance-optimality, namely optimality against appropriately unbiased
  %% mechanisms and local-minimax optimality, and show that the inverse 
  %% sensitivity mechanism and its extensions, with a modest increase
  %% in the privacy parameter (logarithmic or constant), are instance-optimal
  %% for a large family of functions. Moreover, for a natural family of
  %% monotone functions and instances, 
  %% we show that the inverse sensitivity mechanism is strictly
  %% instance-optimal. 
  %% We also show that the inverse sensitivity mechanism uniformly
  %% outperforms the smooth sensitivity framework for any
  %% possible instance and monotone function.
  %% Finally, we develop several representative private 
  %% mechanisms, including for the median and convex minimization 
  %% problems.
\end{abstract}

% \thispagestyle{empty}

% \newpage 
% \setcounter{page}{1}

% -*- Mode: latex -*- %

\section{Introduction}
\label{sec:intro}

We study instance-specific optimality for differentially private release of
a function $\func(\ds)$ of a dataset $\ds \in \domain^\dssize$.
In contrast to existing notions of optimality for private procedures,
which measure mechanisms' worst case performance over all instances, 
we develop instance-specific notions to capture
the difficulty of---and potential adaptivity of private mechanisms to---the
given data $x$, rather than some potential worst case.
%% These optimality notions are strong, as instance-optimal mechanisms
%% always enjoy optimal performance, regardless of the underlying data.

The trajectory of differential privacy research and private mechanisms
reflects the desire to be adaptive both to the function $\func$ to be
computed and dataset $\ds$ at hand. \citeauthor{DworkMcNiSm06}'s original
perspective~\cite{DworkMcNiSm06} targets the former, privatizing $f(x)$ by
adding noise commensurate with the global sensitivity $\GS_\func \defeq
\sup_{\ds,\ds': \dham(\ds,\ds') \le 1} |\func(\ds) - \func(\ds')|$ of
$\func$ and adapting to the function $\func$ at hand; more recent work
expands this function adaptive approach~\cite{BlasiokBuNiSt19}. As the
classical approach can be conservative---it does not reflect the sensitivity
of the underlying data $\ds$---a natural idea is to add noise that scales
with the \emph{local sensitivity} (or local modulus of continuity)
$\LS_\func(\sample) \defeq \sup_{\ds': \dham(\ds,\ds') \le 1} |\func(\ds) -
\func(\ds')|$ of $\func$ at the dataset $\sample$. Unfortunately, this fails
to protect privacy, as the sensitivity itself may be compromising, leading
\citet{NissimRaSm07} to propose mechanisms that rely on smooth upper bounds
to the local sensitivity. Yet these mechanisms are complex and, as our
results show, may be conservative.

To understand these phenomena, we take a two-pronged approach,
presenting both lower bounds on error and complementary (near) optimal
mechanisms. We first consider the desiderata a lower bound should
satisfy, following a program \citet{CaiLo15} develop
(see also~\cite{DuchiRu18}):
\begin{enumerate}[(i)]
\item \label{item:instance-specific} the putative lower bound is
  instance specific, depending on the instance $\ds$ at hand;
\item \label{item:super-efficiency} there is a super-efficiency result, so
  that if an estimator outperforms the lower bound at one instance $\ds$, it
  must be (substantially) worse at another instance $\ds'$;
\item \label{item:uniform-achievability}
  the lower bound is achievable.
\end{enumerate}
Point~\eqref{item:super-efficiency} is important, as it eliminates
trivial mechanisms: consider the mechanism
that for all $\ds \in \mc{X}^n$ releases
$\mech(\ds) = \func(\ds_0)$ for a fixed sample $\ds_0$; this is
optimal at $\ds_0$ but poor elsewhere.
More broadly, to satisfy \eqref{item:instance-specific}
and~\eqref{item:super-efficiency}, we develop two sets of lower bounds; one
for \emph{unbiased} mechanisms and the other via the
private local minimax framework we develop.

For the complementary point \eqref{item:uniform-achievability} on uniform
achievability, we study and develop privacy-preserving mechanisms that adapt
to the instance $x$ at hand. Here, rather than measuring the function
$\func$'s change over datasets $\ds$, which may itself be sensitive, we
consider the \emph{inverse} of the local sensitivity: for a target $t$,
the number of individuals we must change in $\ds$ to reach $\ds'$
with $\func(\ds') = t$, which by definition has limited sensitivity
to $\ds$. Thus, defining the length
\begin{equation}
  \label{eqn:inverse-ls}
  \invmodcont_\func(\ds; t) \defeq
  \inf_{\ds'} \left\{ \dham(\ds, \ds')
  \mid \func(\ds') = t \right\},
\end{equation}
%we develop mechanisms that use this inverse to achieve near
%instance-optimality. These are easy to describe; indeed,
it is possible to define (nearly) instance-optimal mechanisms.
Indeed, consider the case that we wish to compute a discrete $f :
\domain^\dssize \to \range$, $|\range| < \infty$.  Then for a privacy
parameter $\diffp > 0$, the \emph{inverse sensitivity mechanism} $\mech$
instantiates the exponential
mechanism~\cite{McSherryTa07} with $\invmodcont_\func$:
\setcounter{storecounter}{\value{equation}}
\renewcommand{\theequation}{\textsc{M}.1}
\begin{equation}
  \label{mech:discrete}
  \P(\mech(\ds) = t)
  \defeq \frac{e^{-\invmodcont_\func(\ds;t) \diffp/2}}{
    \sum_{s \in \range} e^{-\invmodcont_\func(\ds;s) \diffp/2}}.
\end{equation}
This abstract mechanism is perhaps folklore; \citeauthor{JohnsonSh13}'s
\emph{distance-score} mechanism~\cite[Sec.~5]{JohnsonSh13} is an example,
and it appears in various exercise sets and lecture notes on differential
privacy~\cite[e.g.][Ex.~3.1]{Sheffet18}.  Yet its strong optimality and
adaptivity properties---which we begin to delineate, and which include
always outperforming the Laplace and smooth Laplace mechanisms---appear
unexplored, and its extension to the continuous case requires nontrivial
care.
As a motivation for what follows,
a consequence of our results is that for \emph{every} instance
$\ds \in \domain^\dssize$, mechanism~\eqref{mech:discrete} is more likely to
output $\func(\ds)$ than any other mechanism $\wt{\mech}$
satisfying $\P(\wt{\mech}(\ds) = \func(\ds)) \ge
\P(\wt{\mech}(\ds) = t)$.

% Get counters back to normal
\setcounter{equation}{\value{storecounter}}
\renewcommand{\theequation}{\arabic{equation}}

\subsection{Problem setting, definitions, and contributions}
\label{sec:problem-setting}

To situate our contributions, we begin by recalling a few
privacy definitions, starting with the standard definition of
differential privacy~\cite{DworkMcNiSm06, DworkKeMcMiNa06}.
Central to each is that two datasets $x, x' \in \mc{X}^n$ are
\emph{neighboring} if $\dham(x, x') \le 1$, that is, they differ in at most
one example.
\begin{definition}
  \label{def:DP}
  A randomized algorithm $\mech: \domain^\dssize \to \rangemech$ is
  \emph{$(\diffp,\delta)$-differentially private} if for all neighboring
  datasets
  $\ds,\ds' \in \domain^\dssize$ and all
  measurable $S \subset \rangemech$,
  \begin{equation*}
    \P\left(\mech(\ds) \in S \right)
    \le e^{\diffp} \P\left(\mech(\ds') \in S \right) + \delta.
  \end{equation*}
  If $\delta=0$, then $\mech$ is \emph{$\diffp$-differentially private}.
\end{definition}

We also recall Mironov's \renyi-differential privacy~\cite{Mironov17}.  For
$\renparam \ge 1$ and distributions $P$ and $Q$, the \renyi-divergence of order
$\renparam$ is
\begin{equation*}
  \dren{P}{Q}
  \defeq \frac{1}{\renparam - 1} \log{\int{ \left(\frac{dP}{dQ} \right)^\renparam} dQ},
\end{equation*}
where $\drenone{P}{Q} = \lim_{\renparam \downarrow 1} \dren{P}{Q} = \dkl{P}{Q}$.
Adopting a consistent abuse of notation that for random variables
$X \sim P$ and $Y \sim Q$ we set
$\dren{X}{Y} \defeq \dren{P}{Q}$,
we have
\begin{definition}
  \label{def:renyi-DP}
  A randomized algorithm $\mech: \domain^\dssize \to \rangemech$ is
  \emph{$(\renparam,\diffp)$-\renyi-differentially private} if for all
  neighboring datasets $\ds,\ds' \in \domain^\dssize$
  \begin{equation*}
    \dren{\mech(\ds)}{\mech(\ds')} \le \diffp.
  \end{equation*}
\end{definition}

%% We let $\mechfamilydelta$ and
%% $\mechfamilyrenyi$ denote the family of 
%% $(\diffp,\delta)$-DP and $(\renparam,\diffp)$-\renyi-DP
%% mechanisms, respectively.

We measure a mechanism's performance by its expected loss: for a
function $\func : \domain^\dssize \to \range$ and loss $\loss :
\range \times \range \to \R_+$, the expected loss of a mechanism $\mech$ on
instance $\ds$ is
\begin{equation*}
  \E[\loss(\mech(\ds),\func(\ds))].
\end{equation*}
The first notion of optimality we adopt is
the \emph{local minimax risk}, which for a family
$\mechfamily$ of mechanisms is
\begin{equation}
  \label{eq:local-minimax}
  \LMM(\ds,\loss,\mechfamily) \defeq
  \sup_{\ds' \in \domain^\dssize} \inf_{\mech \in \mechfamily} 
  \max_{\tilde \ds \in \{\ds,\ds'\}} 
  E \left[ \loss(\mech(\tilde \ds),\func(\tilde \ds)) \right].
\end{equation}
This definition descends from an insight of Stein's that a problem should be
as hard as its ``hardest one-dimensional
sub-problem''~\cite{Stein56a}. Accordingly, the infimum over mechanisms
$\mech$ is inside the outer supremum, so that $\mech$ may use that the data
will be either $\ds$ or $\ds'$; the local minimax
risk~\eqref{eq:local-minimax} measures the difficulty of privately
estimating $\func(\ds)$ against the hardest alternative instance $\ds' \in
\domain^\dssize$. \citeauthor{CaiLo15} use an analogous
quantity~\cite[Eq.~(1.4)]{CaiLo15} in nonparametric function estimation.

The local minimax risk~\eqref{eq:local-minimax} is adaptive to the
underlying instance $\ds$, and we show in Section~\ref{sec:lower-bounds}
that it satisfies the super-efficiency
requirement~\eqref{item:super-efficiency} we highlight in the introduction.
With this, we may define local minimax optimality.
%
%% %\citet[proposition 2]{DuchiRu18} show that local minimax satisfies
%% %a desirable super-efficiency property; any mechanism
%% %that outperforms the local-minimax lower bound at 
%% %an instance $\ds$ has to incur significantly higher
%% %loss than the local-minimax lower bound at another
%% %instance $\ds'$. \hacomment{explain this later?}
%% We may now define local-minimax optimal mechanisms.
\begin{definition}
  \label{def:local-minimax-optimal}
  Let $\loss : \range \times \range \to
  \R_+$ and $\mechfamily$ be a family of mechanisms. A mechanism
  $\mech : \mc{X}^n \to \rangemech$ is
  \emph{local minimax optimal} for the family $\mechfamily$ if
  there exists a universal (numerical) constant $C < \infty$ such that
  \begin{equation*}
    \E\left[\loss(\mech(\ds), \func(\ds))\right] \le
    C \cdot \LMM(\ds, \loss, \mechfamily).
  \end{equation*}
  %% A randomized algorithm $\mech: \domain^\dssize \to \rangemech$ is
  %% \emph{$(\loss,\nearoptconst)$-local minimax optimal for the family 
  %% $\mechfamily = \mechfamilydelta$} (or $\mechfamily = \mechfamilyrenyi$) if 
  %% $\mech \in \mechfamily_{\nearoptconst \diffp, \delta} $
  %% (or $\mech \in \mechfamily_{\alpha,\nearoptconst \diffp} $) 
  %% and for any instance $\ds \in \domain^\dssize$, there
  %% exists $c < \infty$ (independent of $\diffp$ and $\ds$) 
  %% such that
  %% \begin{equation*}
  %% 	\E \left[ \loss(\mech(\ds),\func(\ds)) \right] 
  %% 			\le c \LMM(\ds,\loss,\mechfamily).
  %% \end{equation*}
\end{definition}
\noindent
In our development, we often consider families of $(\diffp,
\delta)$-private mechanisms, and show that a $(c\diffp, c
\delta)$-differentially private mechanism satisfies
Definition~\ref{def:local-minimax-optimal}, where $c$ is a constant we
specify; we say such a mechanism is \emph{local minimax $c$-optimal}.
An alternative way to understand $c$-optimality is that a $c$-optimal
mechanism requires a sample $\wt{\ds}$ of size approximately $c n$ (rather
than $n$) to achieve the local minimax risk $\LMM(\ds)$ for $\ds \in \mc{X}^n$.
Indeed, privacy amplification by subsampling
techniques~\cite{BeimelKaNi10, BalleBaGa18} show that a mechanism
$\wt{\mech}$ that randomly subsamples a sample $\wt{\ds}$ of $cn$ examples
down to $\ds \in \mc{X}^n$, then applies $\mech(\ds)$, is
$\diffp$-differentially private.

%% , so that any $c$-optimal mechanism
%% roughly corresponds to being optimal except requiring $c$ times as much data
%% (or a factor $c$ worse privacy parameter $\diffp$).

The local minimax benchmark~\eqref{eq:local-minimax} eliminates the
possibility of trivial mechanisms we mention above. An alternative is to
consider restricted families of mechanisms. In this context, we borrow from
the statistical tradition of unbiased estimators to allow analogues of the
classical Cram\'{e}r-Rao bounds. We may then
define unbiased mechanisms.
\begin{definition}[\cite{LehmannRo05}, Ch.~1.5, Eq.~(1.9)]
  \label{def:unbiased-mech}
  Let $\loss : \range \times \range \to \R_+$.  A randomized algorithm
  $\mech: \domain^\dssize \to \rangemech$ is \emph{$\loss$-unbiased} if for
  any $\ds \in \domain^\dssize$ and $t \in \range$, $\E[
    \loss(\mech(\ds),\func(\ds))] \le \E[ \loss(\mech(\ds),t) ]$.
%  For $\biasconst < \infty$, we say
%  $\mech$ is \emph{$(\loss,\biasconst)$-unbiased} if
%  \begin{equation*}
%    \E \left[
%      \loss(\mech(\ds),\func(\ds)) \right] \le
%    \biasconst \cdot \E \left[
%      \loss(\mech(\ds),t) \right].
%  \end{equation*}
\end{definition}
\noindent
For the squared error $\loss(s, t) = (s - t)^2$, we recover the
familiar equality $\E[\mechanism(\ds)] = f(\ds)$, while an
$\ell_1$-unbiased mechanism has median
$f(\ds)$. A mechanism is unbiased for the \zo loss
if $\P(\mech(x) = f(x)) \ge \P(\mech(x) = t)$
for all $t \in \range$. Most standard mechanisms, including the Laplace
and Gaussian mechanisms, are unbiased.
%
%% Definition~\ref{def:unbiased-mech} aims to exclude trivial mechanisms that try to minimize the loss of a single mechanism (such as $\mech_\ds(\ds') = \func(\ds)$). The family of unbiased mechanisms is fairly large and includes most of existing private algorithms. Indeed, for the $0\text{-}1$ loss $\ell_{0,1}(t_1, t_2) = \indic{t_1 \neq t_2}$,  a mechanism $\mech$ is $\ell_{0,1}$-unbiased if $\P(\mech(\ds)=\func(\ds)) \ge \P(\mech(\ds)=t)$ for any $t \in \range$. 
%% For $\ell_1$-loss $\ell_1(s, t) = \left| s - t \right|$, a mechanism $\mech$ is $\ell_1$-unbiased if the median of the mechanism includes $\func(\ds)$ i.e., $\P(\mech(\ds) \ge \func(\ds)) \ge 1/2$
%% and $\P(\mech(\ds) \le \func(\ds)) \ge 1/2$.
%% For $\ell_2$-loss $\ell_2(s,t) = \left| s - t \right|^2$, a mechanism $\mech$ is $\ell_2$-unbiased if $\E[\mech(\ds)] = \func(\ds)$.
%
We then can parallel Definition~\ref{def:local-minimax-optimal}.
\begin{definition}
  \label{def:instance-optimal}
  Let $\nearoptconst \ge 1$ and the loss $\loss : \range \times \range \to
  \R_+$.  A randomized algorithm $\mech: \domain^\dssize \to \rangemech$ is
  \emph{$\nearoptconst$-optimal against $\loss$-unbiased mechanisms}
  if $\mech$ is $\nearoptconst \diffp$-differentially private,
  and for any $\diffp$-differentially private and $\loss$-unbiased mechanism
  $\mech_{\textup{unb}}: \domain^\dssize \to \rangemech$
  \begin{equation*}
    \E \left[ \loss(\mech(\ds),\func(\ds)) \right] \le \E
    \left[ \loss(\mech_{\textup{unb}}(\ds),\func(\ds)) \right]
    ~~ \mbox{for~all~}
    \ds \in \domain^\dssize.
  \end{equation*}
\end{definition}
\noindent
%% We will also consider nearly instance optimal mechanisms that achieve
%% Definition~\ref{def:instance-optimal} up to lower-order additive constants that
%% should be clear from context.
%%
%% In contrast to existing notions of optimality in differential privacy,
%% definition~\ref{def:instance-optimal} provides exceptionally strong
%% guarantees that hold for all instances.
%% Definition~\ref{def:instance-optimal} states that an instance-optimal
%% mechanism, using a slightly larger $\diffp$, achieves the smallest
%% expected loss for all datasets. Moreover, for any
%% $\nearoptconst \diffp$-DP mechanism $\mech$, standard privacy
%% amplification by sub-sampling~\cite{BeimelKaNi10,BalleBaGa18}
%% gives an $\diffp$-DP mechanism $\mech'$ using a sample approximately
%% $\nearoptconst$ times larger.
%% We frequently write ``instance optimal'' as a shorthand
%% for Definition~\ref{def:instance-optimal}. 
%% %Extensions
%% %of the definition to nearly unbiased or instance-optimal mechanisms
%% %are immediate, so we treat them tacitly.

%% The second notion of optimality we adopt in this work is based on the
%% local minimax risk at a given instance. 
%% More precisely, for a given
%% loss function $\loss : \range \times \range \to \R_+$, we define the
%% \emph{local minimax risk} at instance $\ds$ for a family $\mechfamily$ 
%% of mechanisms to be

\subsubsection{Contributions and outline}

We highlight the three main contributions in this paper.

\paragraph{Instance-specific notions of optimality.}
We move beyond worst-case
(minimax) loss to study instance-optimality in differential privacy,
providing two potential definitions and their consequences and satisfying
desiderata~\eqref{item:instance-specific} and~\eqref{item:super-efficiency}.
Making this concrete, let $\dham$ denote the Hamming distance and
$(\range, \dt)$ be a metric space; then the
\emph{local modulus of continuity} of a function $f : \mc{X}^n \to \range$
at $\ds \in \mc{X}^n$ is
\begin{equation}
  \label{eqn:modcont-def}
  \modcont_\func (\ds;k) = \sup_{\ds' \in \domain^\dssize} \left\{
  \dt(\func(\ds),\func(\ds')) : \dham(\ds,\ds') \le k \right\}.
\end{equation}
Our results show that this local modulus of continuity
characterizes the optimal error in estimation at each instance $\ds \in
\mc{X}^n$ for differentially private mechanisms, including those satisfying
only relaxed (e.g.\ \renyi) definitions of privacy.  We provide these
results in Section~\ref{sec:lower-bounds}.

\paragraph{Nearly instance-optimal mechanisms.}
We show that the inverse sensitivity mechanism \eqref{mech:discrete} and its
arbitrary-valued
extensions (which we develop in Sec.~\ref{sec:mech-continuous})
are $\nearoptconst$-optimal for reasonable
$\nearoptconst$ for many functions $\func$ and losses $\loss$ of interest.
For example, in Section~\ref{sec:optimal-discrete} we show that
mechanism~\eqref{mech:discrete} is $4$-optimal
against unbiased mechanisms for the 0-1-loss
(Def.~\ref{def:instance-optimal}).  For general losses and
functions $\func$, we show the inverse sensitivity mechanisms
are $O(\log n)$-optimal in both
optimality senses we consider (Definitions~\ref{def:local-minimax-optimal}
and~\ref{def:instance-optimal}).

%% We also extend our results to
%% more general losses, including the standard $\ell_1$ and $\ell_2$ losses,
%% showing that the mechanism~\eqref{mech:discrete} is
%% $(\loss,\nearoptconst)$-instance optimal (and local-minimax optimal) for
%% $\loss(s,t)=\ell(\mbox{dist}(s,t))$ for a non-decreasing $\ell : \R_+ \to
%% \R_+$, with $\nearoptconst = O( \log \dssize)$.  In
%% Section~\ref{sec:mech-continuous} we extend the inverse sensitivity
%% mechanism~\eqref{mech:discrete} to arbitrary-valued cases, developing nearly
%% instance-optimal mechanisms for bounded functions $\func: \domain^\dssize
%% \to \range$ for measurable spaces $\range$.  This requires designing a
%% smoother version of $\invmodcont_\func$ to address infinite ranges.

\paragraph{Instance-optimality for sample-monotone functions.}
We define a class of what we term \emph{sample-monotone} functions (see
Section~\ref{sec:monotone-functions}), which are reasonably-behaved with
respect to changes in the sample $\ds \in \mc{X}^n$, and which includes all
$\R$-valued continuous functions on convex domains.  We show that the
inverse sensitivity mechanism is $O(1)$-optimal for many such functions
$\func$, and nearly $O(\log \log n)$-optimal for any sample-monotone
function. Finally, we show in Section~\ref{sec:monotone-functions} that the
inverse sensitivity mechanism uniformly outperforms the Laplace and smooth
Laplace mechanisms over all $\ds$ for any sample-monotone function.

\paragraph{Applications.}
An important component of this work is methodological, and we aim to develop
practicable procedures. We include detailed examples for estimating the
median of a sample (Sec.~\ref{sec:median}) and robust regression
(Sec.~\ref{sec:example-regression-new}). As a side benefit of this
development, we show that the inverse-sensitivity offers a quadratic
improvement in sample complexity over standard smooth Laplace mechanisms. We
include representative experiments in
Section~\ref{sec:experiments}.

\paragraph{Notation}
We let bold symbols $\ds \in \mc{X}^n$ denote samples
and non-bold symbolx $x \in \mc{X}$ denote individual examples.
For $\ds,\ds' \in
\domain^\dssize$, $\dham(\ds,\ds')$ is the Hamming distance.
%We define the modulus of continuity $\modcont_\func (\ds;k) = \sup \left\{ \left| \func(\ds) - \func(\ds') \right| : \dham(\ds,\ds') \le k  \right\} $.
%% For a metric space $\range$ with distance $\dt$, the local modulus of
%% continuity of $\func : \domain^\dssize \to \range$
%% at $\ds \in \domain^\dssize$ is
%% \begin{equation}
%%   \label{eqn:modcont-def}
%%   \modcont_\func (\ds;k) = \sup_{\ds' \in \domain^\dssize} \left\{
%%   \dt(\func(\ds),\func(\ds')) : \dham(\ds,\ds') \le k \right\}.
%% \end{equation}
%and the inverse modulus of continuity $\invmodcont_\func (\ds;t) = \inf \left\{ \dham(\ds,\ds') : \func(\ds') =  t  \right\} $.
Using the local modulus~\eqref{eqn:modcont-def},
the local sensitivity of $\func : \mc{X}^n \to \range$
at instance $\ds$ is $\LS_\func(\ds) =
\modcont_\func (\ds;1)$, and the global sensitivity of $\func$ is $\GS_\func
= \sup_{\ds \in \domain^\dssize} \modcont_\func (\ds;1) $.
For $K \in \N$, we let $[K]
= \{1,2,\dots,K\}$. For a subset $\range$ of a vector space with norm
$\norm{\cdot}$, we let $\diam(\range) = \sup_{s,t \in \range} \norm{s - t}$.
We write $a \lesssim b$ or $a=O(b)$ if there
is a universal (numerical)
constant $c < \infty$ such that $a < cb$.

% notation in this paper:

%$\func : \domain^\dssize \to \range$

%$\ds \in \domain^\dssize$

%$\modcont_\func (\ds;k) = \sup \left\{ \left| \func(\ds) - \func(\ds') \right| : \dham(\ds,\ds') \le k  \right\} $

%$\invmodcont_\func (\ds;t) = \inf \left\{ \dham(\ds,\ds') : \func(\ds') = \func(\ds) + t  \right\} $ say a word about non-numeric values

%$\mech(\ds) = \func(\ds) + \mc{Z}$

%$\LS_\func(\ds) = \modcont_\func (\ds;1) $

%$\GS_\func = \sup_{\ds \in \domain} \modcont_\func (\ds;1) $

\subsection{Background and related work}

The standard method to achieve $\diffp$-DP for releasing
$\R$-valued $\func$ is the \emph{Laplace mechanism}, which adds Laplace noise
proportional to global sensitivity~\cite{DworkMcNiSm06},
\begin{equation}
  \label{mech:laplace}
  \mechlap(\sample) \defeq
  \func(\ds) + \frac{\GS_\func}{\diffp} \laplace(1).
\end{equation}
To
address that the Laplace mechanism adds noise that must scale with the
worst-case sensitivity of $\func$, \citet{NissimRaSm07} introduce the smooth
sensitivity framework, showing that for
appropriate upper bounds $S(\ds)$ on the local sensitivity, the \emph{smooth
  Laplace mechanism}
\begin{equation}
  \label{mech:smooth-laplace}
  \mechsmlap(\sample)
  \defeq \func(\ds) + \frac{2 S(\ds)}{\diffp} \laplace(1)
\end{equation}
is $(\diffp,\delta)$-DP. Here, one requires that $\LS(\ds) \le S(\ds)$ and
that $S(\ds) \le e^{\beta} S(\ds')$ for neighboring instances $\ds,\ds' \in
\domain^\dssize$, where $\beta = \frac{\diffp}{2 \log(2/\delta)}$.  A main
motivation of the smooth Laplace mechanism~\cite{NissimRaSm07}
is to calculate the median of a dataset, and there are instances where the
mechanism~\eqref{mech:smooth-laplace} adds noise scaling as
$\frac{2}{\diffp^2 \dssize}$, while
mechanism~\eqref{mech:laplace} adds noise $\frac{1}{\diffp}$. Yet as the
results in this paper demonstrate, even replacing $S$ with the tighter local
sensitivity $\LS(\ds)$ in the smooth Laplace mechanism---which is not
differentially private---is instance-suboptimal, and the smooth
sensitivity framework does not provide $\diffp$-DP with exponentially
decaying noise.

Our approach is a natural descendant of work that in some sense inverts
sensitivity measurements.  The two most salient works here are
\citeauthor{DworkLe09}'s propose-test-release framework~\cite{DworkLe09} and
\citeauthor{SmithTh13lasso}'s instantiation for high-dimensional
regression (Lasso) problems~\cite{SmithTh13lasso}. Briefly, algorithms in this
framework test whether the maximal local sensitivity in a neighborhood of
the given dataset $\ds$ is upper bounded by a prespecified value $\beta$
(testing that the dataset $\ds$ is ``far'' from any other $\ds'$ for which
the bound $\beta$ fails to hold), then add noise
$\laplace(\beta/\diffp)$ if the test passes. The framework, however,
cannot provide pure $\diffp$-DP, and it does not enjoy instance-optimal
performance as it adds large noise---or fails---if there is even one dataset
in the local neighborhood of the given instance with large local
sensitivity.
%% Similarly, \citet{SmithTh12} (in the
%% context of solving high-dimensional sparse regression, or Lasso, problems)
%% test whether a function $\func(\ds)$ is \emph{constant} on neighborhoods of
%% some fixed size of $\ds$ by finding the distance to an alternative dataset
%% $\ds'$ where $\func(\ds')$ is unstable.
Our examples demonstrating the instance sub-optimality of the smooth
sensitivity framework also show sub-optimality of the propose-test-release
framework.

The discrete version of the inverse sensitivity
mechanism~\eqref{mech:discrete} is in a sense folklore---indeed,
\citet{McSherryTa07} considered it but did not include it in their
development of the exponential mechanism (personal communication). Variants
of it are also explicit or implicit in a number of papers and
lectures~\cite{JohnsonSh13, Sheffet18, BunKaStWu19}.  Yet it appears that no
work investigates the instance-optimality guarantees of
mechanism~\eqref{mech:discrete}, nor does any work explain the advantages of
the framework over other approaches to adding data-dependent noise, such as
smooth sensitivity or propose-test-release; indeed, most work developing
data-dependent noise addition mechanisms is based on these
frameworks~\cite{SmithTh13lasso, GonemGi18, BunSt19}.

%% \hacomment{
%% 1. I'm not really sure this is the right way 
%% 	to cite lecture notes. We should cite this~\cite{Sheffet18}
%%     for median estimation... \\
%% 2. Should we mention that the discrete version, which they sometimes use
%%    for real valued functions, is not in fact instance optimal as we
%%    need the continuous version?}

%

%\hacomment{add some of these references, especially the first one
%	\cite{JohnsonSh13, BunKaStWu19, RaskhodnikovaSm15}}

Our results also belong to an intellectual tradition, familiar from
statistics~\cite{VanDerVaart98, LehmannRo05}, of optimality
relative to classes of algorithms or procedures. Most salient is the
existence of uniformly minimum variance unbiased estimators
(UMVUs)~\cite{LehmannRo05}. In this vein, it is
unsurprising that unbiased mechanisms
(Definitions~\ref{def:unbiased-mech} and \ref{def:instance-optimal})
allow instance-specific optimality results. In more sophisticated settings,
instance-optimality---in the sense of estimators optimal for
the given population generating the sample $\sample \in
\domain^\dssize$---is challenging~\cite[cf.][Ch.~8]{VanDerVaart98},
necessitating local minimax constructions or regular estimators.

In the differential privacy literature, \citet{HardtTa10} develop
minimax lower bounds for estimating linear functions, which
\citet{De12} and \citet{NikolovTaZh13} extend.
A different line of work develops optimality results in the oblivious
setting where private mechanisms add noise independent of the underlying
instance~\cite{SoriaDo13,GengKaOhVi15,GengVi16,GengDiGuKu19} (a natural
restriction on the class of estimators).  \citet{SoriaDo13} and
Geng et al.~\cite{GengVi16,GengDiGuKu19}
independently develop the optimal mechanism in
this setting for one-dimensional real valued functions, showing that the
optimal noise has a staircase-shaped probability density for symmetric
losses. Despite these results, their obliviousness implies the
mechanisms are generally suboptimal.

%% A related but different notion of optimality applies to private mechanisms
%% in a Bayesian and minimax
%% frameworks~\cite{GhoshRoSu09,GupteSu10,BrennerNi10}, where consumers wish to
%% minimize their losses in computing $\func(\ds)$ under Bayesian or
%% worst-case priors.  \citet{GhoshRoSu09} develop a universally optimal
%% discrete variant of the Laplace distribution for single
%% count functions in the Bayesian setting, which \citet{GupteSu10} extend
%% to minimax frameworks.  Unfortunately, \citet{BrennerNi10} prove
%% impossibility results that universally optimal mechanisms cannot
%% exist for more general functions.

% -*- Mode: latex -*- %

\section{Instance-dependent lower bounds}
\label{sec:lower-bounds}

We begin our theoretical investigation by proving 
lower bounds on the loss of private mechanisms
that capture the hardness of the underlying instance.
%The following sections provide 
%lower bounds for our two notions of optimality,
%namely lower bounds for unbiased mechanisms
%and lower bounds for the local minimax loss.

%% First, we prove a lower bound on the $\zo$ loss
%% for unbiased mechanisms. Then, we prove lower bounds
%% for general losses for our two notions of optimality.

\subsection{Unbiased mechanisms and worst-case bounds}

Our first set of lower bounds applies to unbiased mechanisms, where the
lower bounds are relatively straightforward to develop. They rely on the
idea that if a mechanism is $\diffp$-differentially private and unbiased, it
must assign ``enough probability mass'' to each possible output $t$ of
$\func(\ds)$, an idea present in other lower bounds (cf.~\cite{HardtTa10,
  BarberDu14a}). This idea allows an essentially immediate corresponding minimax
result.
We first consider discrete
functions $\func : \domain^\dssize \to \range$, $|\range| < \infty$,
with the $\zo$ loss $\losszo(s,t) = \indic{s \neq t}$, so
$\E[\loss(\mech(\ds),\func(\ds))] = 1 - \P(\mech(\ds) = \func(\ds))$.
See Appendix~\ref{sec:proofs-0-1-loss} for the proof of the next
proposition.
\begin{proposition}[Lower bounds for $\zo$ loss]
  \label{proposition:lower-01-loss}
  Let $\func: \domain^\dssize \to \range$ and let $\mech$ be
  $\diffp$-DP. Then
  \begin{equation*}
    \inf_{\ds \in \domain^\dssize} \P \left(\mech(\ds) = \func(\ds) \right)
    \le   \inf_{\ds \in \domain^\dssize} \frac{1}{
      \sum_{t \in \range} {e^{-\invmodcont_\func(\ds;t) \diffp}}}.
  \end{equation*}
  If $\mech$ is also $\losszo$-unbiased, then for any
  instance $\ds \in \domain^\dssize$,
  \begin{equation*}
    \P \left(\mech(\ds) = \func(\ds) \right)
    \le \frac{1}{ \sum_{t \in \range}
      {e^{-2 \invmodcont_\func(\ds;t) \diffp}}}.
  \end{equation*}
\end{proposition}
%% \noindent
%% See Appendix~\ref{sec:proofs-0-1-loss} for the proof.
\noindent
Proposition~\ref{proposition:lower-01-loss} shows that worst-case bounds may
be pessimistic for many instances $\ds$, depending on
$\invmodcont_\func(\ds;t)$. Consider an instance $\ds$ with a highly stable
neighborhood: $\invmodcont_\func(\ds;t) \gg 1$ for any $t \neq \func(\ds)$.
Instance-dependent bounds show that we may hope to get
$\P(\mech(\ds)=\func(\ds)) \approx 1$ for this instance (as the
mechanism~\eqref{mech:discrete} achieves), in contrast to worst-case bounds.
%% which may be small.

Now we extend our previous lower bound to general loss
functions. We consider functions  
$\func:\domain^\dssize \to \range$ 
for a metric space $\range$ with
distance $\dt : \range \times \range \to \R_+$ and loss function $\loss(s,t)
= \ell(\dt(s,t))$ where $\ell: \R_+ \to \R_+$ is non-decreasing.
We have the following lower bound in this setting,
which highlights the intuitive centrality of the local modulus~\eqref{eqn:modcont-def}: the more sensitive the function $\func$ is to changes in the
underlying sample $\ds$, the more challenging it should be to estimate.
We prove the theorem in Appendix~\ref{sec:proof-bound-general-loss}.
\begin{theorem}[Lower bound for general loss]
  \label{thm:lower-bound-general-loss}
  %Let $\loss(s,t) = \ell(\dt(s,t))$ for a
  %non-decreasing function $\ell : \R_+ \to \R_+$
  %and $\func: \domain^\dssize \to \range$.
  If $M$ is $\diffp$-DP, then for any $k \ge 1$,
  \begin{equation*}
    \sup_{\ds \in \domain^\dssize} \E\left[\loss(\mech(\ds),\func(\ds))\right] 
    \ge \sup_{\ds \in \domain^\dssize}
    \frac{\ell(\modcont_\func(\ds; k) / 2)}{e^{k \diffp} + 1}.
    %% \sup_{\ds \in \domain^\dssize} \E\left[\loss(\mech(\ds),\func(\ds))\right] 
    %% \ge \sup_{\ds \in \domain^\dssize}
    %% \frac{e - 1}{2 } \sum_{k=1}^{ n\diffp} e^{-2k } \ell(\modcont_\func(\ds;k/\diffp)/2).
  \end{equation*}
  If $M$ is additionally $\loss$-unbiased, then for any 
  $\ds \in \domain^\dssize$,
  \begin{equation*}
    \E\left[\loss(\mech(\ds),\func(\ds))\right] 
    \ge
    \frac{\ell(\modcont_\func(\ds; k)/2)}{
      e^{2k \diffp} + 1}.
  \end{equation*}
\end{theorem}

\subsection{Local minimax bounds and super-efficiency}

The lower bounds for unbiased mechanisms in
Proposition~\ref{proposition:lower-01-loss} and
Theorem~\ref{thm:lower-bound-general-loss} are the analogues of the
Cram\'{e}r-Rao bound in classical estimation, with the concomitant failure
to apply beyond unbiased estimators. Accordingly, we turn now to the local
minimax risk~\eqref{eq:local-minimax}, which we characterize in the
following theorem. The result applies whenever the loss in estimation is
measured by a distance $\dt$ on the target space $\range$.  In the theorem,
recall the definition~\eqref{eqn:modcont-def} of the local modulus of
continuity, $\modcont_\func(\ds; k) = \sup_{\ds'} \{\dt(\func(\ds),
\func(\ds')) : \dham(\ds, \ds') \le k\}$.
\begin{theorem}[Local-minimax lower bounds]
  \label{thm:lower-bound-local-minimax}
  Let $\loss(s,t) = \ell(\dt(s,t))$ for a
  non-decreasing function $\ell : \R_+ \to \R_+$
  and distance $\dt$ on $\range$.
  Let $\func: \domain^\dssize \to \range$
  and $\ds \in \mc{X}^n$ be an arbitrary sample.
  If $\mechfamily_\diffp$ is the
  collection of
  of $\diffp$-differentially private mechanisms,
  \begin{subequations}
    \label{eqn:local-minimax-calculation}
    \begin{equation}
      \frac{1}{4} \max_{k \le n}
      \left\{\ell(\modcont_\func(\ds; k)/2) e^{-k \diffp}\right\}
      \le
      \LMM(\ds,\loss,\mechfamily_\diffp).
      %% \le
      %% \max_{k \le n}
      %% \left\{\frac{1}{1 + e^{-k \diffp / 2}}
      %% \ell(\modcont_\func(\ds; k))\right\}.
      \label{eqn:dp-local-minimax}
    \end{equation}
    If $\mechfamily$ is either
    the collection of $(\diffp, \delta)$-DP 
    or $(\renparam, \diffp/2)$-\renyi-DP mechanisms with
    $\renparam \ge 1 + 2 \diffp^{-1} \log \frac{1}{\delta}$, then
    for $K(\diffp, \delta) = \min\{\diffp^{-1} \log \frac{1}{\delta},
    \delta^{-1/2}\}$,
    \begin{equation}
      \frac{1}{8}
      \max_{k \le K(\diffp, \delta)}
      \left\{\ell(\modcont_\func(\ds; k)/2) e^{-k \diffp}\right\}
      \le
      \LMM(\ds, \loss, \mechfamily).
      %% \le
      %% \max_{k \le n} \left\{\frac{1}{1 + e^{-k \diffp / 2}}
      %% \ell(\modcont_\func(\ds; k)) \right\}.
      \label{eqn:adp-local-minimax}
    \end{equation}
    If $\mechfamily_{\renparam, 2\diffp^2}$ is the collection of
    $(\renparam, 2 \diffp^2)$-\renyi-DP mechanisms, where $\renparam \ge 1$,
    then
    \begin{equation}
      \label{eqn:renyi-local-minimax}
      \frac{1}{8}
      \ell\left(\modcont_\func\left(\ds; 1 / (2 \diffp)\right) / 2
      \right)
      \le
      \LMM(\ds,\loss,\mechfamily_{\alpha, 2 \diffp^2}).
    \end{equation}
    For each of the preceding families of mechanisms $\mechfamily$, we have
    \begin{equation*}
      \LMM(\ds, L, \mechfamily)
      \le \max_{k \le n}
      \left\{\frac{1}{1 + e^{k \diffp / 2}} \ell(\modcont_\func(\ds; k))
      \right\}.
    \end{equation*}
  \end{subequations}
\end{theorem}
\noindent
See Appendix~\ref{proof:thm-lower-bound-local-minimax} for a proof.

Theorem~\ref{thm:lower-bound-local-minimax} characterizes, to within
numerical constants, the local minimax risk for each collection of
mechanisms in Sec.~\ref{sec:problem-setting}, showing that for
$\diffp$-private mechanisms, $\LMM(\ds)$
it should scale as $\max_{k \le n}
\ell(\modcont_\func(\ds; k)) e^{-k\diffp}$.  Thus, we see that in an
essential way, the local modulus of continuity~\eqref{eqn:modcont-def}
determines the local minimax risk $\LMM$ at a sample
$\ds$. We expect that the
maxima~\eqref{eqn:local-minimax-calculation} over $k$ should be generally
achieved for $k \asymp \frac{1}{\diffp}$, so that we expect
that for $\diffp$-DP mechanisms (or $(\diffp, \delta)$- or
\renyi-DP mechanisms), the local minimax risk should scale as
\begin{equation}
  \label{eqn:expected-lmm-scaling}
  \LMM(\ds, L, \mechfamily) \asymp
  \ell\left(\modcont_\func\Big(\ds; \frac{1}{\diffp}\Big)\right).
\end{equation}
We \emph{always} have the lower bound $\ell(\modcont_\func(\ds;
\diffp^{-1}))$ by choosing $k = 1/\diffp$ in the
inequalities~\eqref{eqn:local-minimax-calculation}. As an example for
attaintment~\eqref{eqn:expected-lmm-scaling}, if $\loss(s, t) = \dt(s, t)$
and the modulus cannot grow too exponentially quickly, e.g.,
$\modcont_\func(\ds; k) / \modcont(\ds; k_0) \le \exp(k / k_0)$ for $k \ge
k_0 = 1/\diffp$, we have $\max_k \modcont_\func(x; k) e^{-k \diffp} \le
\modcont_\func(x; 1/\diffp)$; similar calculations show
equality~\eqref{eqn:expected-lmm-scaling} under other
conditions.

Regardless of the growth of $\modcont_\func$, we have bounds on the
local minimax error that
satisfy the desideratum~\eqref{item:instance-specific}, that they depend on
the given instance $\ds$. Yet as the local minimax
risk~\eqref{eq:local-minimax} uses (essentially) a one-dimensional
sub-problem of the full estimation problem and still requires an outer
supremum, it is not immediately clear that we satisfy the other two
desiderata: \eqref{item:super-efficiency}, that no estimator can achieve
better performance than the bounds~\eqref{eqn:local-minimax-calculation}
without sacrificing performance elsewhere, and
\eqref{item:uniform-achievability} that it is achievable.  We return to the
latter point in Section~\ref{sec:upper-bounds}, turning now to a
super-efficiency result.

We begin with a general proposition, which shows that any improvements over
an expected loss of (roughly) the modulus of
continuity~\eqref{eqn:modcont-def} at distance $1/\diffp$, as in our
expected scaling~\eqref{eqn:expected-lmm-scaling}, necessarily result in
worse expected losses elsewhere. We focus on $\R$-valued statistics for
simplicity, deferring the proof of the next proposition to
Appendix~\ref{sec:proof-super-efficiency}.
\begin{proposition}
  \label{proposition:super-efficiency}
  Let $\loss(s,t) = \ell(|s-t|)$ for a
  non-decreasing function $\ell : \R_+ \to \R_+$.
  Let $\func : \domain^\dssize \to \R$ and
  $\ds \in \domain^\dssize$.
  Assume the mechanism $\mech$ is $\diffp$-DP and
  for some $\gamma \le \frac{1}{2e}$ achieves
  \begin{equation*}
    \E[\ell(|\mech(\ds) - \func(\ds)|)]
    \le \gamma \ell\left(\modcont_\func(\ds; 1 / \diffp)/2\right).
  \end{equation*}
  Then there exists a sample $\ds' \in \domain^\dssize$ with
  $\dham(\ds, \ds') \le \frac{\log(1/2\gamma)}{2 \diffp}$ such that
  \begin{align*}
    \E[\ell(|\mech(\ds') - \func(\ds')|)] 
    \ge \frac{1}{4} \ell\left(\frac{1}{4} \modcont_\func\left(\ds';
    \frac{\log(1 / 2\gamma)}{2 \diffp}\right)\right).
  \end{align*}
\end{proposition}
\noindent
Roughly, this result says that if any mechanism achieves expected loss
better than the local modulus of continuity at radius $1/\diffp$ around
$\ds$, there is a nearby $\ds'$ such that the expected loss is
quantitatively higher---the local modulus at distance $\frac{1}{\diffp} \log
\frac{1}{\gamma}$. Thus, any method achieving expected error better than our
local minimax risk $\LMM$ at \emph{any} point $\ds$ must necessarily be much
worse at other points.  As one example, we have the following corollary,
which assumes that the modulus $\modcont_\func$ grows
no faster than exponentially, but captures the behavior we expect.
\begin{corollary}
  Let $\gamma \le \frac{1}{2e^3}$.  Let $\func : \mc{X}^n \to \R$ and $\ds
  \in \mc{X}^n$ be such that for any $\ds' \in \mc{X}^n$ satisfying
  $\dham(\ds', \ds) \le \frac{\log(1 / 2 \gamma)}{2 \diffp}$, we have
  $\modcont_\func(\ds', k) \le \modcont_\func(\ds', k_0) e^{k / k_0}$ for $k
  \ge k_0 = \frac{2}{\diffp}$. Let $\mechfamily_\diffp$ denote the
  collection of $\diffp$-differentially private mechanisms and
  $\diffp(\gamma) \defeq 4 \diffp / \log \frac{1}{2 e^2 \gamma}$.  Then if
  $\mech$ is $\diffp$-differentially private, whenever
  \begin{equation*}
    \E\left[|\mech(\ds) - \func(\ds)|\right]
    \le \gamma \cdot \LMM(\ds; |\cdot|, \mechfamily_\diffp),
  \end{equation*}
  there exists a sample $\ds' \in \mc{X}^n$ with
  $\dham(x, x') \le \frac{\log(1/(2 e^2 \gamma))}{2 \diffp}$ such that
  \begin{equation*}
    \E\left[|\mech(\ds') - \func(\ds')|\right]
    \ge \frac{1}{16} \cdot
    \LMM\left(\ds', |\cdot|, \mechfamily_{\diffp(\gamma)}\right).
  \end{equation*}
\end{corollary}
\noindent
An improvement over the local minimax rate $\LMM$ at one $\ds$
thus implies that the expected loss at some other sample $\ds'$ is at
\emph{least} the local minimax rate over the family of mechanisms satisfying
$\diffp(\gamma) = 4 \diffp / \log \frac{1}{2 e^2 \gamma}$-differential
privacy, which is a stronger privacy guarantee.
We note in passing that similar super-efficiency lower bounds
hold for $(\diffp, \delta)$- and \renyi-differential privacy as well,
but the calculations are tedious.

\begin{proof}
  We upper bound the local minimax risk $\LMM$. By
  Theorem~\ref{thm:lower-bound-local-minimax}
  we have
  %% \begin{equation*}
  %%   8 \cdot \LMM(\ds, |\cdot|, \mechfamily)
  %%   \ge \max_{k \le n} \modcont_\func(\ds; k) e^{-k \diffp}
  %%   \ge e^{-1} \modcont_\func(\ds; \diffp^{-1}).
  %% \end{equation*}
  %% Similarly,
  for $k_0 = \frac{2}{\diffp}$
  that
  \begin{equation*}
    \LMM(\ds', |\cdot|, \mechfamily_\diffp)
    \le \max_k e^{-k \diffp/2} \modcont_\func(\ds'; k)
    \le \max_k e^{-k \diffp/2} \modcont_\func(\ds'; k_0) e^{k/k_0}
    = \modcont_\func\left(\ds'; \frac{2}{\diffp}\right)
    \le e^2 \modcont_\func\left(\ds'; \frac{1}{\diffp}\right)
  \end{equation*}
  by the assumptions of the corollary, whenever $\ds'$ is close
  to $\ds$.
  Proposition~\ref{proposition:super-efficiency} thus implies that
  there exists $\ds'$ such that
  \begin{equation*}
    \E[|\mech(\ds') - \func(\ds')|]
    \ge \frac{1}{16} \modcont_\func\left(\ds'; \frac{-\log(-2 e^2\gamma)}{2
      \diffp}\right)
    \ge \frac{1}{16} \LMM(\ds', |\cdot|, \mechfamily_{\diffp(\gamma)}),
  \end{equation*}
  as desired.
\end{proof}

% -*- Mode: latex -*- %

\section{The inverse-sensitivity mechanism}
\label{sec:upper-bounds}

We now turn to develop the inverse-sensitivity mechanism in both discrete
(\S~\ref{sec:optimal-discrete}) and general (\S~\ref{sec:mech-continuous})
cases, showing that the mechanism (nearly) achieves the instance-dependent
lower bounds of the previous section. In
Section~\ref{sec:monotone-functions} to come, we perform a more careful
investigation for a set of natural functions $\func$ of interest, keeping
the development here general.

\subsection{Instance-optimality in the discrete case}
\label{sec:optimal-discrete}

We begin by considering
functions $\func : \domain^\dssize \to \range$ for $|\range| < \infty$.
The inverse sensitivity mechanism~\eqref{mech:discrete}
instantiates the exponential mechanism~\cite{McSherryTa07}, which we
re-display here:
\setcounter{storecounter}{\value{equation}}
\renewcommand{\theequation}{\textsc{M}.1}
\begin{equation}
  \P(\mechdisc(\ds) =  t) =
  \frac{e^{-\invmodcont_\func(\ds; t) \diffp/2}}{
    \sum_{s \in \range} e^{-\invmodcont_\func(\ds;s) \diffp/2}}.
\end{equation}
Privacy guarantees for the inverse sensitivity mechanism~\eqref{mech:discrete}
are nearly immediate
from those for the exponential mechanism.
\setcounter{equation}{\value{storecounter}}
\renewcommand{\theequation}{\arabic{equation}}
%we revisit its unbiasedness for more general settings
%in Section~\ref{sec:unbiasedness-monotone} to come.
\begin{lemma}[\citet{McSherryTa07}, Theorem 6]
  \label{lemma:lipschitz-score}
  Let $h: \domain^\dssize \times \range \to \R_+$ be $1$-Lipschitz with
  respect to the Hamming distance on $\domain$,
  $|h(\ds,t) - h(\ds',t)| \le 1$ for any $t \in \range$
  and neighboring instances $\ds,\ds' \in \domain^\dssize$.
  Let $\mu$ be any measure on $\range$. Then the
  mechanism $\mech$ with density
  \begin{equation*}
    \frac{d\pi}{d\mu}(t)
    = \frac{e^{-h(\ds, t) \diffp / 2}}{
      \int_\range e^{-h(\ds, s) \diffp / 2} d\mu(s)}
  \end{equation*}
  is $\diffp$-differentially private.
\end{lemma}
\noindent
As an immediate consequence of Lemma~\ref{lemma:lipschitz-score},
mechanism~\eqref{mech:discrete} is private
(see Appendix~\ref{sec:proof-mech-discrete-privacy}).
\begin{lemma}
  \label{lemma:mech-discrete-privacy}
  The mechanism~\eqref{mech:discrete} is $\diffp$-DP. Moreover, 
  if $\func$ is binary then it is $\diffp/2$-DP.
\end{lemma}

%% In this section, we prove several instance optimality results
%% for the inverse sensitivity mechanism in the discrete case.
%% In particular, we show it is nearly instance-optimal   
%% with respect to our two notions of optimality for 
%% the $\zo$ loss and other general losses.

An immediate bound on the \zo loss follows from
definition~\eqref{mech:discrete}:
\begin{proposition}%[Upper bound for $\zo$ loss]
  \label{proposition:upper-01-loss}
  Let $\func: \domain^\dssize \to \range$ where $|\range| < \infty$.
  Then the mechanism~\eqref{mech:discrete}
  has
  \begin{equation*}
    \P \left(\mechdisc(\ds) = \func(\ds) \right)
    = \frac{1}{\sum_{t \in \range} {e^{- \invmodcont_\func(\ds;t) \diffp /2}}}.
  \end{equation*}
\end{proposition}
\noindent
Combining Proposition~\ref{proposition:upper-01-loss}
and the lower bound of Proposition~\ref{proposition:lower-01-loss}
implies that the inverse sensitivity mechanism is 
nearly instance optimal for the $\zo$ loss.
\begin{corollary}
  \label{cor:0-1-loss}
  For the 0-1 loss $\losszo$,
  the mechanism $\mechdisc$ is $4$-optimal against
  $\losszo$-unbiased mechanisms (Definition~\ref{def:instance-optimal}) for 
  any discrete function $\func: \domain^\dssize \to \range$.
\end{corollary}

We can extend this analysis to general losses applied to distances in the
target space $\range$.  Let $(\range,\dt)$ be a finite metric space,
and let $\dt^\star = \max_{s,t \in \range} \dt(s,t)$. The following
theorem upper bounds the loss of the inverse sensitivity mechanism (we defer
proof to Appendix~\ref{sec:proof-upper-bound-general-loss}).
\begin{theorem}[Discrete functions: upper bound for general loss]
  \label{thm:upper-bound-general-loss}
  Let $\loss(s,t) = \ell(\dt(s,t))$ for a
  non-decreasing function $\ell : \R_+ \to \R_+$
  and  $\func: \domain^\dssize \to \range$.
  %Assume $\dt(s,t) \le \dt^\star$ for $s \neq t$.
  Then for any $\ds \in \domain^\dssize$ and $\gamma >0$,
  \begin{equation*}
    \E\left[\loss(\mechdisc(\ds),\func(\ds)) \right] 
    \le \ell\left(\modcont_\func\left(\ds; \frac{2}{\diffp}
    \left(\log\frac{2\ell(\dt^\star) \card(\range)}{\gamma \diffp}
    \right)\right) \right) + \gamma.
  \end{equation*}
  Let $\wt{\mech}_\disc$ be the discrete mechanism~\eqref{mech:discrete}
  except that we replace the privacy parameter
  $\diffp$ with
  $\wt{\diffp} = 2 \diffp \log \frac{2 \ell(\dt\opt) \card(\range)}{\gamma
    \diffp}$. Then
  $\E[\loss(\wt{\mech}_\disc(\ds), \func(\ds))] \le
  \ell(\modcont_\func(\ds; \frac{1}{\diffp})) + \gamma$
  for all $\ds \in \mc{X}^n$.
  %%
  %% $\diffp$ with $\wt{\diffp} = 2 \diffp \cdot (\log\frac{\card(\range)}{\diffp}
  %%       + 3)$. Then
  %%       \begin{equation*}
  %%       	\E\left[\loss(\wt{\mech}_\disc(\ds),\func(\ds)) \right] 
  %%       	\le \sum_{k=1}^{ n\diffp} e^{-3(k-1) } \ell(\modcont_\func(\ds;k/\diffp)).
  %%       \end{equation*}
\end{theorem}

Theorem~\ref{thm:upper-bound-general-loss} and the lower
bounds of Theorems~\ref{thm:lower-bound-general-loss} 
and~\ref{thm:lower-bound-local-minimax} imply 
the following corollary.
\begin{corollary}
  \label{cor:general-losses}
  Let the conditions of Theorem~\ref{thm:upper-bound-general-loss} hold, and
  let $\diffp > 0$ and $\sup_{s, t \in \range} \loss(s, t) \le
  \loss\opt$. Define
  $C = 2 \log \frac{2 \loss\opt \card(\range)}{\modcont_\func(\ds;
    \frac{1}{\diffp}) \diffp}$.
  Then the discrete mechanism $\mechdisc$~\eqref{mech:discrete}
  is both local minimax $C$-optimal (Def.~\ref{def:local-minimax-optimal})
  and $C$-optimal against unbiased
  mechanisms (Def.~\ref{def:instance-optimal}).
  %%
  %% and $\ell : \R_+ \to \R_+$ be non-decreasing
  %%       and satisfy $\ell(t / 2) \ge c \ell(t)$ for a constant $c > 0$ and all $t
  %%       \ge 0$.  Then for any discrete-valued $\func: \domain^\dssize \to \range$
  %%       and loss function $\loss(s,t) = \ell(\dt(s,t))$, the discrete mechanism
  %%       $\mechdisc$~\eqref{mech:discrete} is $(\loss,\nearoptconst=2\log (\dssize
  %%       \card(\range)) + 6)$-instance optimal
  %%       and $(\loss,\nearoptconst=2\log (\dssize
  %%       \card(\range)) + 4)$-local minimax optimal
  %%       for the family $\mechfamily_{\diffp,0}$
  %%       (to within a constant factor).
\end{corollary}

\subsection{The inverse sensitivity mechanism for arbitrary-valued functions
			and near optimality}
\label{sec:mech-continuous}

The mechanisms we have thus far developed focus on the discrete case,
so we now show how to extend the ideas to functions taking values in an
arbitrary measurable space $\range$.  For a function $\func: \domain^\dssize
\to \range$, the direct approach is to discretize the range of $\func$,
then apply our mechanism for discrete
functions. Aside from aesthetic inelegance, this mechanism may result in
unsatisfactory running time.  A natural alternative is to replace sums with
integrals in our discrete mechanism, so that
for a base measure $\basemeasure$ on $\range$,
we define $\mech(\ds)$ to have
$\basemeasure$-density
\begin{equation*}
	\pdf_{\mech(\ds)}(t) = \frac{e^{-\invmodcont_\func(\ds;t) \diffp/2}}{
		\int_{\range} e^{-\invmodcont_\func(\ds;s) \diffp/2} d \basemeasure(s)}.
\end{equation*}
Subtleties arise with this direct generalization that
do not in the
discrete case. The complication is that $\invmodcont_\func(\ds;t)=0$ only
for $t = \func(x)$, so that for stable $\func$ the
denominator may be small except at the point $\func(x)$, yielding
large probabilities for $t$ distant from $f(x)$. To circumvent
this difficulty, we sometimes work with a smoother version of
$\invmodcont_\func(\ds;t)$.

Restricting the generality of $\range$ as an arbitrary measure space a bit, let
$\range$ be a vector space with norm $\norm{\cdot}$. Then for $\rho>0$, we
define the $\rho$-smooth version of the inverse sensitivity
\begin{equation}
\label{eq:smooth-invmodcont}
\smoothinvmodcont_\func(\ds;t) = \inf_{s \in \range : \norm{s - t} \le \rho}
\invmodcont_\func(\ds;s).
\end{equation}
For $\range$-valued functions, the inverse sensitivity mechanism
$\mechcont(\ds)$ instantiates the exponential mechanism~\cite{McSherryTa07}
with $\smoothinvmodcont_\func$ via the $\basemeasure$-density
\setcounter{storecounter}{\value{equation}}
\renewcommand{\theequation}{\textsc{M}.2}
\begin{equation}
\label{mech:continuous}
\pdf_{\mechcont(\ds)}(t)
= \frac{e^{-\smoothinvmodcont_\func(\ds;t ) \diffp/2}}{
	\int_{\range} e^{-\smoothinvmodcont_\func(\ds;s) \diffp/2} d\basemeasure(s)}.
\end{equation}
The value of $\rho$ must be small to achieve satisfactory utility yet large
enough to accumulate sufficient weight in the denominator. In most
statistical applications, estimators converge at rate $1 / \sqrt{n}$, so a
typical rule-of-thumb is to take $\rho = 1 / \mathsf{poly}(n) \ll 1 / \sqrt{n}$.
Conveniently, it is easy to see that the general
mechanism $\mechcont$ is $\diffp$-DP.

% Get counters back to normal
\setcounter{equation}{\value{storecounter}}
\renewcommand{\theequation}{\arabic{equation}}

\begin{proposition}
  The mechanism $\mechcont$~\eqref{mech:continuous} is
  $\diffp$-differentially private.
\end{proposition}
\begin{proof}
  Lemma~\ref{lemma:lipschitz-score} shows that it is
  enough to prove that $\smoothinvmodcont_\func(\ds;t)$ is $1$-Lipschitz
  with respect to Hamming distance. Let $\ds,\ds'$ be
  neighboring datasets, $t \in \range$, and assume w.l.o.g.\ that $l =
  \smoothinvmodcont_\func(\ds;t) \le \smoothinvmodcont_\func(\ds';t) $.  If
  $l = \invmodcont_\func(\ds;t)$ then $\smoothinvmodcont_\func(\ds';t)
  \le \invmodcont_\func(\ds';t) \le \invmodcont_\func(\ds;t) + 1 = l + 1$
  since $\invmodcont_\func$ is $1$-Lipschitz.  Otherwise, there exists $s$
  such that $\norm{t - s} \le \rho$ and $l = \invmodcont_\func(\ds;s)$.  It
  follows that $\smoothinvmodcont_\func(\ds';t) \le
  \invmodcont_\func(\ds';s) \le \invmodcont_\func(\ds;s) + 1 = l + 1$ as
  desired.
\end{proof}

We turn to demonstrate near instance-optimality for
general losses and functions $\func$; in the next section, we provide
stronger results when $\func$ is reasonable in a sense we make precise
(for example, if $\func : \mc{X}^n \to \R$ is continuous).
The main result of this section shows that
using a logarithmically larger $\diffp$, mechanism~\eqref{mech:continuous}
achieves  optimal error to constant factors.
For simplicity, we state our results for
the case that $\basemeasure$ is the Lebesgue
measure and $\range = \R$, though our proofs generalize this a bit.
We assume the loss
$\loss(s,t) = \ell(|s - t|)$ for a
non-decreasing function $\ell : \R_+ \to \R_+$
where $\ell(u + v) \le \ell(u) + C_\ell v$ for 
all $u, v$ and some $C_\ell < \infty$.\footnote{If $\func$ is bounded,
  this need hold only over the range of $\func$.}
Under these assumptions and recalling the local modulus of
continuity~\eqref{eqn:modcont-def}, we have the following
upper bound, which we prove in
Appendix~\ref{sec:proofs-theory-continuous}.

\begin{theorem}
  \label{thm:upper-bound-general-loss-continuous}
  %Let $\loss(s,t) = \ell(\norm{s - t})$ for a
  %non-decreasing function $\ell : \R_+ \to \R_+$
  %and  
  Let $\func: \domain^\dssize \to \R$
  %where $\diam(\func(\domain^n)) \le K < \infty$
  %Assume that for $u, v \in [0, K]$,
  %$\ell(u + v) \le \ell(u) + C_\ell v$ for some $C_\ell < \infty$
  and assume that the uniformity 
  condition~\eqref{eqn:basemeasure-uniformity} holds.
  Then for any $\ds \in \domain^\dssize$, %and $\gamma >0$,
  \begin{equation*}
    \E\left[\loss(\mechcont(\ds),\func(\ds)) \right] 
    %		\le \ell\left(\modcont_\func\left(\ds;
    %		\frac{2}{\diffp}
    %		\left[\log\frac{n \ell(K)(1 + \diffp)}{\gamma}
    %		+ \log \Big(\frac{8 \dssize \GS_\func}{\rho}\Big)
    \le \ell\left(\modcont_\func\left(\ds;
    \frac{2}{\diffp}
    \left[\log\frac{1}{\diffp}
      + 2\log \frac{\dssize \GS_\func}{\rho}
      \right]\right)\right)
    %+ O(1) \gamma 
    + O(1) C_\ell \rho.
  \end{equation*}
  Let $\wt{\mech}_\cont$ be the mechanism~\eqref{mech:continuous}
  with privacy parameter
  $\wt{\diffp} = 2 \diffp [\log \frac{1}{2 \diffp}
    + 2 \log \frac{n \GS_\func}{\rho}]$. Then
  $\E[\loss(\wt{\mech}_\cont(\ds), \func(\ds))]
  \le \modcont_\func(\ds; \frac{1}{\diffp})$.
  %% \begin{equation*}
  %%   \E\left[\loss(\wt{\mech}_\cont(\ds),\func(\ds)) \right] \le
  %%   \sum_{i=1}^{n\diffp} e^{-3(i-1)} \ell(\modcont_\func(\ds;i/\diffp))
  %%   + C_\ell \rho.
  %% \end{equation*}
\end{theorem}

Theorem~\ref{thm:upper-bound-general-loss-continuous} and the
lower bounds of Theorems~\ref{thm:lower-bound-general-loss} 
and~\ref{thm:lower-bound-local-minimax}, with their evident
dependence on the modulus $\modcont_\func$,
demonstrate the near instance optimality of
mechanism~\eqref{mech:continuous} whenever $\rho$ is small.
We summarize this
with the following parallel to
Corollary~\ref{cor:general-losses}.
%Assume $\func : \domain^n \to \R$ have range
%$\diam(\func(\domain^n)) \le K$ and we have loss 
%$\loss(s, t) = \ell(|s - t|)$ where $\ell : \R_+ \to \R_+$
%is non-decreasing and satisfies $\ell(t + v) \le \ell(t) + C_\ell v$ for a
%constant $C_\ell < \infty$ and all $0 \le t, v \le K$.
\begin{corollary}
  \label{cor:general-loss-continuous}
  Let $\func: \domain^\dssize \to \R$, $\diffp > 0$, $p \in \N$
  and $\rho = \dssize^{-p}$.
  Then the mechanism~\eqref{mech:continuous}
  is $C = 2 (\log\frac{1}{2 \diffp}
  + 2 \log (n^{1 + p} \GS_\func))$-optimal to within an additive error
  $n^{-p}$ both
  in a local minimax sense (Def.~\ref{def:local-minimax-optimal})
  and against unbiased mechanisms (Def.~\ref{def:instance-optimal}).
\end{corollary}
\noindent
Of course, $C = O(\log n)$ optimality may be unsatisfying: the privacy
parameter $\diffp = \log n$ generally provides limited privacy protections,
and collecting a dataset $\wt{\ds}$ of size $n \log n$ and downsampling, as
we discuss following Def.~\ref{def:local-minimax-optimal}, may be
infeasible. In some realms of theoretical analysis an increase
in sample complexity by a factor $\log n$ is benign, but---as we
do in the coming section---we will aim for better.

% -*- Mode: latex -*- %

\section{Instance-optimality for sample-monotone functions}
\label{sec:monotone-functions}

While the previous sections prove that the inverse sensitivity mechanism
with privacy parameter $\tilde \diffp = O(\log \dssize) \cdot \diffp$ is
instance-optimal against $\diffp$-DP (or $(\diffp,\delta)$-DP or
\renyi-DP) mechanisms, it is important to understand when $O(1)$-optimality
is possible.  To that end, here we prove stronger
results for what we call \emph{sample-monotone} functions, a natural class
including continuous functions over convex domains.  In
Section~\ref{sec:avg-instance-optimal} we show that for ``most'' instances
and well-behaved functions, the inverse sensitivity mechanism is
$O(1)$-optimal.  We also show how allowing a degradation in the
local minimax rate $\LMM(\ds)$ allows near-optimality: the inverse-sensitivity
mechanisms~\eqref{mech:discrete} and~\eqref{mech:continuous} with privacy
parameter $\diffp (\log\frac{1}{\diffp} + \log\log n)$ achieve, for any
$\tau > 0$, expected error $\dssize^\tau \LMM(\ds)$ (see
Section~\ref{sec:near-instance-optimal}). We conclude the section with
comparisons in Section~\ref{sec:comparisons}, where we show that the inverse
sensitivity mechanism \emph{uniformly} outperforms the Laplace and smooth
Laplace mechanisms.

Our function class
consists of functions for which changing the function value
$\func(\ds)$ significantly requires changing more
elements of $\ds$ than does changing $\func(\ds)$ only slightly.
\begin{definition}
  \label{def:monotone}
  Let $\func: \domain^\dssize \to \R$. Then $\func$ is
  \emph{sample-monotone} if for every $\ds \in \domain^\dssize$, and $s, t
  \in \R$ satisfying $\func(\ds) \le s \le t$ or $t \le s \le \func(\ds)$,
  we have $\invmodcont_\func(\ds;s) \le \invmodcont_\func(\ds;t)$.
\end{definition}
\noindent
Many estimands, including the sample mean and median, are
sample-monotone. More, any continuous function over a convex domain is
sample monotone; see Appendix~\ref{sec:proof-cont-is-montonote} for a proof.
\begin{observation}
  \label{observation:cont-is-montonote}
  Let $\func: \domain^\dssize \to \R$ be continuous and $\domain$ be
  convex. Then $\func$ is sample-monotone.
\end{observation}

As one consequence of Observation~\ref{observation:cont-is-montonote},
partial minimization of regularized convex
losses is sample monotone. For example, if we wish to find the
multiplier $\theta \in \reals$ on a feature $j$ in a linear model,
then we generically wish to estimate
\begin{equation}
  \theta(\ds) = \argmin_{\theta \in \R}
  \inf_{\beta \in \R^d} \left\lbrace 
  L(\theta,\beta;x)  
  \defeq \sum_{i=1}^n \ell(\theta,\beta;x_i)
  + \frac{\lambda}{2} \ltwo{[\theta, \beta]}^2 \right\rbrace.
  \label{eqn:partial-minimization}
\end{equation}
for some loss $\ell(\cdot; \cdot) : \R \times \R^{d} \times \mc{X} \to \R$,
which is continuous in each of its arguments.
The continuity of the minimizer $\func(\ds)$ is immediate from
standard stability guarantees~\cite{BousquetEl02}.
%% If $\ell$ convex and continuous in $x_i$, we get that 
%% $f$ is continuous in $x$ 
%% (briefly, $L(\theta,\beta;x)$ has a unique minimizer 
%% therefore if $\norm{x - x'} \to 0$ then 
%% $L(\theta,\beta;x) - L(\theta,\beta;x') \to 0$ which implies
%% that $\norm{[\theta^\star(x) - \theta^\star(x'), \beta^\star(x) - \beta^\star(x')]} \to 0$ as $\ell$ is continuous) 
%% and therefore 
%% proposition~\ref{prop:cont-is-montonote} implies that 
%% $f$ is monotone.

\subsection{Typical instance optimality}
\label{sec:avg-instance-optimal}

We focus here on optimality results that hold for well-behaved
instances $\ds$. Key to our upper bounds is the
expected modulus of continuity of $\func$ for
instance $\ds$ at a random distance,
\begin{equation*}
  \label{eq:expected-modulus}
  \expmod_\func(\ds;\diffp) 
  \defeq \E_{K \sim \mathsf{Geo}(1 - e^{-\diffp}) }
  \left[\modcont_\func(\ds; K) \indic{K \le n}\right],
\end{equation*}
where $K \sim \mathsf{Geo}(\lambda)$ denotes a geometric 
random variable, $\P(K = i) = (1-\lambda)^i \lambda$, $i \in \N$.
Our optimality results depend on the sensitivity
of the modulus of continuity,
motivating the ratio
\begin{equation}
  \label{eq:ratio-expected-modulus}
  \ratiomod_\func(\ds) \defeq \frac{\expmod_\func(\ds;\diffp/4)}{\expmod_\func(\ds;\diffp/2)}. 
\end{equation}
%The ratio $\ratiomod_\func(\ds)$ aims to measure the sensitivity
%of the expected modulus of continuity near an instance $\ds$
%and we prove that this quantity is small for well behaved 
%functions and instances.
%
The main result of this section shows that the inverse sensitivity mechanism
is $O(1)$-optimal for $\ds$ whenever
\begin{equation*}
  \ratiomod_\func(\ds) \lesssim 1,
\end{equation*}
and we have the following theorem, whose proof we defer to
Appendix~\ref{sec:proof-thm-monotone-loss-optimal}.
\begin{theorem}
  \label{thm:loss-monotone}
  Let $\func : \domain^\dssize \to \R $ be sample-monotone, $p \in \N$,
  and $\rho > 0$.
  The mechanism $\mechcont$~\eqref{mech:continuous} satisfies
  \begin{equation*}
    \E\left[|\mech_\cont(\ds) - \func(\ds)|^p \right]
    \le 2^{p+1} \ratiomod_\func(\ds)  
    \max_{1 \le k \le n} \modcont_\func(\ds;k)^p e^{-k \diffp/4} 
    + \gamma,
  \end{equation*}
  where $\gamma = 2^p \rho + 
  \frac{(2\modcont_\func(\ds;n))^{p+1} e^{-(n+1) \diffp/2} }{\rho 
    + \sum_{k=1}^{n} \modcont_\func(\ds;k) e^{-k \diffp/2}}$.
\end{theorem}
%\hacomment{maybe set $\rho=0$ here?}

Recalling Theorems~\ref{thm:lower-bound-general-loss}
and~\ref{thm:lower-bound-local-minimax}, Theorem~\ref{thm:loss-monotone}
implies that the inverse sensitivity mechanism is $O(1)$-optimal to within
small additive factors---in both local-minimax
(Def.~\ref{def:local-minimax-optimal}) and unbiased
(Def.~\ref{def:instance-optimal}) senses---whenever the modulus is smooth
enough that $\ratiomod_\func(\ds) \lesssim 1$.
One consequence of Theorem~\ref{thm:loss-monotone} is that as long as
$\func : \mc{X}^n \to \R$ is sample monotone and grows
at most exponentially in the sample distance $k$, then the inverse
sensitivity mechanism~\eqref{mech:continuous} is $O(1)$-optimal. We
prove the
following corollary
in Appendix~\ref{proof:optimal-exponential-growth}.
\begin{corollary}
  \label{corollary:optimal-exponential-growth}
  Let $\func: \domain^\dssize \to \R $ and $\diffp \lesssim 1$.
  If $\frac{\modcont_\func(\ds;k)}{\modcont_\func(\ds;t)} \lesssim e^{C k/t}$
  for $\frac{1}{\diffp} \le t \le k$
  and $C < \infty$ then $\ratiomod_\func(\ds) \lesssim 1$.
  In particular, $\E[|\mech_\cont(\ds) - \func(\ds)| ] \lesssim
  \modcont_\func(x;\frac{4C}{\diffp}) + e^{-\dssize \diffp /4}$.
\end{corollary}

Additionally, combined with bounds that $\ratiomod_\func(\ds)
\lesssim 1$ for well-behaved functions $\func$ and instances $\ds$,
Theorem~\ref{thm:loss-monotone} demonstrates good behavior of the inverse
sensitivity mechanism~\eqref{mech:continuous}.
For example, when
$f$ is the median and $x_i \simiid P$ for a distribution $P$ with density
near its median, $\ratiomod_\func(\ds) \lesssim 1$ with high probability
(the proof of this is tedious but similar to our derivations for
expected loss of the median in Sec.~\ref{sec:median}).  A
more sophisticated  example is
partial minimization~\eqref{eqn:partial-minimization} of smooth convex
losses.  We consider a setting with losses $\ell(\cdot; x_i) : \R \times
\R^d \to \R$, and we wish to estimate the parameter $\theta(\ds)$ as in
Eq.~\eqref{eqn:partial-minimization} or a similar problem. For shorthand, we
let $\tau = (\theta, \beta) \in \range \subset \R^{1 + d}$, and we consider
the following conditions on the loss $\ell$. For notational simplicity, we
let $\dot{\ell}$ and $\ddot{\ell}$ denote the gradient and Hessian of
$\ell$.
\begin{assumption}
  \label{assumption:smoothness-losses}
  The losses $\ell : \R \times \R^d \times \mc{X} \to \R$ have
  $\lipgrad$-Lipschitz continuous gradient and $\liphess$-Lipschitz
  continuous Hessian. Additionally, for each $\tau \in \range \subset \R
  \times \R^d$, the image of their derivatives satisfies
  \begin{equation*}
    \lipobjin \ball \subset
    \{\dot{\ell}(\tau; x)\}_{x \in \mc{X}}
    \subset \lipobj \ball.
  \end{equation*}
\end{assumption}

Assumption~\ref{assumption:smoothness-losses} imposes conditions on the
losses $\ell$ and data $\ds$ that
are often natural. The final condition on
$\dot{\ell}(\tau; \ds)$ simply means that $\ell$ is $\lipobj$-Lipschitz with
respect to the $\ell_2$-norm and that there exist examples that can move the
loss in many directions. For example, for robust regression problems,
satisfying Assumption~\ref{assumption:smoothness-losses}
is nearly immediate.

\begin{example}[Robust regression]
  \label{example:robust-regression-smoothness}
  Assume the data are of the form $(x, y) \in \R^d \times \R$, where
  $\ltwo{x} \le Bx$. We let $h : \R \to \R_+$ be convex, symmetric, and
  $1$-Lipschitz, satisfying $\sup_{t \in \R} |h'(t)| = 1$, and
  consider the robust regression losses
  \begin{equation*}
    \ell(\tau; x, y) = h(\<\tau, x\> - y);
  \end{equation*}
  this includes the ``$\alpha$-insensitive''
  losses~\cite{DekelShSi03} with $h_\alpha(t) = \alpha \log(1 +
  e^{t/\alpha}) + \alpha \log(1 + e^{-t/\alpha})$ or
  the Huber loss~\cite{HuberRo09}. Consider the
  $\alpha$-insensitive loss for concreteness. In this case, when the
  variables $y \in \R$ may be arbitrary (natural for robust regression) and
  $\ltwo{x} \le \Bx$, we have the equality
  \begin{equation*}
    \{\dot{\ell}(\tau; x, y)\}_{\ltwo{x} \le r, y \in \R}
    = \Bx \ball_2^{d+1},
  \end{equation*}
  so that $\lipobjin = \lipobj$ in
  Assumption~\ref{assumption:smoothness-losses}. Defining
  $p(\tau; x, y) = 1 / (1 + \exp((y - \<\tau, x\>) / \alpha))$,
  the losses satisfy
  $\ddot{\ell}(\tau; x, y)
  = \alpha^{-1} p(\tau; x, y)(1 - p(\tau; x,y)) xx^T
  \preceq (4 \alpha)^{-1} xx^T$, so that
  $\ell$ has Lipschitz gradient and Hessian with
  $\lipgrad \le (4 \alpha)^{-1} r^2$ and
  $\liphess \le r^3 / \alpha^2$.
\end{example}

We then have the following example, which shows that under appropriate
conditions on the empirical loss $\riskn(\tau) = n^{-1} \sum_{i = 1}^n
\ell(\tau; x_i)$ we can almost exactly compute the modulus of continuity
of the partial minimizer $\theta(\ds)$.  This is not surprising, as
standard statistical regularity conditions show the minimizer
$\theta(\ds)$ should be (asymptotically) linear in the
sample~\cite{VanDerVaart98}.

\begin{example}[Smoothness of partial minimization]
  \label{example:smoothness-partial-min}
  Let $\range \subset \R^{1 + d}$ be a closed convex set, and
  for losses $\ell$ satisfying Assumption~\ref{assumption:smoothness-losses},  
  consider the partial minimization problem
  \begin{equation*}
    \theta(\ds) \defeq \argmin_\theta
    \inf_{\beta : (\theta, \beta) \in \range}
    \left\{\frac{1}{n} \sum_{i = 1}^n \ell(\theta, \beta; x_i)
    \right\}.
  \end{equation*}
  Fix a number $k \in \N$ of examples to change.
  For shorthand,
  let $\tau = \tau(\ds)
  = \argmin_{\tau \in \range} \riskn(\tau)$, and let us assume both
  the local strong convexity condition
  that $\ddot{\risk}_n(\tau) \succeq \lambda I$ and
  that $\tau$ is interior to $\range$, satisfying
  $\dist(\tau, \boundary \range) \ge \frac{6 \lipobj}{\lambda} \frac{k}{n}$.
  Then we claim both that
  \begin{equation}
    \modcont_\theta(\ds; k)
    = \left(1 \pm \frac{24 {\lipobj}^2 \lipgrad \liphess}{\lipobjin\lambda^3}
    \frac{k}{n}\right)
    \cdot
    \sup_{\dham(\ds', \ds) \le k}
    \left|\left[
      \frac{1}{n} \sum_{i = 1}^n\ddot{\risk}_n(\tau)^{-1}
      (\dot{\ell}(\tau; x_i) - \dot{\ell}(\tau; x_i'))
      \right]_1 \right|
    \label{eqn:modulus-partial-min}
  \end{equation}
  and for the matrix
  $A = \ddot{\risk}_n(\tau)^{-1}$ having first column $a \in \R^{1 + d}$,
  we also have
  \begin{equation}
    \label{eqn:compute-modulus-partial-sum}
    \frac{k}{n} \ltwo{a} \cdot \lipobjin
    \le
    \sup_{\dham(\ds', \ds) \le k}
    \left|\left[
      \frac{1}{n} \sum_{i = 1}^n \ddot{\risk}_n(\tau)^{-1}
      (\dot{\ell}(\tau; x_i) - \dot{\ell}(\tau; x_i'))\right]_1
    \right|
    \le \frac{2 k}{n} \ltwo{a} \cdot \lipobj.
  \end{equation}
  We provide proofs of both inequalities in
  Appendix~\ref{sec:partial-minimization}.

  The essential observation here is that so long as the privacy parameter
  $\diffp$ and empirical minimizer $\tau = \tau(\ds)$ satisfy $\dist(\tau,
  \boundary \range) \gg \frac{\lipobj}{\lambda \diffp n}$, then it is
  immediate the ratio~\eqref{eq:ratio-expected-modulus} satisfies
  $\ratiomod_\theta(\ds) \lesssim \frac{\lipobj}{\lipobjin}$.
  Theorem~\ref{thm:loss-monotone} then implies that whenever the losses are
  regular enough that $\lipobj / \lipobjin \lesssim 1$, for example, in the
  case of robust regression
  (Ex.~\ref{example:robust-regression-smoothness}), the continuous inverse
  sensitivity mechanism~\eqref{mech:continuous} with smoothing parameter
  $\rho = n^{-2}$ attains expected loss
  \begin{equation*}
    \E[|\mechcont(\ds) - \theta(\ds)|^p]
    \le O(1) \cdot \LMM(\ds, |\cdot|^p, \diffp)
    + O(n^{-2 p}).
  \end{equation*}
  That is, it is $O(1)$-optimal.
\end{example}

We discuss two points in Example~\ref{example:smoothness-partial-min}. The
assumption that the empirical loss $\riskn$ satisfies the local strong
convexity condition $\ddot{\risk}_n(\tau(\ds)) \succeq \lambda I$ is natural
assuming the data $x_i \simiid P$ from some distribution $P$, and that the
population minimizer $\tau\opt = \argmin_\tau \{\risk_\infty(\tau) \defeq
\E[\loss(\tau; x_i)]\}$ satisfies both $\tau\opt \in \interior \range$ and
$\ddot{\risk}_\infty(\tau\opt) \succeq 2 \lambda I$. In this case, it is
standard that the conditions of the example occur with high
probability~\cite[Ch.~5.8]{VanDerVaart98}.
The combination of Examples~\ref{example:robust-regression-smoothness}
and~\ref{example:smoothness-partial-min} show that in robust regression,
with high probability over the drawn sample, the inverse sensitivity
mechanism~\eqref{mech:continuous} is $O(1)$-optimal.

\subsection{Near and conditional instance optimality}
\label{sec:near-instance-optimal}

While the optimality results in the previous section hold for well-behaved
monotone functions, in this section we provide conditional---when a good
unbiased private estimator exists, the inverse sensitivity
mechanism~\eqref{mech:continuous} is good---and near-optimality results for
sample monotone functions.  In the first result, we provide an asymptotic
guarantee that the inverse sensitivity mechanism~\eqref{mech:continuous} is
nearly $O(\log \frac{1}{\diffp} + \log \log n)$-optimal, to within a
sub-polynomial multiplicative factor on the expected loss.
\begin{proposition}
  \label{proposition:near-instance-opt-monotone}
  Let $\func: \domain^\dssize \to \R$ be sample-monotone and assume
  $\GS_\func \le n$, $n^{-1} \le \diffp \lesssim 1 $, $p > 0$.  Let $K =
  \frac{8(p+2)\log{n}}{\diffp}$, $\lambda = \frac{1}{\log{K}} $. The
  mechanism~\eqref{mech:continuous} with $\rho = n^{-p}$ satisfies
  \begin{align*}
    \E\left[ \left| \mech_\cont(\ds) - \func(\ds) \right| \right] 
    & \le 2 e \exp{\left(  \frac{4 \log \dssize}{\log{\frac{2}{\diffp}} 
	+ \log \log \dssize} \right)} \max_{1 \le k \le K} e^{-\lambda k \diffp/2 } \modcont_\func(x;k) + O(n^{-p}). 
  \end{align*}
    %& \le 
    %\exp{ \left( \frac{4 \log{\frac{\GS_\func}{\rho \diffp}}}{ \log{\frac{1}{\diffp}} + \log\log{\frac{\GS_\func}{\rho \diffp}} } \right) } \max_{1 \le k \le K_0} e^{-\lambda k \diffp } \modcont_k + O(\rho) \\
%	In particular, let $\wt{\mech}_\cont$ be the mechanism~\eqref{mech:continuous}
%	with privacy parameter
%	$\wt{\diffp} = 2 \diffp (\log{\frac{8(2+p)}{\diffp}} + \log \log \dssize)$. 
%	Then for any $\tau >0$, 
%	\begin{align*}
%		\E\left[ \left| \wt{\mech}_\cont(\ds) - \func(\ds) \right| \right] 
%		& \le O(n^\tau) \max_{1 \le k \le K} e^{-k \tilde \diffp } \modcont_\func(x;k) 
%			  + O(n^{-p}). 
%	\end{align*}
\end{proposition}
\noindent
See Appendix~\ref{sec:proof-near-instance-opt-monotone} for a proof.

Comparing Proposition~\ref{proposition:near-instance-opt-monotone} to our
upper bound for general functions in
Theorem~\ref{thm:upper-bound-general-loss-continuous}, it is clear that if
the privacy parameter $\diffp$ is small enough that $\diffp = n^{-\beta}$
for some $\beta>0$, both upper bounds achieve a similar result: the inverse
sensitivity mechanism is $O(\log \frac{1}{\diffp}$-optimal (in either
defnition~\ref{def:local-minimax-optimal} or~\ref{def:instance-optimal} of
optimality).  Proposition~\ref{proposition:near-instance-opt-monotone}
provides a different bound in the important regime where $\diffp =
\Theta(1)$. Indeed, letting the mechanism $\wt{\mech}_\cont$ be
\eqref{mech:continuous} with privacy parameter $\wt{\diffp} = O(\log \log n)
\cdot \diffp$, Proposition~\ref{proposition:near-instance-opt-monotone}
implies that for any $\tau >0$,
\begin{align*}
  \E\left[ \left| \wt{\mech}_\cont(\ds) - \func(\ds) \right| \right] 
  & \le O(n^\tau) \max_{1 \le k \le n} e^{-k \diffp} \modcont_\func(x;k) 
  + O(n^{-p}). 
\end{align*}
In typical applications with scaling
$\modcont_\func(x; k) = O(k/n)$ (e.g., of the sample mean) or another
polynomial, we have
local minimax rate $\LMM(\ds) \lesssim \modcont_\func(\ds; 1/\diffp)
= O((n \diffp)^{-1})$, so that $\wt{\mech}_\cont$ achieves nearly
optimal scaling.

We conclude this section with a conditional optimality
result for the inverse sensitivity mechanism.
As  Theorem~\ref{thm:lower-bound-local-minimax} makes clear,
the essential quantity
in lower bounds on the local minimax risk is
the local modulus
$\modcont_\func(x; k\subopt)$
at sample radius $k\subopt = \Theta(\diffp^{-1})$.
The next result shows
if any $\ell_1$-unbiased and $\diffp/8$-differentially private mechanism
achieves this convergence guarantee for the absolute error,
then the inverse
sensitivity mechanism~\eqref{mech:continuous} is also $O(1)$-optimal.
We provide the proof in Appendix~\ref{proof:conditional-optimality}.
\begin{proposition}
  \label{proposition:conditional-optimality}
  Let $\func: \domain^\dssize \to \R$ be sample monotone and $\diffp =
  O(1)$.  If there exists an $\ell_1$-unbiased $\diffp/16$-differentially
  private mechanism $\wt{\mech}$ such that $\E[|\wt{\mech}(\ds) - \func(\ds)
    | ] \lesssim \modcont_\func(x;\tilde k) e^{-\tilde k \diffp/16}$ for
  some $\tilde{k} \in \N$, then $\ratiomod_\func(\ds) \lesssim e^{7\tilde k
    \diffp/16}$.
\end{proposition}
\noindent
Thus, whenever the postulated mechanism achieves
$\E[|\wt{\mech}(\ds) - \func(\ds) | ] \lesssim \modcont_\func(x;
O(\frac{1}{\diffp}))$, we may take
$\tilde{k} = O(1/\diffp)$, and the ratio~\eqref{eq:ratio-expected-modulus}
satisfies $\ratiomod_\func(\ds) \lesssim 1$. Then
Theorem~\ref{thm:loss-monotone} implies that $\E[|\mechcont(\ds) -
  \func(\ds)| ] \lesssim \modcont_\func(x;O(\frac{1}{\diffp})) + e^{-3 n
  \diffp/8}$ as $\gamma \lesssim e^{-7 n \diffp/16}$. In particular,
$\mechcont$ is $O(1)$-optimal \emph{whenever} an $O(1)$-optimal unbiased
mechanism exists.  Similar results to
Proposition~\ref{proposition:conditional-optimality} hold for
alternative losses and unbiased mechanisms.

\subsection{Comparisons and expected loss}
\label{sec:comparisons}

While the previous sections show strong adaptivity guarantees of the inverse
sensitivity mechanism to the underlying data $\sample \in \domain^n$ and
function $\func$, it not immediately clear how these guarantees differ from
more standard and classical mechanisms.  It is thus instructive to compare
the bounds to those available for the Laplace~\eqref{mech:laplace} and
smoothed Laplace mechanisms~\eqref{mech:smooth-laplace}. Focusing on
sample-monotone functions (Def.~\ref{def:monotone}), we show that the
inverse sensitivity mechanism uniformly outperforms the Laplace and smooth
Laplace mechanisms.

We begin with a stylized example and consider quantizing the
average of a binary query $q : \mc{X} \to \{0, 1\}$ into
 steps of width $T > 0$, where $T\diffp \gg 1$, defining
$\func_{\mathrm{step}}(\ds) = \floor{{\sum_{i=1}^{\dssize} q(x_i)}/{T} }$.
For instances $\ds$ near discontinuities in $\func_{\textup{step}}$, such as
$\ds$ with $\sum_{i=1}^{\dssize} q(x_i) = T$, the smooth Laplace
mechanism~\eqref{mech:smooth-laplace} adds Laplace noise with variance at
least ${1}/{\diffp^2}$, as $\LS(\ds) \ge 1$. In contrast, as $T \diffp$ is
large, the inverse sensitivty mechanism~\eqref{mech:discrete} returns the
value $\func_{\mathrm{step}}(\ds)$ or $\func_{\mathrm{step}}(\ds) - 1$ with
probability near $1$.  Fig.~\ref{fig:heat-map} illustrates this by
plotting 90\% confidence sets for each mechanism against the value
$\sum_{i=1}^n q(x_i)$. The Laplace and smooth-Laplace mechanisms have much
larger confidence sets (regions of plausible release) than
mechanism~\eqref{mech:discrete}.

\begin{figure}[ht]
  %% \vspace{-.6cm}
  \vspace{-0.2cm}
  \begin{center}
    \begin{overpic}[width=.7\columnwidth]{%,grid]{%
	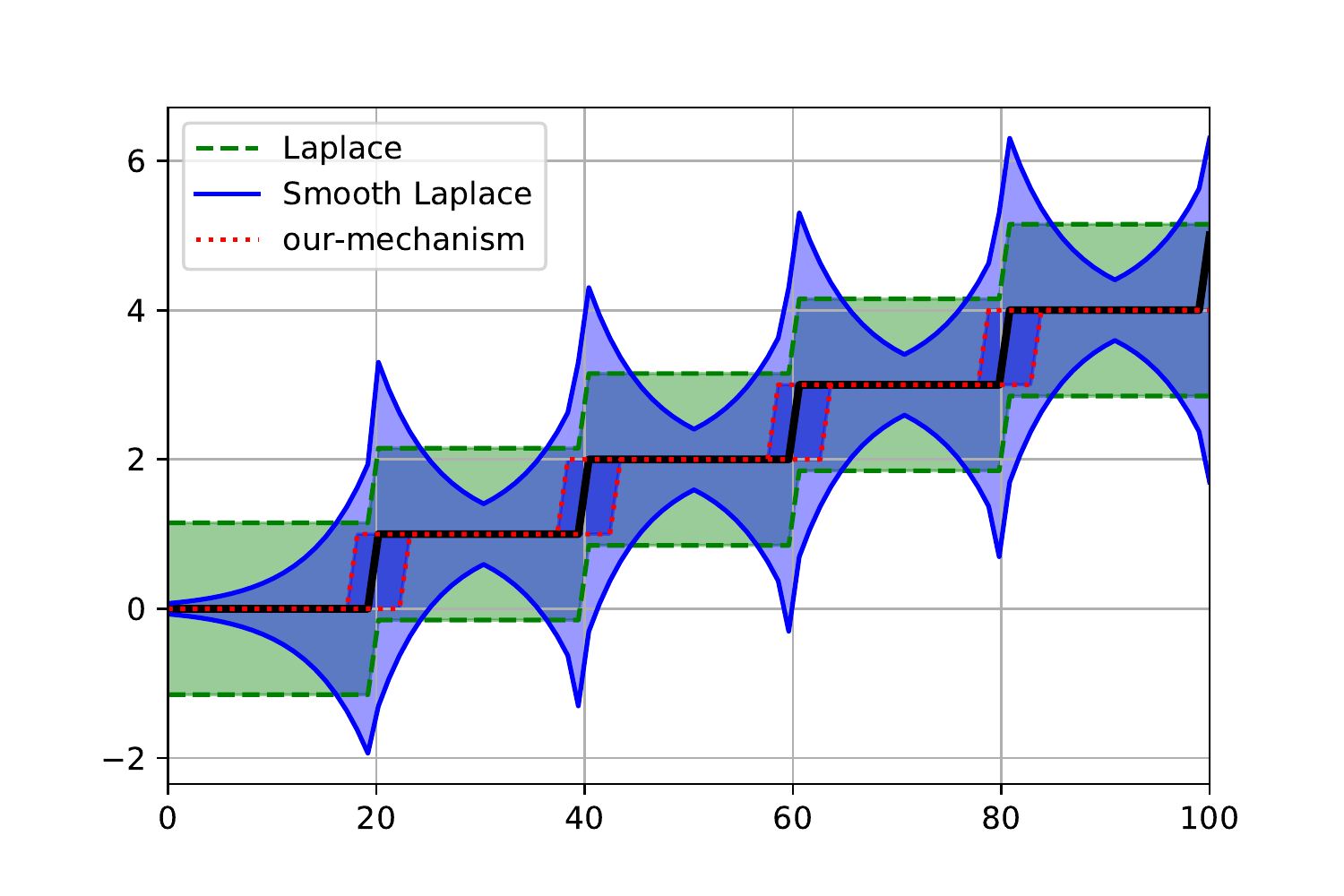}
      \put(43,0.5){{\small $\sum_{i=1}^{\dssize} q(x_i) $}}
      \put(3,28){
	\rotatebox{90}{{\small $\func_{\mathrm{step}}(\ds)$}}}
      \put(19,48){
	\tikz{\path[draw=white,fill=white] (0, 0) rectangle (2.2cm,1cm);}
      }
      \put(20,55){{\scriptsize Laplace}}
      \put(20,51.5){{\scriptsize Smooth Laplace}}
      \put(19,48){{\scriptsize Inverse sensitivity}}
    \end{overpic}
    %\put(32,72){
    %	\tikz{\path[draw=white,fill=white] (0, 0) rectangle (3cm,.35cm);}
    %}
    %      \end{overpic} &
    %      \hspace{-.3cm}
    %      \begin{overpic}[width=.4\columnwidth]{%,grid]{%
    %	  plots/Smooth-Laplace}
    %	\put(43,0.5){{\small $\sum_{i=1}^{\dssize} q(x_i) $}}
    %	\put(0,28){
    %	  \rotatebox{90}{{\small $\func_{\mathrm{step}}(\ds)$}}}
    %	%\put(32,72){
    %	%	\tikz{\path[draw=white,fill=white] (0, 0) rectangle (3cm,.35cm);}
    %	%}
    %      \end{overpic} \\
    %      %(a) & (b)
    %    \end{tabular}
    %\vspace{-.4cm}
    \caption{\label{fig:heat-map} The function $\func_{\textup{step}}$ as a
      function of $\sum_{i=1}^{\dssize} q(x_i) $. Consecutive integers on
      the horizontal axis represent neighboring datasets.
      For each mechanism $\mech$ and instance $\ds$, the plot shows
      the smallest
      interval $I = [a,b]$ satisfying $\P(\mech(\ds) \in I) \ge 0.9$.  }
  \end{center}
  \vspace{-.6cm}
\end{figure}
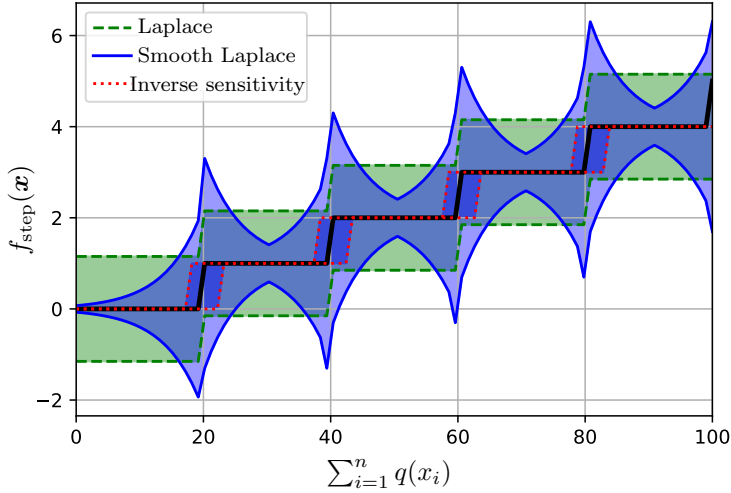

We proceed to a theoretical comparison of the three mechanisms.  Focusing
for simplicity on the absolute loss for a function $\func : \domain^n \to
\R$, we immediately obtain the lower bounds
\begin{equation}
  \label{eq:laplace-lower-bound}
  \begin{split}
    \E \left[ |\mechlap(\ds) - \func(\ds)| \right] 
    & = \frac{\GS_\func}{\diffp }, \\
    \E \left[ |\mechsmlap(\ds) - \func(\ds)| \right] 
    & \ge \frac{2}{e \cdot \diffp}
    \max_{\ds'} \left\{\LS(\ds')
    : \DSdist(\ds,\ds') \le \frac{2}{\diffp} \log \frac{2}{\delta}
    \right\}.
  \end{split}
\end{equation}
Theorem~\ref{thm:lower-bound-general-loss}
states that any
$\ell_1$-unbiased mechanism $\mech$ satisfies
$\E[|\mech(\ds) - \func(\ds)|]  
\ge \Omega(1) \cdot \modcont_\func(\ds;1/\diffp)$, which may be smaller,
and the definition~\eqref{eqn:modcont-def} of the modulus
of continuity gives $\modcont_f(\ds; k) \le k \max_{\sample' :
  \dham(\sample, \sample') \le k} \LS(\sample')$. Thus, applying either of
Theorems~\ref{thm:upper-bound-general-loss}
or~\ref{thm:upper-bound-general-loss-continuous}, we
obtain
\begin{corollary}
  \label{corollary:general-risk-bound}
  Assume that $\func : \domain^n \to \R$
  and $\diam(\func(\domain^n)) \le \mathsf{poly}(n)$. Then for any
  numerical constant $b < \infty$ there is a numerical constant $C = C(b) <
  \infty$ such that the following holds. Let $\mech$ be either of
  $\mechdisc$ or $\mechcont$ with $\basemeasure$ the Lebesgue measure.
  %and $p < \infty$.  
  Then
  \begin{align*}
    \E[|\mech(\sample) - \func(\sample)|]
    & \le
    %\left(\modcont_\func\left(\sample, \frac{C \log n}{\diffp}\right)
    %+ n^{-b}\right)^p \\
    \modcont_\func\left(\sample, \frac{C \log n}{\diffp}\right)
    + n^{-b} \\
    %		& = \mc{O}\left(
    %		\left(\frac{\log n}{\diffp}\right)^p
    %		\max_{\sample'}
    %		\left\{\LS_\func(\sample')
    %		\mid \dham(\sample, \sample') \le \frac{C \log n}{\diffp}
    %		\right\}^p\right)
    %		+ \mc{O}(n^{-b p}).
    & = \mc{O}\left(
    \frac{\log n}{\diffp}
    \max_{\sample'}
    \left\{\LS_\func(\sample')
    \mid \dham(\sample, \sample') \le \frac{C \log n}{\diffp}
    \right\}\right)
    + \mc{O}(n^{-b}).
  \end{align*}
\end{corollary}
\noindent
Of course, we always have $\LS_\func \le \GS_\func$, and so
Corollary~\ref{corollary:general-risk-bound} guarantees that
the inverse sensitivity mechanisms~\eqref{mech:discrete}
and~\eqref{mech:continuous} can never have expected loss
more than a factor of $\log n$ larger than either
the Laplace or smooth-Laplace mechanisms.

%For appropriately monotone functions $\func$, we can show a stronger result
%that the inverse sensitivity mechanism always achieves loss no
%worse than the (smoothed) Laplace mechanism.
%\begin{definition}[Monotone functions]
%	\label{def:monotone}
%	Let $\func: \domain^\dssize \to \R$. Then $\func$
%	is \emph{monotone} if for every
%	$\ds \in \domain^\dssize$, if $s, t \in \R$ satisfy $\func(\ds) \le s \le t$
%	or $t \le s \le \func(\ds)$, then $\invmodcont_\func(\ds;s) \le
%	\invmodcont_\func(\ds;t)$.
%\end{definition}
%\noindent
%The monotonicity definition~\ref{def:monotone} holds for many reasonable
%functions, including the mean and median of a sample $\sample$.

We now show that variants of the continuous
mechanism~\eqref{mech:continuous} (as we change the base measure
$\basemeasure$) have strong accuracy guarantees for 
sample-monotone functions (Definition~\ref{def:monotone});
taking $\basemeasure$ to be discrete gives the discrete
mechanism~\eqref{mech:discrete} as a special case. We elaborate the setting
somewhat, and assume that the range $\range \subset \R$ is \emph{uniformly
	spaced}, meaning that $\range$ is either $\R$ or of the form $\range = \{k
\beta \mid k \in \Z\}$, where $\beta \in \R$ is a fixed value.  With this
condition, we have the following proposition, whose proof we defer to
Appendix~\ref{sec:proof-prop-worse-case-upper-bound}.

\begin{proposition}
  \label{prop:worse-case-upper-bound}
  %Let $p \ge 1$ be arbitrary.
  Let $\func: \domain^\dssize \to \range$ be sample-monotone,
  $\range$ be uniformly spaced and the base measure
  $\basemeasure$ in mechanism~\eqref{mech:continuous} be uniform on $\range$,
  where the smoothing $\rho \ge 0$ is arbitrary.
  Then for any $\ds \in \domain^\dssize$, 
  \begin{subequations}
    \begin{equation}
      \label{eqn:as-good-as-laplace}
      \E\left[| \mech(\ds) - \func(\ds) | \right]
      \lesssim %( 5(1+\diffp) +  8 e^{\diffp/2} )  
      % p \cdot {\GS_\func}/{\diffp}. 
	  {\GS_\func}/{\diffp}. 
    \end{equation}
    Let $0 \le \gamma \le \frac{\rho \diffp}{\GS_\func^{2}}$
    and define
    \begin{equation*}
      L \defeq \max_{\ds'}
      \left\{\LS(\ds') : \DSdist(\ds,\ds') \le
      \max\left\{\frac{4}{\diffp},
      \frac{2}{\diffp}\left(\log \frac{1}{\gamma}
      + \log\frac{\GS_\func^{2}}{\rho \diffp}\right)\right\}\right\}.
    \end{equation*}
    Then
    \begin{equation}
      \label{eqn:as-good-as-smooth}
      \E\left[| \mech(\ds) - \func(\ds) | \right]
      \lesssim %(5(1+\diffp) +  8 e^{\diffp/2}) 
      % \frac{p}{\diffp} \cdot \left[L +
      % \gamma \log \frac{1}{\gamma}\right].
      \frac{1}{\diffp} \cdot \left[L +
	\gamma \log \frac{1}{\gamma}\right].
    \end{equation}
  \end{subequations}
\end{proposition}
We know from Eq.~\eqref{eq:laplace-lower-bound} that $\E[ |\mechlap(\ds) -
  \func(\ds)|] = {\GS_\func}/{\diffp }$ and $\E[ |\mechsmlap(\ds) -
  \func(\ds)|] \ge \Omega(1) \cdot {L}/{\diffp}$ if $\delta \le
\dssize^{-1}$. So whenever the global sensitivity $\GS_\func =
\mathsf{poly}(n)$ and $\diffp \gtrsim \mathsf{poly}(n)^{-1}$, for
sample-monotone functions, the inverse sensitivity mechanisms must
outperform the Laplace and smooth
Laplace mechanisms for all instances $\ds \in \domain^\dssize$.

% -*- Mode: latex -*- %

\section{Methodologies: instantiations and approximations of the
  inverse sensitivity mechanism}
\label{sec:examples}

As a second purpose of this paper is to develop practical near-optimal
mechanisms, we highlight the methodological possibilities of the inverse
sensitivity mechanism here. To that end, we provide two
concrete examples: (i) computing the median, where we can efficiently
compute the inverse sensitivity and 
(ii) minimization of robust regression losses.
%(ii) a one-dimensional regression,
%which requires an approximate calculation of the inverse
%sensitivity; and (iii) minimization of smooth and Lipschitz-continuous
%convex losses, which we sketch at a high-level, suggesting a
%mechanism that closely approximates the inverse sensitivity for
%well-behaved problems.
In each  example, we show how to calculate the
inverse modulus~\eqref{eqn:inverse-ls} and briefly analyze
the procedure.
In each case, we recover a mechanism distinct from the traditional
exponential mechanism~\cite{McSherryTa07}.
%% We defer experimental evaluation for each of our examples
%% to Section~\ref{sec:experiments}, demonstrating the effectiveness
%% of the inverse sensitivity mechanism and its approximations.

% -*- Mode: latex -*- %

\newcommand{\datarange}{R}

\subsection{Median of a dataset}
\label{sec:median}

We begin with the median, where the mechanism is
folklore~\cite[Ex.~3.1]{Sheffet18}, though its consistency properties
are not.  Given a dataset $\ds  \in \R^n$,
we wish to calculate $\mathrm{Median}(\ds)$.
For simplicity we assume $x_i \in
[0,\datarange]$ for $\datarange>0$, though our
theory and derivations are completely identical
if the data may be unbounded and we redefine $\func(\ds) = \min\{\datarange,
\max\{-\datarange, \mbox{Median}(\ds)\}\}$.  To implement the
mechanism~\eqref{mech:continuous}, we calculate $\invmodcont_\func$:
%To state the lemma,
%for $t \in [0,\datarange]$ let $c_t = |\{x_i : x_i \le t \} |$ and $u_t =
%|\{x_i : x_i \ge t \} |$.
\begin{lemma}
  \label{lemma:median-invmodcont}
  Let $m = \mathrm{Median}(\ds)$. Then
  for $t \in [0,\datarange]$,
  $\invmodcont_\func(\ds;t) = | \{x_i : x_i \in (t,m] \cup [m,t) \} | $.
      %  \begin{equation*}
      %    \invmodcont_\func(\ds;t) = 
      %    \begin{cases}
      %      0	& \text{if } t = \mathrm{Median}(\ds) \\
      %      \frac{\dssize-1}{2} - c_t	& \text{if } c_t < \frac{\dssize-1}{2}  \\
      %      \frac{\dssize-1}{2} - u_t	& \text{if } c_t > \frac{\dssize-1}{2} .
      %    \end{cases}
      %  \end{equation*}
\end{lemma} 

\begin{proof}
  If $t = m$, certainly $\invmodcont_\func(\ds;t)=0$.
  Otherwise, if $t<m $, then to make $t$ a median, we need to replace
  all the values $x_i$ such that $t < x_i \le m $.
  The case $t>m$ follows similarly.
\end{proof}

Before stating performance guarantees, we describe a procedure
implementing~\eqref{mech:continuous} in time $O(n \log n)$.  Sort the data
so that $x_1 \le \cdots \le x_n$. For $0 \le k \le \dssize$, let $I_k = \{ t
\in [0,\datarange] : \invmodcont_\func(\ds;t) = k \}$, a union of two
intervals by Lemma~\ref{lemma:median-invmodcont}.  As $\ds$ is sorted, we
can find the sets $I_k$ by iterating once over $\ds$.  Then to
sample the output of the mechanism, sample an index $k
\in [\dssize]$ with probability proportional to $e^{-k \diffp/2} |I_k|$
where $|I_k| = \int_{I_k} dt$ is the volume of $I_k$, and return an element
uniform in $I_k$.
%\textcolor{red}{maybe add a comment that smooth sensitivity 
%depends on $\diffp$ and therefore even more inefficient}

Let us examine the behavior of mechanism~\eqref{mech:continuous}.
Assume $\ds \in [0, \datarange]^n$ has $x_i \simiid P$, and we make the
standard assumption~\cite[Ch.~21]{VanDerVaart98} that $P$ has a continuous
density $\pdf_P$ near $m \defeq \mathrm{Median}(P)$. Then
the mechanism $\mechcont$ returns an accurate
estimate with exponentially high probability, as the
following proposition shows. (See
Appendix~\ref{sec:median-proofs} for a proof.)
\begin{proposition}
  \label{prop:mech-median-performance-improved}
  Let $\gamma >0$, $0 \le u \le \gamma/4$, and 
  $p_{\min} =  \inf_{|t - m| \le 2 \gamma} \pdf_P(t)$.
  Let $\what{m}_n \defeq \mathrm{Median}(\ds)$.
  Under the above conditions, the mechanism $\mechcont$
  with smoothing parameter $\rho$
  satisfies
  \begin{equation*}
    \P \left(|\mechcont(\ds) - \what{m}_n | > 2 u + \rho  \right) 
    \le  \frac{\datarange}{\rho} \exp\left(-\frac{n p_{\min} u \diffp}{4}\right)
    +  4 \exp\left(- \frac{n \gamma^2 p_{\min}^2}{4}\right)
    + \frac{2 \gamma}{u} \exp\left(-\frac{n p_{\min} u}{8}\right),
  \end{equation*}
  where the probability is jointly over $x_i \simiid P$ and
  the randomness in $\mechcont$.
\end{proposition}

As an alternative to understanding the (conditional on $\ds$) instance
optimality of a mechanism, as Defs.~\ref{def:local-minimax-optimal}
and~\ref{def:instance-optimal} identify, we may instead consider the rate at
which we can take the privacy parameter $\diffp = \diffp_n \to 0$ as
the sample size $n$ grows while maintaining optimal statistical convergence.
To understand this, recall that the empirical median
$\what{m}_n \defeq \mathrm{Median}(\ds)$
satisfies
\begin{equation}
  \label{eqn:empirical-median-estimate}
  \sqrt{n}(\what{m}_n - m) \cd \normal\left(0,
  \frac{1}{4 \pdf_P(m)^2}\right),
\end{equation}
where the rate $\sqrt{n}$ and variance $1/4 \pdf_P(m)^2$ are
optimal~\cite[Chs.~21, 25.3]{VanDerVaart98}. In this case, the inverse
sensitivity mechanism~\eqref{mech:continuous} can achieve the
asymptotics~\eqref{eqn:empirical-median-estimate} whenever $\diffp \gg \log
n / \sqrt{n}$. Indeed, in
Proposition~\ref{prop:mech-median-performance-improved}, take the triple $u,
\diffp, \gamma$ such that $\gamma = n^{-1/4}$, $\log n /n \ll u \ll 1 /
\sqrt{n}$, $\rho = 1/n$, and $u \diffp \gg \log n / n$. Then the proposition
guarantees $\sqrt{n} (\mechcont(\ds) - \what{m}_n) \cp 0$, and Slutsky's
lemmas~\cite[Ch.~2.8]{VanDerVaart98} yield
\begin{equation*}
  \sqrt{n}\left(\mechcont(\ds) - m\right)
  \cd \normal\left(0, \frac{1}{4 \pdf_P(m)^2}\right)
  ~~~
  \mbox{whenever} ~ \diffp \gg \frac{\log n}{\sqrt{n}}.
\end{equation*}
%% Proposition~\ref{prop:mech-median-performance-improved} shows
%% that $\mechcont$ achieves this whenever $\diffp \gg 1 / \sqrt{n}$.

\newcommand{\mechsmoothlap}{\mech_{\textup{smooth-Lap}}}

In contrast, the following lower bound shows that the smooth Laplace
mechanism may achieve this rate only if $\diffp \gg n^{-1/4}$, ignoring
logarithmic factors, assuming the mechanism is $(\diffp, \delta)$ private
with $\delta \le 1/n$.  As typically $\delta$ is assumed negligible in
$n$~\cite{DworkRo14}, this is no restriction, and we make the
problem easier by assuming $P$ has density $\pdf_P$ satisfying $p_{\min} \le
\pdf_P \le p_{\max}$.
\begin{lemma}
  \label{lemma:smooth-laplace-median-failure}
  Let the above conditions hold.
  The smooth Laplace mechanism
  adds noise $\frac{\alpha}{\diffp} \mathsf{Lap}(1)$
  where $\alpha \ge \frac{\log(n)}{2e p_{\max} n \diffp}$ 
  with probability at least
  $q \defeq 1 - 2 \datarange n \diffp p_{\max} \exp(-\frac{p_{\min}}{16 p_{\max}}
  \frac{\log n}{\diffp})$, so
  \begin{equation*}
    \E[| \mechsmlap(\ds) - \mathrm{Median}(\ds)| ] \ge
    \frac{q}{2e p_{\max} } \frac{\log n}{n \diffp^2}.
  \end{equation*}
  %% \textcolor{red}{currently, proof works for $\rho=\datarange$ but we can make it work for $\rho = n^{-0.4} $}.
\end{lemma}
\noindent
See Appendix~\ref{sec:proof-smooth-laplace-median-failure} for a proof.  To
achieve the asymptotically optimal
convergence~\eqref{eqn:empirical-median-estimate} for the smooth Laplace
mechanism we may provide $(\diffp, \delta)$-privacy only for $\diffp \gg
n^{-1/4}$, quadratically worse than the inverse sensitivity
mechanism~\eqref{mech:continuous}.

% -*- Mode: latex -*- %

\newcommand{\dsx}{\boldsymbol{x}}
\newcommand{\dsy}{\boldsymbol{y}}

\subsection{Robust regression problems}
\label{sec:example-regression-new}

We revisit the robust regression
problems we identify in Example~\ref{example:robust-regression-smoothness},
sketching application of the inverse
sensitivity mechanism.
As in the example, we have data
$(x_i,y_i) \in \R^d \times \R$ with
$\norm{x_i}_2 \le \Bx$ and loss
$\ell(\theta;x_i,y_i) = h(\<\theta,x_i\> - y_i)$
where $h: \R \to \R_+$ is convex, symmetric and $1$-Lipschitz.
Rather than partial minimization as in Ex.~\ref{example:smoothness-partial-min},
we consider the minimizer
\begin{equation*}
  \what{\theta}_n(\dsx,\dsy) 
  = \argmin_{\theta \in \Theta} \riskn(\theta; \dsx, \dsy)
  \defeq \frac{1}{n} \sum_{i=1}^{n} \ell(\theta;x_i,y_i).
\end{equation*}
%where $h: \R \to \R_+$ is convex, symmetric and $1$-Lipschitz 
%(i.e., $\sup_{t \in \R} |\nabla h(t)| = 1$), 
%and $\Theta = \{\theta: \norm{\theta}_2 \le \Btheta \} $.
%Many major applications fall into this setting
%including the Huber loss, $\diffp$-insensitive loss.
%% Our formulation includes for example the $\alpha$-insensitive loss
%% where
%% $h_\alpha(t) = \alpha \log( 1 + e^{{t}/{\alpha}})
%% + \alpha \log( 1 + e^{-{t}/{\alpha}})$; 
%% in Section~\ref{sec:experiment-regression-new} we provide 
%% an empirical evaluation for minimization of the 
%% $\alpha$-insensitive loss.

For simplicity in the calculations, we consider a slight tweak to our
privacy definitions to instead consider addition of a user rather than
substitution (a standard alternative~\cite{DworkRo14}).
%% To simplify calculations, here we choose to pursue the 
%% user addition privacy definition (see for example 
%% report noisy max algorithm~\cite[Chapter 3]{DworkRo14}).
Here, a mechanism is $\diffp$-differentially private with respect to user
additions if for any two instances $(\dsx,\dsy)$ and $(\dsx',\dsy') = (\dsx,\dsy) \cup
\{(x'_{n+1},y'_{n+1})\}$,
\begin{equation*}
  \frac{\P(\mech(\dsx,\dsy) \in S)}{\P(\mech(\dsx',\dsy') \in S)} \le e^{\diffp}
  ~~ \mbox{and} ~~
  \frac{\P(\mech(\dsx',\dsy') \in S)}{\P(\mech(\dsx,\dsy) \in S)} \le e^{\diffp}.
\end{equation*}
We then define the
user-addition length
\begin{equation}
  \label{eq:inv-usr-add}
  \invmodusradd(\dsx,\dsy;\theta) = \inf_{\dsx' \in \R^{n' \times d}, \dsy' \in \R^{n'}} 
  \{\dssize' - \dssize :   \what{\theta}_n(\dsx',\dsy') = \theta,
  x'_i = x_i, y'_i = y_i, 1 \le i \le n \}.
\end{equation}
The inverse sensitivity mechanism for this problem 
instantiates~\eqref{mech:continuous} with 
the user-addition inverse sensitivity~\eqref{eq:inv-usr-add}.
The following lemma shows that it is $\diffp$-private (for addition of users),
as $\invmodusradd$ is $1$-Lipschitz with respect to dataset additions.

%% \hacomment{this is not always correct. Because user-addition 
%% only allows adding users. Maybe we should make this 
%% more clear?}
 
%While it is not immediately clear if using this new version
%of the inverse sensitivity yields private mechanisms,
%the following lemma shows that the global
%sensitivity of $\invmodusradd$ with respect to 
%user additions is equal to $1$, hence the inverse
%sensitivity mechanism with the modified 
%path-length~\eqref{eq:inv-usr-add} is differentially
%private. We prove the lemma in Appendix~\ref{sec:proof-inv-add-sensitivity}.
\begin{lemma}
  \label{lemma:inv-add-sensitivity}
  The mechanism~\eqref{mech:continuous} with
  $\invmodusradd$~\eqref{eq:inv-usr-add}
  is $\diffp$-DP (with respect to user-additions).
%	 Let $(\dsx,\dsy)= \{(x_i,y_i)\}_{i=1}^n$ and 
%	 $(\dsx',\dsy') =  (\dsx,\dsy) \cup \{(x'_{n+1},y'_{n+1})\}$.
%	 Then for every $\theta \in \Theta$,
%	 \begin{equation*}
%	 	|\invmodusradd(\dsx,\dsy;\theta) - 
%	 		\invmodusradd(\dsx',\dsy';\theta) | \le 1.
%	 \end{equation*}
\end{lemma}

The following lemma, proved in Appendix~\ref{sec:proof-inv-add-calc},
characterizes the inverse sensitivity.
\begin{lemma}
  \label{lemma:inv-add-calc}
  Assume $\ltwo{x} \le \Bx$ for all $x$.
  If $\what{\theta}_n(\dsx,\dsy) \in \interior{\Theta}$,
  then for every $\theta \in \Theta$,
  \begin{equation*}
    \invmodusradd(\dsx,\dsy;\theta) = 
    \ceil{\frac{n \norm{\nabla \riskn(\theta;\dsx,\dsy)}_2}{\Bx}}.
  \end{equation*}
\end{lemma}
\noindent
Whenever $\what{\theta}_n(\dsx, \dsy) \in \interior\Theta$, the inverse
sensitivity mechanism~\eqref{mech:continuous} is thus
\begin{equation}
  \label{eq:inv-mech-erm}
  \pdf_{\mathsf{inv}}(\theta)
  \propto \exp\left(-\frac{n \diffp}{2}
  \ceil{\frac{n \norm{\nabla \riskn(\theta;\dsx,\dsy)}_2}{ \Bx}} \right).
\end{equation} 
No matter the value of $\what{\theta}_n$ (even if it is on the boundary
$\boundary \Theta$, so that $\what{\theta}_n \not\in \interior\Theta$), the
mechanism~\eqref{eq:inv-mech-erm} is still $\diffp$-private, as $\boundary
\Theta$ has Lebesgue measure 0 (its utility may suffer).
The mechanism~\eqref{eq:inv-mech-erm} is distinct from the
standard exponential mechanism~\cite{McSherryTa07} for empirical risk
minimization problems; the exponential mechanism has density
\begin{equation*}
  \pdf_{\mathsf{exp}}(\theta)
  \propto \exp\left(-\frac{n \diffp}{2 \Bx} \riskn(\theta; \dsx, \dsy)\right).
\end{equation*}
For $\theta = \theta_n + \Delta$ near $\theta_n = \what{\theta}_n(\dsx,
\dsy)$, if we ignore large $\Delta$, this is (heuristically) a density
\begin{equation*}
  \pdf_{\mathsf{exp,heur}}(\theta_n + \Delta) \propto
  \exp\left(-\frac{n \diffp}{2 \Bx} \<\nabla \riskn(\theta_n; \dsx, \dsy), \Delta\>
  \right),
\end{equation*}
which is similar to but distinct from the inverse
mechanism~\eqref{eq:inv-mech-erm}.

We proceed with a heuristic analysis of the inverse sensitivity mechanism
here, assuming throughout that $\theta_n \defeq \what{\theta}_n(\dsx,\dsy)
\in \interior{\Theta}$.  (As in
Example~\ref{example:smoothness-partial-min}, utility crucially relies on
this interior assumption.)  We assume the losses $h(\cdot)$ are $\mc{C}^2$,
and that the Hessian $\nabla^2 \riskn(\theta_n; \dsx,\dsy) \succ
0$. Using the Taylor approximation $\nabla \riskn(\theta_n; \dsx,\dsy) =
\nabla^2 \riskn(\theta_n; \dsx,\dsy)(\theta - \theta_n) +
O(\norm{\theta-\theta_n}^2)$, we get that for small $\Delta$, the
density~\eqref{eq:inv-mech-erm} is approximately
\begin{equation}
  \label{eqn:heuristic-path-erm}
  \pdf(\theta_n + \Delta)
  \propto \exp\left(-\frac{n \diffp}{2 \Bx}
  \norm{\nabla^2 \riskn(\theta_n; \dsx,\dsy) \Delta}\right).
\end{equation}
We can compute the error of the heuristic
mechanism~\eqref{eqn:heuristic-path-erm}. To sample such a distribution, we
draw a radius $R \sim \gammadist(d, 1)$ and $U \sim
\uniform(\sphere^{d-1})$, then set $\Delta = \frac{2\Bx}{n \diffp} \nabla^2
\riskn(\theta_n; \dsx,\dsy)^{-1} \cdot R \cdot U$ (see
Appendix~\ref{sec:sample-gamma-like}).  Our
heuristic~\eqref{eqn:heuristic-path-erm} for
mechanism~\eqref{eq:inv-mech-erm} thus achieves
\begin{equation*}
  \E\left[\norm{\theta - \theta_n}^2\right]
  \asymp \frac{\Bx^2 \E[R^2]}{n^2 \diffp^2}
  \E[\norm{\nabla^2 \riskn(\theta_n; \dsx,\dsy)^{-1} U}^2]
  = \frac{\Bx^2 (d + 1)}{n^2 \diffp^2}
  \tr\left(\nabla^2 \riskn(\theta_n;  \dsx,\dsy)^{-2}\right).
\end{equation*}

While it may be challenging to efficiently
sample from the inverse sensitivity mechanism~\eqref{eq:inv-mech-erm},
a Metropolis-Hastings (MH) scheme yields
an $(\diffp,\delta)$-differentially private mechanism.
The MH algorithm samples points $\theta, t$
as follows, beginning from $\theta$. Given a transition kernel
$q(\theta, \cdot)$ satisfying $q(\theta, A) \ge 0$ and $q(\theta, \R^d) =
1$, we iteratively draw $T \sim q(\theta, \cdot)$ and accept the move $T =
t$ with probability $\min\{\frac{\pi(t)}{\pi(\theta)} \frac{q(t,
  \theta)}{q(\theta, t)}, 1\}$, otherwise remaining at $\theta$.
If the proposal
$q$ is independent of the initial point $\theta$, then we have geometric
mixing if $q(t) / \pi(t) \ge \beta > 0$ for all $t$, and
$\tvnorms{P^n(\theta_0, \cdot) -
  \pi} \le (1 - \beta)^n$
under this condition~\cite[Thm.~2.1]{MengersenTw96}.

We now sketch a fast mixing algorithm for sampling from the inverse
sensitivity mechanism under a few additional conditions on the losses
$h$. In addition to our previous assumptions, we assume that the losses $h$
have $\lipgrad(x)$-Lipschitz gradient over $\Theta$.  We assume that
$\nabla^2 \riskn(\theta_n) \succeq \lambda I$. (It is possible to use the
propose-test-release framework~\cite[cf.][Ch.~7.2]{DworkRo14} to check that
this holds, failing with only a prescribed probability $\delta$.) Let $r_n$
be a be a fixed rate satisfying
$n^{-1} \ll r_n \ll n^{-2/3}$. Under these conditions, we have the
following lemma, whose proof we sketch in
Appendix~\ref{sec:proof-bounded-proposals}, as the mixing time analysis is not
our main focus.
\begin{lemma}
  \label{lemma:fast-mixing-ratio}
  Let the conditions above hold, and
  consider the proposal density
  \begin{equation*}
    q(\theta) \propto \exp\left(-\frac{n \diffp}{2}
    \left(\norm{\nabla^2 \riskn(\theta_n)
      (\theta - \theta_n)} \wedge r_n\right) \right)
    \indic{\theta \in \Theta}.
  \end{equation*}
  There exists a numerical constant $\beta > 0$ such that
  $q(t) / \pi(t) \ge \beta$ for all $t \in \Theta$.
\end{lemma}

Assuming that we can compute $\vol(\Theta)$, the Lebesgue volume of
$\Theta$, it is possible to efficiently sample from $q(\theta)$.  Indeed, as
in Appendix~\ref{sec:sample-gamma-like}, a change of variables allows easy
sampling from the density $f(z) \propto \exp(-\norm{A z})$.  Thus, by
rejection sampling from the density proportional to $\exp(-\frac{n
	\diffp}{2} \norm{\nabla^2 \riskn(\theta_n) (\theta - \theta_n)})$ we can
draw exactly from $q(\theta)$.  Running the Metropolis-Hastings algorithm
using proposal $q$ for (say) $n$ steps then gives an $(\diffp,
e^{-\Omega(n)})$-differentially private algorithm
using~\cite[Thm.~2.1]{MengersenTw96} and
Lemma~\ref{lemma:fast-mixing-ratio}.

% -*- Mode: latex -*- %

\section{Experiments}
\label{sec:experiments}

We conclude with an experimental evaluation
of the inverse sensitivity mechanisms
for our examples in Section~\ref{sec:examples}.
%% We begin with our experiments for median estimation (Section~\ref{sec:experiment-median}),
%% then we proceed to a risk
%% minimization experiment (Section~~\ref{sec:experiment-regression-new}).
%and finally we conclude with an empirical risk minimization experiment (Section~~\ref{sec:experiment-erm}).
Our theory suggests that the inverse sensitivity mechanism
should provide a competitive utility against any private mechanism, 
especially for the high-privacy regime where $\diffp \ll 1$, and indeed,
the inverse sensitivity mechanism demonstrates strong utility
in our two experiments.
%Our experiments reaffirm our theory, demonstrating the strong performance of our inverse sensitivity mechanisms. 

% -*- Mode: latex -*- %

\vspace{-.25cm}
\begin{figure}[ht]
	\begin{tabular}{cc}
		%\hspace{-.25cm}
		\begin{overpic}[width=.5\columnwidth]{%,grid]{%
				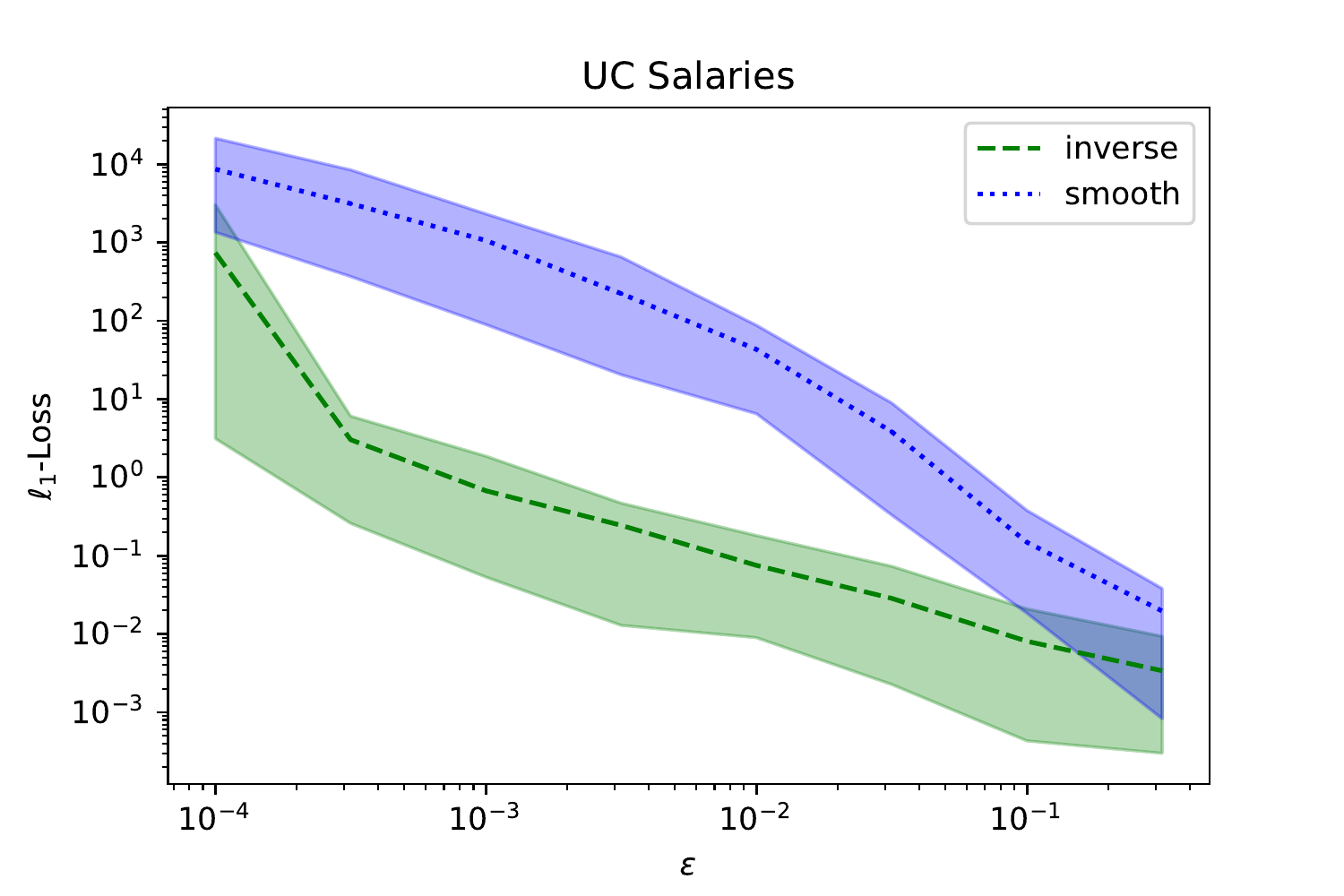}
			\put(0,28){
				\tikz{\path[draw=white,fill=white] (0, 0) rectangle (.4cm,2.5cm);}
			}
			\put(43,0.5){
				\tikz{\path[draw=white,fill=white] (0, 0) rectangle (2cm,.3cm);}
			}
			\put(-0.5,23){
				\rotatebox{90} {\small Loss [in $1K$ \$]}
			}
			\put(50,0.5){{\small $\diffp$}}
		\end{overpic}
		&
		\begin{minipage}{.4\columnwidth}
			\vspace{-5cm}
			\hspace{-2cm}
			\caption{\label{fig:median} 
				The accuracy of each mechanism $|\mech(\ds) - \text{Median}(\ds)|$ as a function 
				of the privacy parameter $\diffp$ with $0.9$ confidence intervals on the UC salary 
				dataset.
			}
		\end{minipage}
	\end{tabular}
	\vspace{-.4cm}
\end{figure}
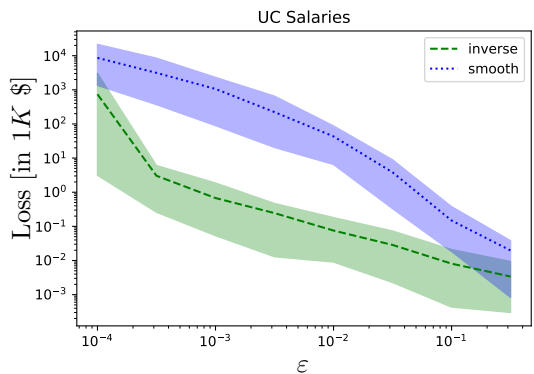

\subsection{Median of a dataset}
\label{sec:experiment-median}

We begin our experiments with the median example 
of Section~\ref{sec:median}, where we aim 
to evaluate the performance of the inverse sensitivity mechanism~\eqref{mech:continuous} (setting $K=10^7$ and $\rho = 1/n$)
against the smooth Laplace mechanism~\eqref{mech:smooth-laplace} (setting the standard value $\delta=n^{-1.1}$)
for estimating the median of a dataset.
Our experiments use a publicly available dataset consisting of the salaries
of all employees in the University of California system.
We run each method $50$ times and report the median of the absolute loss
$|\mech(\ds) - \text{Median}(\ds)|$, with 90\% confidence intervals across
all experiments.
Fig.~\ref{fig:median} shows the results of each mechanism
as a function of the privacy parameter $\diffp$. 
The plot---as expected---shows a 2--3 order of magnitude improvement in
error of the inverse sensitivity mechanism~\eqref{mech:continuous}
over the smooth Laplace mechanism (which additionally only guarantees
$(\diffp, \delta)$-differential privacy).
%% Further, we stress that the inverse sensitivity mechanism~\eqref{mech:discrete}
%% preserves $\diffp$-differential privacy, in contrast to the smooth Laplace
%% mechanism~\eqref{mech:smooth-laplace} that only guarantees $(\diffp,\delta)$-differential
%% privacy for $\delta=n^{-1.1}$ in our experiments.

%\begin{figure}[ht]
%	%% \vspace{-.6cm}
%	\vspace{-0.2cm}
%	\begin{center}
%			\begin{overpic}[width=.5\columnwidth]{%,grid]{%
%					../code/plots/UC-final}
%				\put(0,28){
%					\tikz{\path[draw=white,fill=white] (0, 0) rectangle (.4cm,3cm);}
%				}
%				\put(43,0.5){
%					\tikz{\path[draw=white,fill=white] (0, 0) rectangle (2cm,.3cm);}
%				}
%				\put(0,23){
%					\rotatebox{90} {\small Loss [in $1K$ \$]}
%				}
%				\put(50,0.5){{\small $\diffp$}}
%				%\put(0,28){
%				%	\rotatebox{90}{{\small $\func_{\mathrm{step}}(\ds)$}}}
%				%\put(32,72){
%				%	\tikz{\path[draw=white,fill=white] (0, 0) rectangle (3cm,.35cm);}
%				%}
%			\end{overpic}
%		%\vspace{-.4cm}
%		\caption{\label{fig:median} 
%		  The accuracy of each mechanism $|\mech(\ds) - \text{Median}(\ds)|$ as a function 
%		  of the privacy parameter $\diffp$ with $0.9$ confidence intervals on the UC salary 
%		  dataset. 
%		 }
%	\end{center}
%	\vspace{-.6cm}
%\end{figure}

% -*- Mode: latex -*- %

\newcommand{\step}{\eta} % Wide tilde

\subsection{Robust regression and empirical risk minimization}
\label{sec:experiment-regression-new}

In our final experiment, we investigate the robust regression problem of
Section~\ref{sec:example-regression-new}.  For a fixed $\theta\opt$, the
data follows the distribution $y_i = x_i \theta\opt + w_i$, where $x_i
\simiid \uniform[-2, 2]$, $w_i \simiid \uniform[-.05, .05]$, and in each
repetition of the experiment we draw $\theta\opt \sim \uniform[-5,5]$.
Following our notation from
Example~\ref{example:robust-regression-smoothness}, we consider losses
$\ell_\alpha(\theta; x, y) = h_\alpha(\theta x - y)$ for $h_\alpha(t) \defeq
\alpha \log(1 + e^{t/\alpha}) + \alpha \log(1 + e^{-t/\alpha})$, varying the
$\alpha$ parameter as well to induce more smoothness ($\alpha$ large) or
less ($\alpha$ small, so that $h(t) \approx |t|$).

We compare two algorithms in this experiment. The first is
specialization~\eqref{eq:inv-mech-erm} of the inverse sensitivity
mechanism~\eqref{mech:continuous}, which we implement by running $500$ steps
of Section~\ref{sec:example-regression-new}'s Metropolis-Hastings
procedure. We also consider private Stochastic Gradient Descent
with a moments accountant~\cite{AbadiChGoMcMiTaZh16}, which achieves
state-of-the-art performance. Briefly, at each iteration, private SGD
subsamples a set $S \subset [n] $ of users with probability $q$, then sets
\begin{equation}
  \label{eq:priv-sgd}
  \theta_{t+1} = \theta_t - \frac{\step_t}{|S|} \left( \sum_{i \in S}
  \nabla \ell_\alpha(\theta_t;x_i,y_i) +
  \normal(0,\sigma^2 \lipobj^2)\right),
\end{equation}
where $\step_t$ is a stepsize rate, $\sigma$ is a noise parameter, and
$\lipobj$ a bound on the $\ell_2$-norm of the gradient, which in this case
is exactly $\max_i |x_i|$.

%We use standard values for these parameters
%without searching for optimal configuration.

\begin{figure}[t]
  %% \vspace{-.6cm}
  \vspace{-0.2cm}
  \begin{center}
    \begin{tabular}{ccc}
      \begin{overpic}[width=.36\columnwidth]{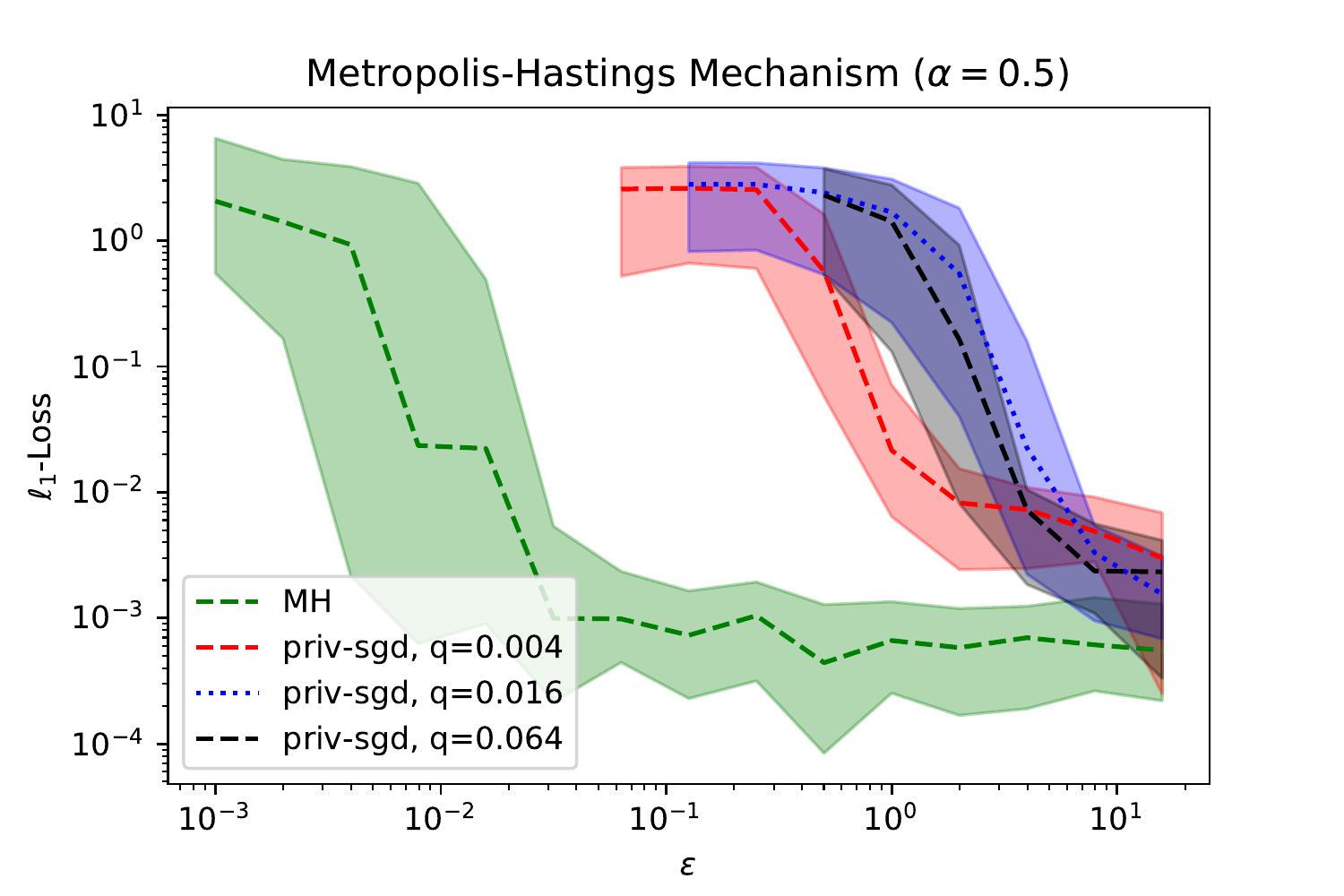}
	\put(14,59){
	  \tikz{\path[draw=white,fill=white] (0, 0) rectangle (5cm,.5cm);}
	}
      \end{overpic} &
      \hspace{-1cm}
      \begin{overpic}[width=.36\columnwidth]{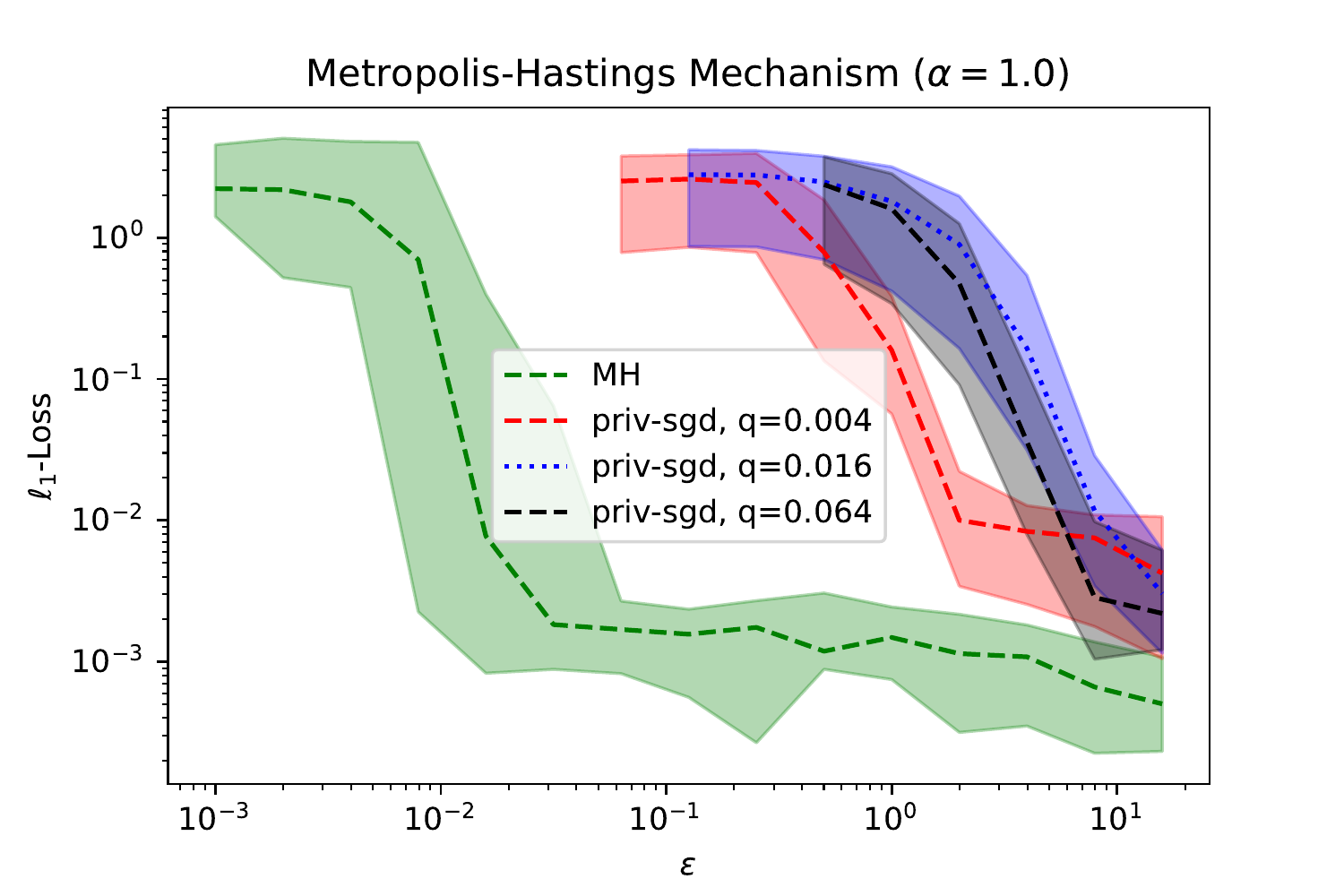}
	\put(14,59){
	  \tikz{\path[draw=white,fill=white] (0, 0) rectangle (5cm,.5cm);}
	}
      \end{overpic} &
      \hspace{-1cm}
      \begin{overpic}[width=.36\columnwidth]{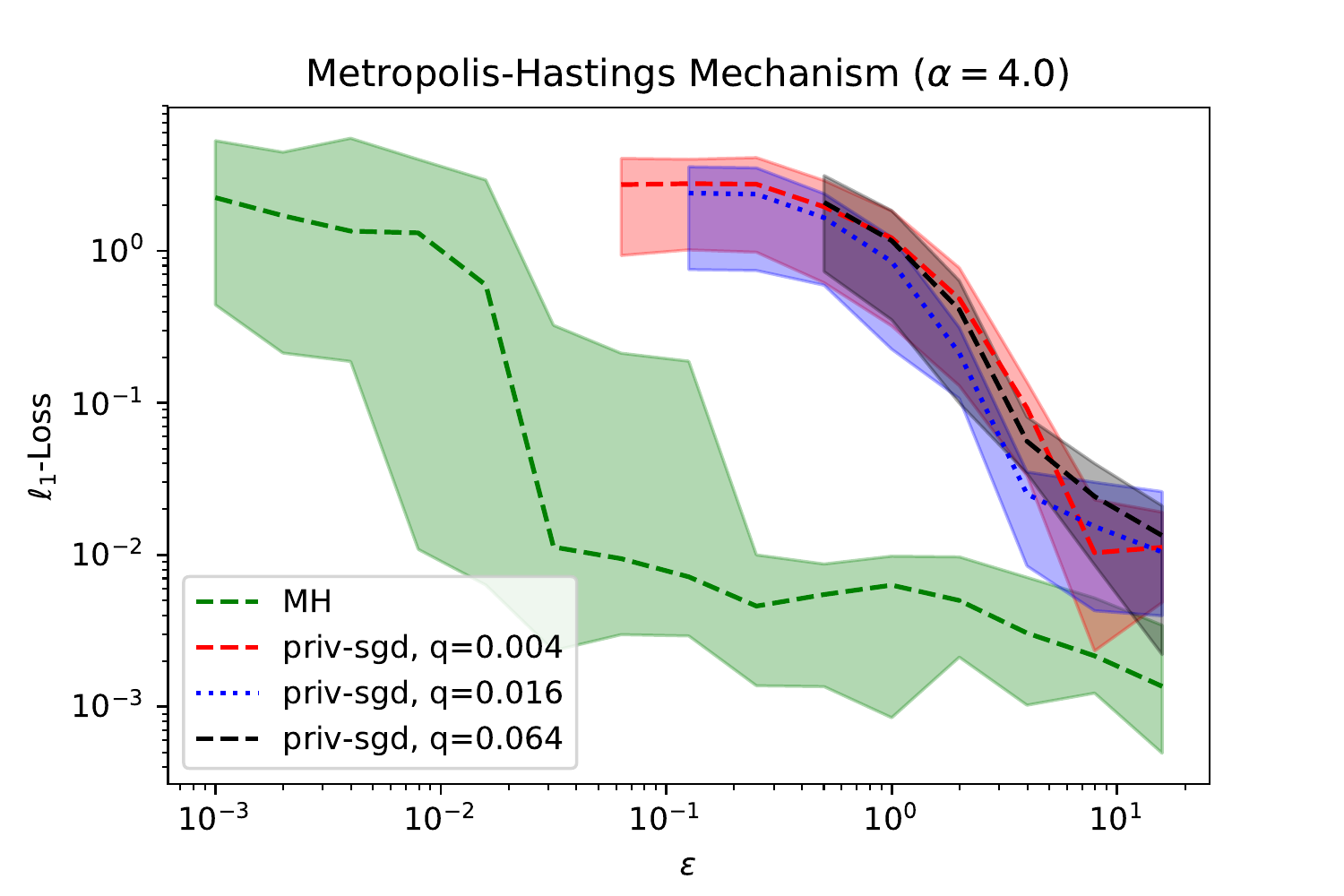}
	\put(14,59){
	  \tikz{\path[draw=white,fill=white] (0, 0) rectangle (5cm,.5cm);}
	}
      \end{overpic} \\
      (a) & \hspace{-.9cm} (b) & \hspace{-.8cm}(c)
    \end{tabular}
    %\vspace{-.4cm}
    \caption{\label{fig:erm} Absolute loss $|\mech(x,y) - \theta^\star|$ as
      a function of the privacy parameter $\diffp$ for the inverse
      sensitivity mechanism~\eqref{eq:inv-mech-erm} and private
      SGD~\eqref{eq:priv-sgd} for the robust regression problem with losses
      $\ell_\alpha(\theta; x, y) = h_\alpha(\<\theta, x\> - y)$ and
      $h_\alpha(t) = \alpha \log(1 + e^{t/\alpha}) + \alpha \log(1 +
      e^{-t/\alpha})$. $q$ specifies the subsampling rate for private SGD.
      (a) $\alpha = 0.5$. (b) $\alpha = 1$. (c) $\alpha = 4$.}
  \end{center}
  \vspace{-.6cm}
\end{figure}
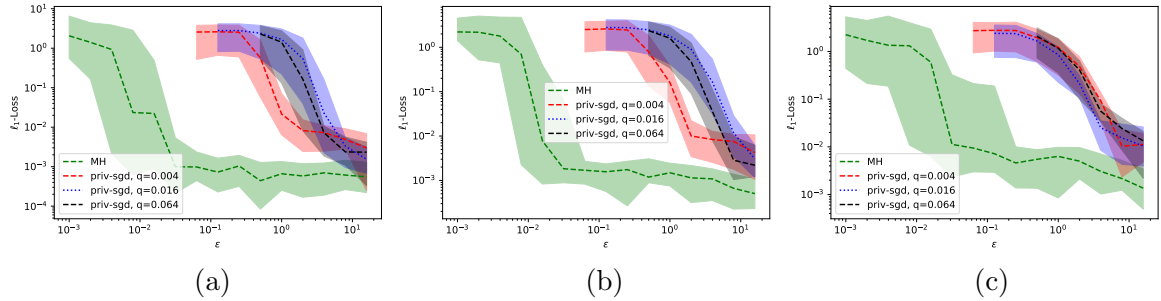

The private SGD procedure requires several parameters, and we
attempt to optimize them.  We vary $q \in \{0.004,0.016,0.064\}$ and
set $\sigma = 2$; for a desired privacy level $\diffp$, we then calculate
the maximum number of iterations $T$ that the (computational)
moment-accounting technique~\cite{AbadiChGoMcMiTaZh16} allows for the given
$q, \sigma$.  We used a stepsize $\step_t = \step_0 / \sqrt{t}$ for each
step; as SGD is extremely sensitive to the choice of $\step_0$ even in the
non-private setting~\cite{AsiDu19}, we vary $\step_0 \in \{.05, .1, .3,
1, 3, 10\}$.  We run each method $30$ times,
where to most advantage the private stochastic gradient
algorithm~\eqref{eq:priv-sgd}, we choose the value of $\step_0$ for each
distinct privacy level $\diffp$ that yields the best convergence.  (The
mechanism~\eqref{eq:priv-sgd} is sensitive to these values.)

In Figure~\ref{fig:erm}, we report our results. Each plot
displays the median $\ell_1$-loss
$|\mech(x,y) - \theta^\star|$ with 95\% coverage over all experiments versus the
attained privacy level $\diffp$, as $\diffp$ varies from $10^{-3}$ to 1,
and plots (a), (b), and (c) correspond (respectively) to the choice
$\alpha = .5, 1, 4$ in the loss $\ell_\alpha$. Each plot makes clear
that the inverse sensitivity mechanism achieves much better convergence
than private SGD~\eqref{eq:priv-sgd}, which (because of the guarantees
the moments-accountant gives) cannot provide privacy for $\diffp \le .1$ or
so. We see, roughly, that there are several orders of magnitude difference
in the losses of the inverse sensitivity mechanism~\eqref{mech:continuous}
and the SGD procedure, except when $\diffp$ is (perhaps unacceptably) large.
We hope this lends credence to the desiderata we have tried to highlight
in this paper, that one should attempt to be optimal for the problem at hand.

%Finally, we remark that our inverse sensitivity mechanism
%preserves $\diffp$-differential privacy, in contrast to private SGD
%which can ensure only $(\diffp,\delta)$-differential privacy (with $\delta=n^{-1.1}$
%in our experiments).

\renewcommand{\loss}{L}

\appendix

\section{Proofs of lower bounds (Section~\ref{sec:lower-bounds})}

% -*- Mode: latex -*- %

\subsection{Proof of Proposition~\ref{proposition:lower-01-loss}}
\label{sec:proofs-0-1-loss}

We begin with a useful lemma formalizing the intuition that if a mechanism
$\mech$ is unbiased and returns the correct answer $\func(\ds)$ with high
probability, it must as well return $t$ with high probability---depending on
$\invmodcont_\func(\ds;t)$---to preserve differential privacy.
\begin{lemma}
	\label{lemma:diff-values-ratio}
	Let $\mech$ be $\losszo$-unbiased and $\diffp$-DP. Then 
	for any $t \in \range$,
	\begin{equation*}
		\P(\mech(\ds) = t)
		\ge  e^{-2 \invmodcont_\func(\ds;t) \diffp}
		\P(\mech(\ds)=\func(\ds)) .
	\end{equation*}
\end{lemma}
\begin{proof}
	For any two instances $\ds,\ds' \in \domain^\dssize$,
	\begin{align*}
		\frac{\P(\mech(\ds) = \func(\ds))}{\P(\mech(\ds') =\func(\ds'))}
		=  \frac{\P(\mech(\ds) = \func(\ds))}{\P(\mech(\ds') = \func(\ds))}
		\frac{\P(\mech(\ds') = \func(\ds))}{\P(\mech(\ds') = \func(\ds'))} 
		& \stackrel{(i)}{\le} e^{\DSdist(\ds,\ds') \diffp} \frac{\P(\mech(\ds') = \func(\ds))}{\P(\mech(\ds') = \func(\ds'))} \\
		& \stackrel{(ii)}{\le} e^{\DSdist(\ds,\ds') \diffp},	
	\end{align*}
	where $(i)$ follows from the definition of $\diffp$-DP and $(ii)$
	follows since $\mech$ is $\ell_{0,1}$-unbiased.
	Denote $\ell = \invmodcont_\func(\ds;t)$,
	%The definition of $\invmodcont_\func(\ds;t)$ implies 
	%that 
	so there exists an instance $\ds' \in \domain^\dssize$ such that $\DSdist(\ds,\ds') = \ell$ and 
	$\func(\ds') = t$. To finish the proof,
	we use the above inequality and the definition of differential privacy to get
	\begin{equation*}
		\P(\mech(\ds) =  t) 
		\ge e^{-\ell \diffp} \P(\mech(\ds')= t) 
		= e^{-\ell \diffp} \P(\mech(\ds')=\func(\ds'))
		\ge  e^{-2 \ell \diffp} \P(\mech(\ds)=\func(\ds)) .
	\end{equation*}
	as desired.
\end{proof}

We now return to prove Proposition~\ref{proposition:lower-01-loss}, beginning with the
instance-dependent bounds.  For any $\ds \in \domain^\dssize$,
Lemma~\ref{lemma:diff-values-ratio} implies that
\begin{align*}
	1  
	= \sum_{t \in \range} {\P(\mech(\ds)= t} )
	\ge   \sum_{t \in \range} e^{-2 \invmodcont_\func(\ds;t) \diffp} \P(\mech(\ds)=\func(\ds)).
\end{align*}
The second part of the proposition follows.
Now we prove the first part of the proposition.
For a mechanism $\mech$, let $p = \min_{\ds \in \domain^\dssize} {\P(\mech(\ds)=\func(\ds))}$.
For any $t \in \range$, 
there exists a dataset $\ds'$ such that $\func(\ds')=t$ and $\DSdist(\ds,\ds') = \invmodcont_\func(\ds;t)$. 
Therefore, the definition of differential privacy implies
\begin{align*}
	\P(\mech(\ds)= t) 
	& \ge e^{-\invmodcont_\func(\ds;t) \diffp} \P(\mech(\ds') =  t)
	=   e^{-\invmodcont_\func(\ds;t) \diffp} \P(\mech(\ds') = \func(\ds')) 
	\ge e^{-\invmodcont_\func(\ds;t) \diffp} p.
\end{align*}
Using the last inequality, the claim follows as
$1  
= \sum_{t \in \range} {\P(\mech(\ds)= t) } 
\ge \sum_{t \in \range} {e^{-\invmodcont_\func(\ds;t) \diffp} p}.
$

\subsection{Proof of Theorem~\ref{thm:lower-bound-general-loss}}
\label{sec:proof-bound-general-loss}

%Before proving the lower bound, we provide a useful 
%lemma which states that unbiased mechanisms
%incur similar losses for neighboring instances.

As in the proof of Proposition~\ref{proposition:lower-01-loss}, we begin
with a lemma that nearby datasets must incur similar losses under
differential privacy.

\begin{lemma}
	\label{lemma:neigh-loss-paper}
	Let $\mech$ be $\loss$-unbiased and $\diffp$-DP.
	Then for instances $\ds,\ds' \in \domain^\dssize$, 
	\begin{equation*}
		\E \left[ \loss(\mech(\ds),\func(\ds))\right] 
		\le e^{\DSdist(\ds,\ds') \diffp} \E \left[\loss(\mech(\ds'),\func(\ds')) \right] .
	\end{equation*}
\end{lemma}
\begin{proof}
	The definition of $\diffp$-DP implies that for any $t \in \range$
	\begin{align*}
		\E \left[ \loss(\mech(\ds),t) \right] 
		& = \sum_{s \in \range} \P(\mech(\ds) = s) \loss(s, t) \\
		& \le \sum_{s \in \range}
		e^{\DSdist(\ds,\ds') \diffp}\P(\mech(\ds') = s) \loss(s, t)
		= e^{\DSdist(\ds,\ds') \diffp} \E \left[\loss(\mech(\ds'),t) \right]. 
	\end{align*}
	Now, we use that $\mech$ is $\loss$-unbiased and the previous
	inequality to get
	\begin{align*}
		\frac{\E \left[ \loss(\mech(\ds),\func(\ds))\right]} {\E \left[ \loss(\mech(\ds'),\func(\ds'))\right]}
		= \frac{\E \left[ \loss(\mech(\ds),\func(\ds))\right]} {\E \left[ \loss(\mech(\ds),\func(\ds'))\right]}
		\frac{\E \left[ \loss(\mech(\ds),\func(\ds'))\right]} {\E \left[ \loss(\mech(\ds'),\func(\ds'))\right]}
		\le e^{\DSdist(\ds,\ds') \diffp}
	\end{align*}
	as claimed.
\end{proof}

We use the previous lemma to prove the lower bound
on unbiased mechanisms in
Theorem~\ref{thm:lower-bound-general-loss}.
\begin{lemma}
	\label{lemma:lower-bound-general-loss-unbiased}
	%Let $\loss(t_1,t_2) = \ell(|t_1-t_2|)$ for a
	%non-decreasing function $\ell : \R^+ \to \R$.
	Let $\mech$ be $\diffp$-DP and $\loss$-unbiased.
	Then for $k \ge 1$ and any $\sample \in \domain^\dssize$,
	\begin{equation*}
		\E\left[\loss(\mech(\ds),\func(\ds))\right]
		\ge
		\frac{\loss\left({\modcont_\func(\ds;k)}/{2}\right)}{e^{2k\diffp}
			+ 1}.
	\end{equation*}
\end{lemma}
\begin{proof}
	Let $\alpha$ be such that
	$\E[\loss(\mech(\ds),\func(\ds))] \le \alpha$. We shall prove
	a lower bound on $\alpha$.
	Using Markov's inequality, we have
	\begin{align*}
		\P\left(\dt(\mech(\ds),\func(\ds))
		\ge   \frac{\modcont_\func(\ds;k)}{2} \right) 
		\le
		\P\big( \loss(\mech(\ds),\func(\ds))
		\ge \ell({\modcont_\func(\ds;k)}/{2}) \big) 
		% & \le \P\left(\loss(\mech(\ds),\func(\ds))  \ge   \frac{1}{2} \ell({\modcont_\func(\ds;k)}) \right)
		\le \frac{\alpha}{\ell\left({\modcont_\func(\ds;k)}/{2}\right)},
	\end{align*}
	where the first inequality holds since $\ell$ is non decreasing.  The
	definition of $\modcont_{\func}(\ds;k)$ implies the existence of $\ds' \in
	\domain^\dssize$ such that $\DSdist(\ds,\ds') = k$ and
	$\dt(\func(\ds),\func(\ds'))= \modcont_{\func}(\ds;k)$.  Now, we prove
	that $\P(\dt(\mech(\ds),\func(\ds')) \ge \frac{\modcont_\func(\ds;k)}{2}
	)$ is also small.  Lemma~\ref{lemma:neigh-loss-paper} implies that the
	loss of the mechanism $\mech$ for $\ds'$ is also small,
	i.e. $\E[\loss(\mech(\ds'),\func(\ds'))] \le e^{k
		\diffp}\alpha$.  Using Markov's inequality for $\ds'$ and the definition
	of $\diffp$-DP, we have
	\begin{align*}
		\P \left(\dt(\mech(\ds),\func(\ds')) \ge  \frac{\modcont_\func(\ds;k)}{2}
		\right) 
		& \le e^{k \diffp} \P\left(\dt(\mech(\ds'),\func(\ds')) \ge
		\frac{\modcont_\func(\ds;k)}{2} \right)   \\
		& \le e^{k \diffp} \P\left(\loss(\mech(\ds'),\func(\ds'))
		\ge   \ell\left({\modcont_\func(\ds;k)}/{2}\right) \right) \\
		& \le \frac{ e^{2k\diffp} \alpha}{ \ell\left({\modcont_\func(\ds;k)}/{2}\right)}.
	\end{align*}
	As $\dt(\func(\ds),\func(\ds')) =   \modcont_\func(\ds;k)$, we have
	\begin{align*}
		1 
		& \ge \P\left(\dt(\mech(\ds),\func(\ds))<  \frac{\modcont_\func(\ds;k)}{2}\right) + 
		\P\left(\dt(\mech(\ds),\func(\ds')) < \frac{\modcont_\func(\ds;k)}{2}\right) \\
		& \ge 2 - \frac{\alpha}{ \ell\left({\modcont_\func(\ds;k)}/{2}\right)} - 		
		\frac{ e^{2k\diffp} \alpha}{ \ell\left({\modcont_\func(\ds;k)}/{2}\right)}.
	\end{align*}
	The lemma follows by rearranging terms in the last inequality.
\end{proof}

Finally, we have the minimax bound we claim in
Theorem~\ref{thm:lower-bound-general-loss}:
\begin{lemma}
	\label{lemma:lower-bound-general-loss-minimax}
	Let $\mech$ be $\diffp$-DP. Then for $k \ge 1$,
	\begin{equation*}
		\sup_{\ds \in \domain^\dssize} \E\left[\loss(\mech(\ds),\func(\ds))\right] 
		\ge \sup_{\ds \in \domain^\dssize}
		\frac{ \ell\left({\modcont_\func(\ds;k)}/{2}\right)}{e^{k\diffp} + 1}.
	\end{equation*}
\end{lemma}
\begin{proof}
	The proof of the worst-case bound is nearly identical
	to that of Lemma~\ref{lemma:lower-bound-general-loss-unbiased}.
	Noting that instead of using $\E[\loss(\mechanism(\ds'), \func(\ds'))]
	\le e^{k \diffp} \alpha$ we may use
	$\E[\loss(\mechanism(\ds'), \func(\ds'))]
	\le \alpha$ as we seek a uniform bound,
	we repeat the argument \emph{mutatis-mutandis} to obtain that
	if $\alpha \ge \sup_{\ds \in \domain^\dssize}
	\E[\loss(\mech(\ds), \func(\ds))]$, then
	\begin{equation*}
		1 \ge 2 - \frac{\alpha}{\ell(\modcont_\func(\ds; k) / 2)}
		- \frac{ e^{k \diffp} \alpha}{
			\ell(\modcont_\func(\ds; k) / 2)}.
	\end{equation*}
	Rearranging gives the result.
\end{proof}

% -*- Mode: latex -*- %

\subsection{Proof of Theorem~\ref{thm:lower-bound-local-minimax}}
\label{proof:thm-lower-bound-local-minimax}

We divide the proof into two parts; the first on lower bounds and the second
proving the (nearly) matching upper bounds on $\LMM$.

\subsubsection{Lower bounds}
We begin with a simple lemma, which upper bounds the variation distance
between private mechanisms.
\begin{lemma}
  \label{lemma:DP-bound-tv}
  Let $\mech$ be $(\diffp,\delta)$-differentially private. Then for 
  $\ds,\ds' \in \domain^\dssize$ with $\dham(\ds,\ds') \le k$,
  \begin{equation*}
    \tvnorm{\mech(\ds) - \mech(\ds')} 
    \le 1 -  e^{-k \diffp} + k e^{-\diffp} \delta
    \le k (\diffp + e^{-\diffp} \delta).
  \end{equation*}
\end{lemma}
\begin{proof}
  We have from group privacy~\cite[Theorem 2.2]{DworkRo14} that 
  for any measurable $S \subset \rangemech$,
  \begin{equation*}
    \P(\mech(\ds) \in S) \le e^{k \diffp} \P(\mech(\ds') \in S)
    + k e^{(k-1)\diffp} \delta
  \end{equation*}
  or, rearranging, $\P(\mech(\ds') \in S) \ge e^{-k \diffp} \P(\mech(\ds) \in S)
  - k e^{-\diffp} \delta$.
  Consequently we obtain
  \begin{align*}
    \tvnorm{\mech(\ds) - \mech(\ds')} 
    & = \sup_{S}  \P(\mech(\ds) \in S) - \P(\mech(\ds') \in S) \\
    & \le \sup_{S}  \P(\mech(\ds) \in S) 
    - e^{-k \diffp} \P(\mech(\ds) \in S)
    + k e^{-\diffp} \delta \\
    & \le 1 -  e^{-k \diffp} + k e^{-\diffp} \delta,
  \end{align*}
  as desired. The final inequality is simply that
  $e^t \ge 1 + t$ for all $t \in \R$.
\end{proof}

Similar upper bounds hold for
\renyi-differential privacy.
\begin{lemma}
  \label{lemma:renyi-bound-tv}
  Let $\renparam \ge 1$.
  Let $\mech$ be $(\renparam, 2 \diffp^2)$-\renyi-DP.
  Then for 
  any $\ds,\ds' \in \domain^\dssize$ with $\dham(\ds,\ds') \le k$,
  \begin{equation*}
    \tvnorm{\mech(\ds) - \mech(\ds')} 
    \le k \diffp.
  \end{equation*}
  Alternatively, if $\mech$ is $(\renparam, \diffp/2)$-\renyi-DP with
  $\renparam \ge 1 + 2 \diffp^{-1} \log \frac{1}{\delta}$, then
  \begin{equation*}
    \tvnorm{\mech(\ds) - \mech(\ds')} \le 1 - e^{-k \diffp} + k e^{-\diffp}
    \delta
    \le k(\diffp + e^{-\diffp} \delta).
  \end{equation*}
\end{lemma}
\begin{proof}
  The definition of 
  $(\renparam,\diffp^2)$-\renyi-DP implies that 
  $\drens{\mech(\ds^0)}{\mech(\ds^1)} \le \diffp^2$ for any
  neighboring datasets. As $\dham(x, x') \le k$, there exist
  datasets $x = x^0, x^1, \ldots, x^k = x' \in \mc{X}^n$ such that
  $\dham(x^i, x^{i+1}) \le 1$ for each $i$, and thus
  \begin{align*}
    \tvnorm{\mech(\ds) - \mech(\ds')}
    & \le \sum_{i = 0}^{k-1} \tvnorm{\mech(\ds^i) - \mech(\ds^{i+1})} \\
    & \stackrel{(i)}{\le}
    \sum_{i = 1}^{k-1} \sqrt{\dkl{\mech(\ds^i)}{\mech(\ds^{i+1})} / 2}
    \stackrel{(ii)}{\le}
    k \sqrt{\diffp^2},
  \end{align*}
  where inequality~$(i)$ is Pinsker's inequality 
  and~$(ii)$ follows because
  $\drens{\cdot}{\cdot}$ is decreasing in $\renparam$.

  For the second result, we recall \citet[Prop.~3]{Mironov17}, which shows
  that an $(\renparam, \diffp)$-\renyi-DP mechanism is also $(\diffp +
  \frac{\log \delta^{-1}}{\renparam - 1}, \delta)$-DP for any $0 < \delta <
  1$. Thus, in the case that $\mech$ is $(\renparam, \diffp/2)$-\renyi-DP
  with $\renparam \ge 1 + 2 \diffp^{-1} \log \frac{1}{\delta}$,
  Lemma~\ref{lemma:DP-bound-tv}
  implies the second result.
\end{proof}

We can now lower bound the loss of private mechanisms for two instances.  We
adapt Le Cam's two-point method for lower
bounds~\cite[cf.][Ch.~5.2]{Duchi18}. For $\ds, \ds' \in \domain^\dssize$,
we denote
\begin{equation*}
  d_\loss(\ds,\ds') = \inf_{t \in \range} \{ \loss(\func(\ds),t) +
  \loss(\func(\ds'),t) \}.
\end{equation*}
Then we have the following bound.
\begin{lemma}
  \label{lemma:lecam-bound}
  %Let $\ds,\ds' \in \domain^\dssize$ with $\dham(\ds,\ds') \le k$.
  For any mechanism $\mech$, and instances 
  $\ds,\ds' \in \domain^\dssize$,
  \begin{equation*}
    \max_{\tilde \ds \in \{ \ds,\ds'\}} 
    \E\left[\loss(\mech(\tilde \ds),\func(\tilde \ds))\right] 
    \ge \frac{1}{4} d_\loss(\ds,\ds') \left(1 - \tvnorm{\mech(\ds) - \mech(\ds')}\right).
  \end{equation*}
\end{lemma}
\begin{proof}
  Let $S_\ds = \{t: \loss(\func(\ds),t) < \frac{d_\loss(\ds,\ds')}{2} \}$, so 
  we have that $\loss(\func(\ds),t) \ge  \frac{d_\loss(\ds,\ds')}{2}$
  for $t \notin S_\ds$. Further, the definition of $d_\loss$ implies that
  $\loss(\func(\ds'),t) \ge \frac{d_\loss(\ds,\ds')}{2} $ for $t\in S_{\ds}$. 
  Hence, we have
  \begin{align*}
    \max_{\tilde \ds \in \{ \ds,\ds'\}} 
    \E\left[\loss(\mech(\tilde \ds),\func(\tilde \ds))\right]
    & \ge \frac{1}{2}\E\left[\loss(\mech(\ds),\func(\ds))\right]
    + \frac{1}{2}\E\left[\loss(\mech(\ds'),\func(\ds'))\right] \\
    & \ge \frac{d_\loss(\ds,\ds')}{4} \left(\P(\mech(\ds) \notin S_{\ds}) 
    + \P(\mech(\ds') \in S_{\ds})  \right) \\
    & = \frac{d_\loss(\ds,\ds')}{4} \left(1 - \P(\mech(\ds) \in S_{\ds}) 
    + \P(\mech(\ds') \in S_{\ds})  \right) \\
    & \ge \frac{d_\loss(\ds,\ds')}{4} 
    \left(1 - \tvnorm{\mech(\ds) - \mech(\ds')} \right)
  \end{align*}
  by definition of the variation distance.
\end{proof}

In our setting, $\loss(s,t) = \ell(\dt(s,t))$
for a non-decreasing function $\ell : \R_+ \to \R_+$, so that
recalling the definition~\eqref{eqn:modcont-def} of
the modulus $\modcont_\func(\ds; k) =
\sup_{\ds'} \{\dt(\func(\ds), \func(\ds')) : \dham(\ds, \ds') \le k\}$,
for the $x' \in \mc{X}^n$ attaining
$\dt(\func(\ds), \func(\ds')) = \modcont_\func(\ds; k)$ (or one arbitrarily
close to obtaining it), we have
\begin{equation*}
  d_\loss(\ds, \ds')
  = \inf_t \left\{\ell(\dt(\func(\ds), t))
  + \ell(\dt(\func(\ds'), s))\right\}
  \ge \ell\left(\frac{\modcont_\func(\ds; k)}{2}\right)
\end{equation*}
because for any $t$, at least one of $\dt(\func(\ds), t) \ge
\modcont_\func(\ds; k)/2$ or $\dt(\func(\ds'), t) \ge \modcont_\func(\ds; k)
/ 2$.  Lemma~\ref{lemma:lecam-bound} then implies that for
any $k \in \N$, there exists $x' \in \mc{X}^n$ with
$\dham(x, x') \le k$ such that for any mechanism $\mech$, we have
\begin{equation}
  \label{eqn:lecam-bound-tv}
  \max_{\tilde{\ds} \in \{\ds, \ds'\}}
  \E\left[\loss(\mech(\tilde{\ds}), \func(\tilde{\ds}))\right]
  \ge \frac{1}{4}
  \ell\left(\half \modcont_\func(\ds;k)\right)
  \left(1 - \tvnorm{\mech(\ds) - \mech(\ds')}\right).
\end{equation}

The lower bounds in
Theorem~\ref{thm:lower-bound-local-minimax} now each follow
from inequality~\eqref{eqn:lecam-bound-tv} and our
bounds on the total
variation distance of each family of mechanisms
via Lemmas~\ref{lemma:DP-bound-tv} and~\ref{lemma:renyi-bound-tv}.
The lower bound~\eqref{eqn:dp-local-minimax} follows
from inequality~\eqref{eqn:lecam-bound-tv} with Lemma~\ref{lemma:DP-bound-tv}.
The lower bound~\eqref{eqn:adp-local-minimax}
for the family of $(\diffp,\delta)$-differentially
private mechanisms and
$(\renparam, \diffp/2)$-\renyi-DP mechanisms with
$\renparam \ge 1 + 2 \diffp^{-1} \log \frac{1}{\delta}$
follows from Lemma~\ref{lemma:DP-bound-tv} and~\ref{lemma:renyi-bound-tv},
as 
as $k \le \min \{\frac{\log{\frac{1}{\delta}}}{\diffp},
\frac{1}{\sqrt{\delta}} \}$ implies that
$1 - \tvnorm{\mech(\ds) - \mech(\ds')} 
\ge e^{-k \diffp} - k e^{-\diffp} \delta 
\ge \frac{1}{2} e^{-k \diffp}$.
The lower bound~\eqref{eqn:renyi-local-minimax}
for \renyi-DP follows by setting $k = \frac{1}{2 \diffp}$
in Lemma~\ref{lemma:renyi-bound-tv}.

\subsubsection{Upper bounds}

For the upper bounds, given samples $\ds^0, \ds^1 \in \mc{X}^n$ with
$\dham(\ds^0, \ds^1) \le k$, we may construct a mechanism guaranteed to be
differentially private. We set
\begin{equation*}
  \mech\subopt(\ds) = \begin{cases} \func(x^0) & \mbox{with~probability~}
    \frac{e^{-\diffp \dham(\ds, \ds^0) / 2}}{
      e^{-\diffp \dham(\ds, \ds^0) / 2} +
      e^{-\diffp \dham(\ds, \ds^1) / 2}}
    \\
    \func(x^1) & \mbox{with~probability~}
    \frac{e^{-\diffp \dham(\ds, \ds^1) / 2}}{
      e^{-\diffp \dham(\ds, \ds^0) / 2} +
      e^{-\diffp \dham(\ds, \ds^1) / 2}}
    %% \frac{\exp(-\diffp \dham(\ds, \ds^1) / 2)}{
    %%   \exp(-\diffp \dham(\ds, \ds^0) / 2) +
    %%   \exp(-\diffp \dham(\ds, \ds^1) / 2)}.
  \end{cases}
\end{equation*}
The mechanism $\mech\subopt$ is evidently $\diffp$-differentially private, and
we have
\begin{equation*}
  \E\left[L(\mech\subopt(\ds^0), \func(\ds^0))\right]
  = \E\left[L(\mech\subopt(\ds^1), \func(\ds^1))\right]
  = \frac{1}{1 + e^{\diffp \dham(\ds^0, \ds^1) / 2}}
  \ell\left(\dt(\func(\ds^0), \func(\ds^1))\right).
\end{equation*}
Fixing $\ds^0$ and taking a supremum over $\ds^1$ gives upper bounds nearly
matching each of the lower bounds~\eqref{eqn:local-minimax-calculation}. The
first~\eqref{eqn:dp-local-minimax} is immediate. For the
second~\eqref{eqn:adp-local-minimax}, any $\diffp$-DP mechanism is also
$(\diffp, \delta)$-DP, and moreover, satisfies $(\renparam,
\diffp)$-differential privacy for all $\renparam \in [1, \infty]$. For the
final bound~\eqref{eqn:renyi-local-minimax}, if $Q(\cdot \mid \ds)$ denotes
the distribution of the mechanism $\mech\subopt$, we have for any
neighboring samples $\ds, \ds'$ that
\begin{equation*}
  \int \left(\frac{dQ(z \mid \ds)}{dQ(z \mid \ds')}\right)^2
  dQ(z \mid \ds')
  = 1 + \int\left(\frac{dQ(z \mid \ds)}{dQ(z \mid \ds')} - 1\right)^2
  dQ(z \mid \ds')
  \le 1 + (e^\diffp - 1)^2.
\end{equation*}
In particular, for $\renparam = 2$, we have
$\dren{\mech\subopt(\ds)}{\mech\subopt(\ds')} \le \log(1 + (e^\diffp - 1)^2)
\le 2 \min\{\diffp^2, \diffp\}$.

\subsection{Proof of Proposition~\ref{proposition:super-efficiency}}
\label{sec:proof-super-efficiency}

\newcommand{\chiaffi}[2]{\rho \left({#1} |\!| {#2}\right)}

Our proof is similar to the proof of Proposition 2
in~\cite{DuchiRu18b}. We start with some notation.
For distributions $P_0,P_1$ we define the $\chi^2$-affinity
$\chiaffi{P_1}{P_0} \defeq \E_{P_1}[ \frac{dP_1}{dP_0}]
= \dchi{P_1}{P_0} + 1$.
We have the following constrained risk inequality.
\begin{lemma}[\citet{DuchiRu18}, Corollary 1]
  \label{lemma:duchru-ineq}
  Assume $\ell: \R_+ \to \R_+$ is a non-decreasing function and 
  let $\ds,\ds' \in \domain^\dssize$. 
  Define $\Delta = \ell(\frac{1}{2}|\func(\ds) - \func(\ds')|)$.
  If $\E[|\mech(\ds) - \func(\ds)|] \le \gamma$ then
  \begin{equation*}
    \E[\ell(|\mech(\ds') - \func(\ds')|)] 
    \ge \left[ \Delta^{1/2} - 
      (\chiaffi{\mech(\ds')}{\mech(\ds)}\gamma)^{1/2} \right]_+^2 .
  \end{equation*}
\end{lemma}

We can now use Lemma~\ref{lemma:duchru-ineq} to 
prove Proposition~\ref{proposition:super-efficiency}.
Denote $\Delta_\ds^k  = \ell(\frac{1}{2} \modcont_\func(\ds;k) ) $
%\sup_{\ds'} \{\ell(\frac{1}{2}|\func(\ds) - \func(\ds')|): \dham(\ds,\ds') \le k \}$ 
and let
$\ds'$ be such that 
$\modcont_\func(\ds;k) = |\func(\ds) - \func(\ds')| $.
Since $\mech$ is $\diffp$-DP, we have
that $\chiaffi{\mech(\ds')}{\mech(\ds)} \le e^{k \diffp}$.
As $\E[\ell(|\mech(\ds) - \func(\ds)|)] \le \gamma \Delta_\ds^{1/\diffp} $,
Lemma~\ref{lemma:duchru-ineq} implies that
for any $k \in \N$, if we choose $\ds'$ such that
$\dham(\ds, \ds') \le k$ and
$|\func(\ds') - \func(\ds)| = \modcont_\func(\ds; k)$,
then
\begin{equation*}
  \E[\ell(|\mech(\ds') - \func(\ds')|)] 
  \ge \hinge{\sqrt{\Delta_\ds^k} - 
    (e^{k \diffp } \Delta_\ds^{1/\diffp} \gamma)^{1/2} }^2 .
\end{equation*}
Taking $k = \diffp^{-1} \log \frac{1}{2\gamma}$, we obtain
\begin{align}
  \label{eq:main-lb}
  \E[\ell(|\mech(\ds') - \func(\ds')|)] 
  \ge \left[ \sqrt{\Delta_\ds^{(1/\diffp) \log(1/2\gamma)}} -
    \sqrt{\Delta_\ds^{1/\diffp} / 2} \right]_+^2
  \ge \frac{2}{7} \Delta_\ds^{(1/\diffp) \log\frac{1}{2\gamma}},
\end{align}
where the second inequality follows because
$1 - 1/\sqrt{2} > \frac{2}{7}$ and
$\Delta_\ds^k$ is increasing in $k$ (we have assumed
$\gamma \le e^{-1} / 2$).
Now we prove that if $\dham(\ds,\ds') \le k$ then
\begin{equation}
  \label{eq:mod-relation}
  2 \modcont_\func(\ds;2k) \ge \modcont_\func(\ds';k).
\end{equation}
Assume w.l.o.g.\ that $\func(\ds) \le \func(\ds')$,
and let $\tilde \ds$ be such that $\dham(\ds',\tilde \ds) \le k$
and $\modcont_\func(\ds';k) = |\func(\tilde \ds) - \func(\ds')| $.
Note that $\dham(\ds,\tilde \ds) \le 2k$. We consider two cases:
(i) that $\func(\tilde{\ds}) \ge \func(\ds)$ and (ii)
that $\func(\tilde{\ds}) \le \func(\ds)$. In the former case (i),
if $\func(\tilde{\ds}) \ge \func(\ds')$ then
$\modcont_\func(\ds; 2k)
\ge \func(\tilde{\ds}) - \func(\ds)
\ge \func(\tilde{\ds}) - \func(\ds')
= \modcont_\func(\ds; k)$, while if
$\func(\tilde{\ds}) \le \func(\ds')$ then
$\func(\tilde{\ds}) \in [\func(\ds), \func(\ds')]$, so that
we must have $\func(\tilde{\ds}) = \func(\ds)$ as $\dham(\ds, \ds') \le k$,
and again, inequality~\eqref{eq:mod-relation} holds.
In case (ii) that $\func(\tilde \ds) < \func(\ds)$,
we have
\begin{equation*}
  2 \modcont(\ds; 2k)
  \ge \func(\ds') - \func(\ds) + \func(\ds) - \func(\tilde{\ds})
  = \func(\ds') - \func(\tilde{\ds})
  = \modcont_\func(\ds; k),
\end{equation*}
as desired.  As $\ell$ is non-decreasing, inequalities~\eqref{eq:main-lb}
and~\eqref{eq:mod-relation} imply
\begin{align*}
  \E[\ell(|\mech(\ds') - \func(\ds')|)] 
  \ge \frac{1}{4} \ell\left( \frac{1}{4} \modcont_\func\left(\ds';
  \frac{\log(1 / 2\gamma)}{2 \diffp}\right)
  \right).
\end{align*}

% -*- Mode: latex -*- %

%\section{Proofs for
%  Section~\ref{sec:upper-bounds}}
%\label{sec:DP-basic-proofs}
% In this section, we prove the results of Section~\ref{sec:theory-discrete}.

\section{Proofs of general upper bounds (Section~\ref{sec:upper-bounds})}

\iftoggle{proofs}{}{%
  \subsection{Proof of Lemma~\ref{lemma:mech-discrete-privacy}}
  \label{sec:proof-mech-discrete-privacy}
  % -*- Mode: latex -*- %

The first claim is immediate, as  $\ds \mapsto \invmodcont_\func(\ds; t)$
is $1$-Lipschitz with respect to the Hamming metric on $\mc{X}^n$.
The binary case is more subtle; we assume w.l.o.g.\ that $\range =
\{0, 1\}$.  Let $\ell(\ds) = \invmodcont_\func(\ds;1 - \func(\ds))$ be
the distance to the closest instance $\ds'$ with $\func(\ds') \neq
\func(\ds)$.  Then
\begin{equation*}
  \P(\mechdisc(\ds) = t) =
  \begin{cases}
    \frac{e^{\ell(\ds) \diffp/2}}{e^{\ell(\ds) \diffp/2} + 1} & \text{if } t = \func(\ds) \\
    \frac{1}{e^{\ell(\ds) \diffp/2} + 1} & \text{otherwise}. \\
  \end{cases}
\end{equation*}
For any neighboring instances $\ds,\ds' \in \domain^\dssize$, if $\func(\ds)
\neq \func(\ds')$, then we have
$\frac{\P(\mechdisc(\ds)=\func(\ds))}{\P(\mechdisc(\ds')=\func(\ds))} =
e^{\diffp/2}$.  Conversely, when $\func(\ds) = \func(\ds')$, we have
$|\ell(\ds) - \ell(\ds')| \le 1$, and therefore
\begin{equation*}
  \frac{\P(\mechdisc(\ds)=\func(\ds))}{\P(\mechdisc(\ds')=\func(\ds))} 
  = \frac{e^{\ell(\ds) \diffp/2}}{e^{\ell(\ds') \diffp/2}}
  \frac{e^{\ell(\ds') \diffp/2} + 1}{e^{\ell(\ds) \diffp/2} + 1} 
  = \frac{e^{-\ell(\ds') \diffp/2} + 1}{e^{-\ell(\ds) \diffp/2} + 1}
  \le e^{\diffp/2}.
\end{equation*}
We also have
\begin{equation*}
  \frac{\P(\mechdisc(\ds)=1 - \func(\ds))}{\P(\mechdisc(\ds')=1 - \func(\ds))} 
  = \frac{e^{\ell(\ds') \diffp/2} + 1}{e^{\ell(\ds) \diffp/2} + 1} 
  \le e^{\diffp/2}.
\end{equation*}

}

% -*- Mode: latex -*- %

\subsection{Proof of Theorem~\ref{thm:upper-bound-general-loss}}
\label{sec:proof-upper-bound-general-loss}

We begin with the following lemma, which we use in the proofs of later
results as well.
\begin{lemma}
  \label{lemma:upper-bound-general-loss-initial}
  Let $\func: \domain^\dssize \to \range$ and $T\ge1$ an integer.
  For any $\ds \in \domain^\dssize$, the 
  mechanism $\mechdisc$~\eqref{mech:discrete} has
  \begin{equation*}
    \E\left[\loss(\mechdisc(\ds),\func(\ds)) \right] 
    \le \ell(\modcont_\func(\ds;T))  +  \frac{2 \ell(\dt^\star) \card(\range)}{
      \diffp}  e^{-T \diffp/2}.
  \end{equation*}
\end{lemma}
\begin{proof}
  First, we define the slice of $\range$ containing those $t \in \range$
  satisfying $\invmodcont_\func(\ds;t)=k$ by
  \begin{equation}
    \label{eqn:slice}
    S_\ds^k \defeq \{t \in \range:
    \invmodcont_\func(\ds;t)=k\}.
  \end{equation}
  The slices $S_\ds^k$ are disjoint for varying $k$, and
  we have $\card(S_\ds^k) \le \card(\range)$ for all $k\ge1$. Therefore 
  \begin{align*}
    \P \left(\mechdisc(\ds) \in S_\ds^k \right) 
    \le e^{-k \diffp/2} \P \left(\mechdisc(\ds) = \func(\ds) \right)
    \card(S_\ds^k)
    \le  e^{-k \diffp/2} \card(\range).
  \end{align*}
  By construction, for $t \in S_\ds^k$ we have $\dt(t,\func(\ds)) \le
  \modcont_\func(\ds;k)$, so that $\loss(t,\func(\ds)) \le
  \ell(\modcont_\func(\ds;k))$.  We use the previous inequality to get
  \begin{align*}
    \E\left[\loss(\mechdisc(\ds),\func(\ds)) \right] 
    & = \sum_{t \neq \func(\ds)} \P(\mechdisc(\ds)=t) \loss(t,\func(\ds)) \\
    & \le \sum_{k=1}^{n} \P \left(\mechdisc(\ds) \in S_\ds^k \right)
    \ell(\modcont_\func(\ds;k)) \\
    %% & \le \ell(\modcont_\func(\ds;T)) \sum_{k=1}^{T} \P \left(\mechdisc(\ds) \in S_\ds^k \right)
    %% +  \sum_{k=T+1}^{n} \P \left(\mechdisc(\ds) \in S_\ds^k \right) \ell(\modcont_\func(\ds;k)) \\
    %	& \le \Delta_D^T +  \P(M(D)=f(D)) \Delta^2 \sum_{\ell=T+1}^{n} \ell e^{-\ell \delta} \\
    %	& \le \Delta_D^T + \Delta^2 e^{-(T-1s)\diffp} \left(\frac{1}{(e^{\diffp}-1)^2} + \frac{T}{e^{\diffp}-1} \right)
    & \le \ell(\modcont_\func(\ds;T)) +  
    \card(\range)
    \sum_{k=T+1}^{n} e^{-k \diffp /2} \ell(\modcont_\func(\ds;k)) \\
    & \le \ell(\modcont_\func(\ds;T)) + \frac{2 \ell(\dt^\star)
      \card(\range) }{\diffp}  e^{-T \diffp/2},
  \end{align*}
  where the last inequality follows as $\modcont_\func(\ds;k) \le \dt^\star$
  and $\sum_{i=0}^\infty e^{-i \diffp} = \frac{e^\diffp}{e^\diffp-1} \le \frac{e^\diffp}{\diffp} $.
\end{proof}

Setting $T=\frac{2}{\diffp} (\log\gamma^{-1} + \log\frac{2 \ell(\dt^\star)
  \card(\range)}{\diffp})$ in
Lemma~\ref{lemma:upper-bound-general-loss-initial}, we get the first part of
Theorem~\ref{thm:upper-bound-general-loss}.  We prove the second claim. If
$\wt{\diffp} = 2 \diffp \log \frac{2 \ell(\dt\opt) \card(\range)}{\gamma
  \diffp}$, then
\begin{equation*}
  \frac{2}{\wt{\diffp}}
  \log \frac{2 \ell(\dt\opt) \card(\range)}{\gamma \wt{\diffp}}
  = \frac{1}{\diffp}
  \cdot \frac{1}{\log \frac{2 \ell(\dt\opt) \card(\range)}{\gamma \diffp}}
  \bigg[\log \frac{2 \ell(\dt\opt) \card(\range)}{\gamma \diffp}
    - \underbrace{
    \log \left(2 \log \frac{2 \ell(\dt\opt) \card(\range)}{\gamma \diffp}
    \right)}_{\ge 0}\bigg],
\end{equation*}
so that $\E[\loss(\wt{\mech}_\disc(\ds), \func(\ds))]
\le \gamma + \ell(\modcont_\func(\ds; \diffp^{-1}))$ as claimed.

%% Now, we prove the second part.
%% Using the same notation and steps as in the proof of
%% Lemma~\ref{lemma:upper-bound-general-loss-initial} to get for $k \ge
%% 1/\diffp$
%% \begin{align*}
%%   \P \left(\wt{\mech}_\disc(\ds) \in S_\ds^k \right) 
%%   \le e^{-k \tilde \diffp /2} \P \left(\mechdisc(\ds) = \func(\ds) \right)
%%   \card(S_\ds^k)
%%   \le \diffp e^{-3 k \diffp} .
%% \end{align*}
%% Therefore we have
%% \begin{align*}
%%   \P \left(\wt{\mech}_\disc(\ds) \in
%%   \cup_{(k-1)/\diffp< l \le k/\diffp} S_\ds^l \right)
%%   = \sum_{l > (k-1)/\diffp}^{k /\diffp} \P \left(\wt{\mech}_\disc(\ds)
%%   \in S_\ds^l \right) 
%%   \le e^{-3 (k-1) \diffp}
%% \end{align*}
%% because $\sum_{l = (k-1) / \diffp}^{k/\diffp} e^{-3l \diffp}
%% \le \sum_{l = (k-1) / \diffp}^{k/\diffp}  \diffp e^{-3 l \diffp}
%%  \le e^{-3 (k-1) \diffp}$
%% and so
%% \begin{align*}
%%   \E\left[\loss(\mechdisc(\ds),\func(\ds)) \right] 
%%   %% & = \sum_{t \neq \func(\ds)} \P(\mechdisc(\ds)=t) \ell(\dt(t,\func(\ds))) \\
%%   & \le \sum_{k=1}^{n} \P \left(\mechdisc(\ds) \in S_\ds^{k} \right) \ell(\modcont_\func(\ds;k)) \\
%%   & \le \ell(\modcont_\func(\ds;1/\diffp)) +  \sum_{k=2}^{n\diffp} 
%%   \P \left(\mechdisc(\ds) \in \cup_{(k-1)/\diffp< l \le k/\diffp} S_\ds^{l} \right)
%%   \ell(\modcont_\func(\ds;k/\diffp)) \\
%%   & \le \ell(\modcont_\func(\ds;1/\diffp)) + \sum_{k=2}^{n\diffp} e^{-3(k-1)} \ell(\modcont_\func(\ds;k/\diffp)) .
%% \end{align*}

% -*- Mode: latex -*- %

%\section{Proofs for general mechanisms (Section~\ref{sec:theory-continuous})}
%\label{sec:proofs-continuous}

\subsection{Proof of Theorem~\ref{thm:upper-bound-general-loss-continuous}}
\label{sec:proofs-theory-continuous}

We prove the result in a somewhat more general setting than claimed
in the theorem. We allow $\mc{T}$ to be a subset of a vector
space, and instead of Lebesgue measure
we assume that the
measure $\basemeasure$ approximates a 1-dimensional uniform measure
on $\range$, meaning that for the
unit ball $\ball \defeq \{t \in \range : \norm{t} \le 1\}$,
\begin{equation}
  \label{eqn:basemeasure-uniformity}
  \frac{\basemeasure(S)}{\basemeasure(t + \rho \ball)}
  %\le \left(\frac{\diam(S)}{\rho}\right)^d
  \le \frac{\diam(S)}{\rho}
  ~~ \mbox{for~all}~ S \subset \range
  ~ \mbox{and} ~ t \in \range.
\end{equation}
Certainly the Lebesgue measure on $\R$ satisfies this, but so too does
any discrete measure on equi-spaced points in $\R$.
We also assume the loss function $\loss$ satisfies
$\loss(s,t) = \ell(\norm{s - t})$ for a
non-decreasing function $\ell : \R_+ \to \R_+$,
and $\norm{s - t} \le K$ for all $s, t \in \range$.

We start with the following lemma, which proves the first part
of Theorem~\ref{thm:upper-bound-general-loss-continuous}.
\begin{lemma}
  \label{lemma:upper-bound-continuous}
  Let the conditions of Theorem~\ref{thm:upper-bound-general-loss-continuous}
  hold.
  There exists a numerical constant $c < \infty$ such that
  for any $T \in N$
  and $\ds \in \domain^\dssize$,
  \begin{equation*}
    \E\left[\loss(\mechcont(\ds), \func(\ds))\right] 
    \le \ell(\modcont_\func(\ds;T)) +
    C_\ell \rho
    %+ c \ell(K) \cdot n (1 + \diffp)
    %\exp\left(-T \diffp / 2 +  \log \frac{8 n \GS_\func}{\rho}
     + c  C_\ell \rho 
     \exp\left(-T \diffp / 2 +   \log \frac{n^2 \GS_\func^2}{\rho^2 \diffp}
    \right).
    %% \ell(K) \left(\frac{4}{\diffp} +
    %% \frac{8 \GS_\func T}{\rho \diffp^2} \right)   e^{-T \diffp/2} + C_\ell \rho.
  \end{equation*}
\end{lemma}

Before proving the lemma, we show how it gives the theorem.
Setting $T=\frac{2}{\diffp} [\log\frac{1}{\diffp} + 2\log \frac{\dssize
    \GS_\func}{\rho}]$ in Lemma~\ref{lemma:upper-bound-continuous} gives the
first part of Theorem~\ref{thm:upper-bound-general-loss-continuous}.
For the second claim, note that if
$\wt{\diffp} = 2 \diffp [\log\frac{1}{2\diffp}
  + 2 \log \frac{n \GS_\func}{\rho}]$, then
\begin{align*}
  \lefteqn{\frac{2}{\wt{\diffp}}
    \left[\log \frac{1}{\wt{\diffp}}
      + 2 \log \frac{n \GS_\func}{\rho}\right]} \\
  & = \frac{1}{\diffp
    (\log \frac{1}{2 \diffp} + 2 \log \frac{n \GS_\func}{\rho})}
  \left[\log\frac{1}{2 \diffp} + 2 \log \frac{n \GS_\func}{\rho}
    - \log \left(\log\frac{1}{2\diffp} + 2 \log \frac{n \GS_\func}{\rho}
    \right)\right]
  \le \frac{1}{\diffp},
\end{align*}
so that the first claim of the theorem gives the result.

\begin{proof}
  As in the proof of Lemma~\ref{lemma:upper-bound-general-loss-initial},
  we define the slices
  \begin{equation*}
    S_\ds^k = \{t \in \range : \smoothinvmodcont_\func(\ds;t) = k \},
  \end{equation*}
  so that the $S_\ds^k$ partition $\range$.  Since $\func$ has global
  sensitivity $\GS_\func$, we have $\diam(S_\ds^k) \le 2( k \GS_\func +
  \rho)$ for all $k\ge1$, by coupling the discretization $\rho$ with the
  fact that changing one element in $\ds$ changes $f(\ds)$ by at most
  $\GS_\func$.  The definition of $\mechcont$ and the
  uniformity~\eqref{eqn:basemeasure-uniformity} then imply
  \begin{align}
    \P \left(\mechcont(\ds) \in S_\ds^k \right) 
    = \frac{ \int_{S_\ds^k} e^{-\smoothinvmodcont_\func(\ds;t) \diffp/2}
    d\basemeasure(t)}{
      \int_{\range} e^{-\smoothinvmodcont_\func(\ds;s) \diffp/2} d\basemeasure(s)} 
    & \le e^{-k \diffp/2} \left(\frac{\diam(S_\ds^k)}{\rho}\right)
    \nonumber \\
    & \le 2 \left(k \frac{\GS_\func}{\rho} + 1\right) e^{-k \diffp/2} .
    \label{eqn:probability-bound-dimension}
  \end{align}
  The previous inequality implies
  \begin{align*}
    \E\left[\loss(\mechcont(\ds),\func(\ds))\right] 
    & = \int_{t \in \range} \pdf_{\mechcont(\ds)}(t) \loss(t,\func(\ds))
    d\basemeasure(t) \\
    & \le \sum_{k=1}^{\dssize} \P \left(\mechcont(\ds) \in S_\ds^k \right) \ell(\modcont_\func(\ds;k) + \rho) \\
    & \le  \ell(\modcont_\func(\ds;T))  \sum_{k=1}^{T} \P \left(\mech(\ds) \in S_\ds^k \right)
    +  \sum_{k=T+1}^{\dssize} \P \left(\mechcont(\ds) \in S_\ds^k \right) \ell(\modcont_\func(\ds;k))   + C_\ell\rho\\
%    & \le  \ell(\modcont_\func(\ds;T)) 
%    +  2 \ell(K) \sum_{k=T+1}^{\dssize} (k {\GS_\func}/{\rho} + 1)
%    e^{-k \diffp/2}   + C_\ell\rho.
     & \le  \ell(\modcont_\func(\ds;T)) 
     +  2 C_\ell \sum_{k=T+1}^{\dssize} (k {\GS_\func}/{\rho} + 1)
     k \GS_\func e^{-k \diffp/2}   + C_\ell\rho.  
  \end{align*}
  Now, we recall the incomplete gamma function $\Gamma(a, b) \defeq
  \int_a^\infty z^{b - 1} e^{-z} dz$, which satisfies $\Gamma(a, b) = O(1)
  a^{b - 1} e^{-a}$ as $a$ grows~\cite[Eq.~(1.5)]{BorweinCh09}, so that
  \begin{align*}
    \int_T^\infty \left(\frac{\GS_\func^2 }{\rho} z^2 \right) e^{-z \diffp / 2} dz
    =
    \frac{2}{\diffp} \left(\frac{4 \GS_\func^2}{\rho \diffp^2}\right)
    \int_{T \diffp / 2}^\infty u^2 e^{-u} du
    = O(1) \frac{2}{\diffp}
    \left(\frac{4 \GS_\func^2}{\rho \diffp^2}\right)
    (T \diffp / 2)^{2} e^{-T \diffp / 2}
    %% & \le \alpha^{-1} \left(\frac{a}{\alpha}\right)^d
    %% \left(\int_0^\infty u^{2d} e^{-u} du\right)^{1/2}
    %% \left(\int_{T\alpha}^\infty e^{-u} du\right)^{1/2} \\
    %% & = \alpha^{-1} \left(\frac{a}{\alpha}\right)^d
    %% \Gamma(2d + 1)^{1/2} e^{-T\alpha/2}.
  \end{align*}
  Returning to our string of inequalities, we obtain
  \begin{align*}
    \E[\loss(\mechcont(\sample), \func(\sample))]
    & \le \ell(\modcont_\func(\sample; T))
    + C_\ell \rho
    %+ O(1) \ell(K)
    %T (1 + \diffp)
    %\exp\left(-T \diffp / 2 + \log \frac{8 T \GS_\func}{\rho}\right),
     + O(1) C_\ell \rho
     \exp\left(-T \diffp / 2 + \log \frac{T^2 \GS_\func^2}{\rho^2 \diffp}\right),
  \end{align*}
  and noting that we necessarily have $T \le n$ gives the result.
\end{proof}

\section{Proofs for sample-monotone estimands
  (Section~\ref{sec:monotone-functions})}

\subsection{Proof of Observation~\ref{observation:cont-is-montonote}}
\label{sec:proof-cont-is-montonote}

Let $\ds \in \domain^\dssize$ and assume without loss of generality that
$f(\ds) \le s \le t$.  We need to prove that $ \invmodcont_\func(\ds;s) \le
\invmodcont_\func(\ds;t)$.  If $\invmodcont_\func(\ds;t) = \infty$ then we
are done.  Otherwise there exists $\ds'$ such that $\func(\ds') = t$ and
$\dham(\ds,\ds') = \invmodcont_\func(\ds;t)$.  We define the function
$g(\lambda) = \func(\lambda \ds + (1-\lambda)\ds')$ for $\lambda \in
[0,1]$. The function $g(\cdot)$ is continuous as $f(\cdot)$ is
continuous. We also know that $g(0) = \func(\ds)$ and $g(1) = \func(\ds') =
t$. The intermediate value theorem implies that there exists $0 \le
\lambda_s \le 1$ such that $g(\lambda_s) = s$.  Since $\domain$ is convex,
we get that $\ds_s = \lambda_s \ds + (1-\lambda_s)\ds' \in \domain^\dssize$
and $\func(\ds_s) = s$. To finish the proof, we note that $\dham(\ds,\ds_s)
\le \dham(\ds,\ds')$ immediately by the definition of $\ds_s$ as a convex
combination of $\ds$ and $\ds'$.

\subsection{Proof of Theorem~\ref{thm:loss-monotone}}
\label{sec:proof-thm-monotone-loss-optimal}

We first define the right and left moduli of continuity
$\modcont_\func(\ds;k^+) = \sup_{\ds'} \{f(\ds') - f(\ds) : \dham(\ds,\ds')
\le k\} $ and $\modcont_\func(\ds;k^-) = \sup_{\ds'} \{f(\ds) - f(\ds') :
\dham(\ds,\ds') \le k\} $.  Denote $\modcont_\func(\ds;k) = \modcont_k$,
$\modcont_\func(\ds;k^+) = \modcont_{k^+}$ and $\modcont_\func(\ds;k^-) =
\modcont_{k^-}$ for shorthand.  We have the following lemma.
\begin{lemma}
  \label{lemma:montone-loss}
  Let $\func : \domain^\dssize \to \R $ be monotone. Then
  \begin{align*}
    \E\left[|\mech_\cont(\ds) - \func(\ds)|^p \right]
    & \le  2^{p-1} \rho 
    + 2^{p-1} \frac{\sum_{k=1}^{n}  \modcont_{k^+}^{p} (\modcont_{k^+} - \modcont_{{k-1}^+} ) e^{-k \diffp/2}}
    {\rho + \sum_{k=1}^{n} \left( (\modcont_{k^+} - \modcont_{{k-1}^+})
      + (\modcont_{k^-} - \modcont_{{k-1}^-}) \right) 	e^{-k \diffp/2}} \\
    & \qquad + 2^{p-1} \frac{\sum_{k=1}^{n}  \modcont_{k^-}^{p} (\modcont_{k^-} - \modcont_{{k-1}^-} ) e^{-k \diffp/2}}
    {\rho + \sum_{k=1}^{n}   \left( (\modcont_{k^+} - \modcont_{{k-1}^+})
      + (\modcont_{k^-} - \modcont_{{k-1}^-}) \right) e^{-k \diffp/2}}.
  \end{align*}
\end{lemma}
\begin{proof}
  Let $S_{\ds}^{k^+} = \{t>\func(\ds) : \smoothinvmodcont_\func(\ds;t) = k \}$
  and $S_{\ds}^{k^-} = \{t<\func(\ds) : \smoothinvmodcont_\func(\ds;t) = k \}$. 
  Clearly we have $\rho \le \diam(S_{\ds}^0)$ 
  and $ \diam(S_{\ds}^{k^+}) = \modcont_{k^+} - \modcont_{{k-1}^+}$
  and $\diam(S_{\ds}^{k^-})=\modcont_{k^-} - \modcont_{{k-1}^-}$ for $k \ge 1$.
  As $\func$ is monotone,
  we have that for $k \ge 1$,
  \begin{equation}
    \P \left( \mech(\ds) \in S_{\ds}^{k^+} \right) \le
    \frac{(\modcont_{k^+} - \modcont_{{k-1}^+}) e^{-k \diffp/2}}{
      \rho + \sum_{k=1}^{n}  \left( (\modcont_{k^+} - \modcont_{{k-1}^+})
      + (\modcont_{k^-} - \modcont_{{k-1}^-}) \right) e^{-k \diffp/2}}.
    \label{eqn:f-monotone-slice-prob}
  \end{equation}
  As identical bound holds for $S_{\ds}^{k^-}$ except that we swap
  $\modcont_{k^+}$ with $\modcont_{k^-}$. Thus
  \begin{align*}
    \lefteqn{\E\left[|\mech_\cont(\ds) - \func(\ds)|^p \right]} \\
    & \le \P \left( \mech(\ds) \in S_{\ds}^0 \right)  \rho^p + 
    \sum_{k=1}^{n} \P \left( \mech(\ds) \in S_{\ds}^{k^+} \right) (\modcont_{k^+} + \rho)^p
    + \sum_{k=1}^{n} \P \left( \mech(\ds) \in S_{\ds}^{k^-} \right) (\modcont_{k^-} + \rho)^p \\
    & = 2^{p-1} \rho^p 
    + 2^{p-1} \sum_{k=1}^{n} \P \left( \mech(\ds) \in S_{\ds}^{k^+} \right) \modcont_{k^+}^{p} 
    + 2^{p-1} \sum_{k=1}^{n} \P \left( \mech(\ds) \in S_{\ds}^{k^-} \right) \modcont_{k^-}^{p}
  \end{align*}
  %%   & \le 2^{p-1} \rho 
  %%   + 2^{p-1} \frac{\sum_{k=1}^{n}  \modcont_{k^+}^{p} (\modcont_{k^+} - \modcont_{{k-1}^+} ) e^{-k \diffp/2}}
  %%   {\rho + \sum_{k=1}^{n}  \left( (\modcont_{k^+} - \modcont_{{k-1}^+})
  %%     + (\modcont_{k^-} - \modcont_{{k-1}^-}) \right) e^{-k \diffp/2}} \nonumber \\
  %%   & \qquad + 2^{p-1} \frac{\sum_{k=1}^{n}  \modcont_{k^-}^{p} (\modcont_{k^-} - \modcont_{{k-1}^-} ) e^{-k \diffp/2}}
  %%   {\rho + \sum_{k=1}^{n}  \left( (\modcont_{k^+} - \modcont_{{k-1}^+})
  %%     + (\modcont_{k^-} - \modcont_{{k-1}^-}) \right) e^{-k \diffp/2}} \nonumber.
  %%   %		& \le \rho + \frac{\sum_{k=1}^{K_0}  \modcont_k (\modcont_k - \modcont_{k-1}) e^{-k \diffp}}
  %%   %		{\rho + \sum_{k=1}^{n} (\modcont_k - \modcont_{k-1}) e^{-k \diffp}}
  %%   %		+ \frac{1}{\rho} \sum_{k=K_0 + 1}^{n} \modcont_k (\modcont_k - \modcont_{k-1}) e^{-k \diffp} \nonumber. 
  %% \end{align}
  Substituting the bound~\eqref{eqn:f-monotone-slice-prob}
  for $\P(\mech(\ds) \in S_{\ds}^{k^+})$ and the symmetric
  bound for $\P(\mech(\ds) \in S_{\ds}^{k^-})$ gives the lemma.
\end{proof}

We can now use Lemma~\ref{lemma:montone-loss} to prove
the theorem. Indeed, we have 
\begin{align*}
	\sum_{k=1}^{n} \left( (\modcont_{k^+} - \modcont_{{k-1}^+})
	+ (\modcont_{k^-} - \modcont_{{k-1}^-})\right) e^{-k \diffp/2} 
			& \ge (1 - e^{-\diffp/2}) \sum_{k=1}^{n} (\modcont_{k^+} + \modcont_{k^-})  e^{-k \diffp/2} \\
		 & \ge (1 - e^{-\diffp/2}) \sum_{k=1}^{n} \modcont_k e^{-k \diffp/2}.
\end{align*}
We also have
\begin{align*}
  \sum_{k=1}^{n}  \modcont_{k^+}^{p} (\modcont_{k^+} - \modcont_{{k-1}^+}) e^{-k \diffp/2}
  & \le \sum_{k=1}^{n}  (\modcont^{p+1}_{k^+} - \modcont^{p+1}_{{k-1}^+}) e^{-k \diffp/2} \\
  & = (1 - e^{-\diffp/2}) \sum_{k=1}^{n}  \modcont^{p+1}_{k^+} e^{-k \diffp/2}
  +  \modcont^{p+1}_{n^+} e^{-(n+1) \diffp/2} \\
  & \le (1 - e^{-\diffp/2}) \sum_{k=1}^{n}  \modcont^{p+1}_{k} e^{-k \diffp/2}
  +  \modcont^{p+1}_{n} e^{-(n+1) \diffp/2}
\end{align*}
Using the same argument for $\modcont_{k^-}$, we have
\begin{align*}
  \lefteqn{\E\left[|\mech_\cont(\ds) - \func(\ds)|^p \right]
    \le 2^{p-1} \rho + 
    2^p \frac{\sum_{k=1}^{n}  \modcont^{p+1}_k e^{-k\diffp/2}}{
      \sum_{k=1}^{n}  \modcont_k e^{-k \diffp/2}}
    +  \frac{(2\modcont_n)^{p+1}}{\rho} e^{-(n+1) \diffp/2}} \\
  & \qquad ~
  \le 2^{p-1} \rho + 2^p \max_{1 \le k \le n}  \modcont^{p}_k e^{-k\diffp/4} 
  \frac{\sum_{k=1}^{n}  \modcont_k e^{-k\diffp/4}}
       {\sum_{k=1}^{n}  \modcont_k e^{-k \diffp/2}}
       +  \frac{(2\modcont_n)^{p+1}}{\rho 
	 + \sum_{k=1}^{n} \modcont_k e^{-k \diffp/2}} e^{-(n+1) \diffp/2}.
\end{align*}
Theorem~\ref{thm:loss-monotone}
now follows from the definition of $\expmod_\func(\ds;\diffp)$
and that $\frac{1 - e^{-\diffp/2}}{1 - e^{-\diffp/4}} \le 2$.

\subsection{Proof of Corollary~\ref{corollary:optimal-exponential-growth}}
\label{proof:optimal-exponential-growth}

We need to prove that $\ratiomod_\func(\ds) =
\frac{\expmod_\func(\ds;\diffp/4)}{\expmod_\func(\ds;\diffp/2)} \lesssim 1$.
The assumptions of the corollary imply that $\modcont_\func(\ds;k)
\lesssim \modcont_\func(\ds;\frac{8C}{\diffp}) e^{k \diffp /8}$ for $k \ge
k' = \frac{8C}{\diffp}$. Therefore we have
\begin{align*}
  \expmod_\func(\ds;\diffp/4) 
  & = (1 - e^{-\diffp/4}) \sum_{i=1}^{n} \modcont_\func(\ds;i) e^{-i \diffp/4} \\
  & \lesssim
  (1 - e^{-\diffp/4}) \sum_{i=1}^{n} \modcont_\func(\ds;k') e^{-i \diffp/8}
  % \\
  \le
  \modcont_\func(\ds;k') \frac{1 - e^{-\diffp/4}}{1 - e^{-\diffp/8}} 
  \le 2 \modcont_\func(\ds;k')
\end{align*}
and
\begin{align*}
  \expmod_\func(\ds;\diffp/2) 
  % & = (1 - e^{-\diffp/2}) \sum_{i=1}^{n} \modcont_\func(\ds;i) e^{-i \diffp/2} \\
  & \ge (1 - e^{-\diffp/2}) \sum_{i=k'}^{n} \modcont_\func(\ds;k') e^{-i \diffp/2}
  =  \modcont_\func(\ds;k') e^{-k' \diffp/2} (1 - e^{-(n-k'+1)\diffp/2})
  \gtrsim \modcont_\func(\ds;k') ,
\end{align*}
where the last inequality follows since $k' = \frac{8C}{\diffp}$.  Combining
our upper and lower bounds yields $\ratiomod_\func(\ds) \lesssim 1$.
Now we prove the claim about utility.  Since $\ratiomod_\func(\ds) \lesssim
1$, the exponential growth assumption implies that $\max_k
\modcont(\ds; k) e^{-k\diffp/4} \le \modcont(\ds; \frac{4C}{\diffp})$,
which therefore
proves the claim using Theorem~\ref{thm:loss-monotone}, once
we note that $\gamma
\lesssim e^{-n \diffp / 4}$ in the theorem.
%% \lesssim \frac{\modcont_n e^{-n \diffp /2}}{\modcont_{\frac{4C}{\diffp}}}
%% \le e^{-n \diffp /4}$.

\subsection{Proof of Proposition~\ref{proposition:conditional-optimality}}
\label{proof:conditional-optimality}
We will show that $\ratiomod_\func(\ds) \lesssim 1$.
Denote $\modcont_k \defeq \modcont_\func(\ds;k)$
for shorthand.
We have
\begin{align*}
  \expmod_\func(\ds;\diffp/4)
  & = (1 - e^{-\diffp/4})\sum_{i=1}^{n} \modcont_{i} e^{-i \diffp/4} 
  \stackrel{(\star)}{\lesssim}
  (1 - e^{-\diffp/4})
  \sum_{i=1}^{n} \modcont_{\tilde k} e^{-\tilde k \diffp/16} e^{-i \diffp/8} 
  % & = (1 - e^{-\diffp/4}) \modcont_{\tilde k} e^{-\tilde k \diffp/16}  \sum_{i=1}^{n} e^{-i \diffp/8} 
  \le 2 \modcont_{\tilde k} e^{-\tilde k \diffp/16}.
\end{align*}
Here inequality~$(\star)$ follows by
Theorem~\ref{thm:lower-bound-general-loss}, which states
that any $\diffp$-DP $\ell_1$-unbiased mechanism must satisfy
$\half e^{-2k \diffp} \modcont_k \le \E[|\mech(\ds) - \func(\ds)|]$,
so that the postulated $\diffp/16$-DP unbiased mechanism $\wt{\mech}$
must therefore satisfy
\begin{equation*}
  \half \max_k \modcont_k e^{-\frac{k \diffp}{8}}
  \le \E[|\wt{\mech}(\ds) - \func(\ds)|]
  \lesssim \modcont_{\tilde{k}} e^{-\tilde{k} \diffp / 16},
\end{equation*}
or $\modcont_i e^{-i \diffp/4} \lesssim
\modcont_i e^{-i \diffp/8} \modcont_{\tilde{k}} e^{-\tilde{k}\diffp / 16}$ for
each $i$. 
We also have the lower bound
\begin{align*}
  \expmod_\func(\ds;\diffp/2) 
  & = (1 - e^{-\diffp/2}) \sum_{i=1}^{n} \modcont_{i} e^{-i \diffp/2} 
  \ge (1 - e^{-\diffp/2}) \sum_{i=\tilde k}^{n} \modcont_{\tilde k} e^{-i \diffp/2} \\
  & = (1 - e^{-\diffp/2}) \modcont_{\tilde k} e^{-\tilde k \diffp/2} \sum_{i=0}^{n - \tilde k} e^{-i \diffp/2} \\
  & = (1 - e^{-\diffp/2})  \modcont_{\tilde k} e^{-\tilde k \diffp/2} \frac{e^{\diffp/2} - e^{-(n - \tilde k) \diffp/2}}{e^{\diffp/2} - 1}
  \gtrsim \modcont_{\tilde k} e^{-\tilde k \diffp/2}.
\end{align*}
Overall we have that the ratio~\eqref{eq:ratio-expected-modulus} satisfies
$\ratiomod_\func(\ds) \lesssim e^{7\tilde k \diffp/16}$.

%%  The second claim
%% of the proposition now follows from Theorem~\ref{thm:loss-monotone}, as
%% whenever
%% $\tilde k = O(1/\diffp)$ we have $\ratiomod_\func(\ds) \lesssim 1$ and
%% $\gamma \lesssim \frac{\modcont_{n} e^{-n \diffp/2}}{\modcont_{\tilde
%%     k}} \lesssim e^{-7 n \diffp/16}$.

\subsection{Proof of Proposition~\ref{proposition:near-instance-opt-monotone}}
\label{sec:proof-near-instance-opt-monotone}

We start with the following useful lemma.

\begin{lemma}
\label{lemma:bound-loss-holder}
	Let $0 \le a_1 \le \dots \le a_K$ and $\lambda = \frac{1}{\log{K}} $. 
	Then for $t \ge 1$,
	\begin{equation*}
	\frac{\sum_{k=1}^{K}  a_k^{t} (a_k - a_{k-1} ) e^{-k \diffp}}
	{\sum_{k=1}^{K} (a_k - a_{k-1} ) e^{-k \diffp}} 
			\le e \exp{\left(  \frac{4\log \dssize}{\log{\frac{1}{\diffp}} + \log \log \dssize} \right)} \max_{1 \le k \le K} e^{-\lambda k \diffp } a_k^t .
	\end{equation*}
\end{lemma}
\begin{proof}
	We have:
	\begin{align*}
		\frac{\sum_{k=1}^{K}  a_k^t (a_k -a_{k-1}) e^{-k \diffp}}{
			\sum_{k=1}^{K} (a_k - a_{k-1}) e^{-k \diffp}}
		& = \frac{\sum_{k=1}^{K}  e^{-\lambda k \diffp } a_k^t (a_k - a_{k-1}) e^{-(1 - \lambda)k \diffp}}{
			\sum_{k=1}^{K} (a_k - a_{k-1}) e^{-k \diffp}} \\
		& \stackrel{(i)}{\le}
		\frac{\left(\sum_{k=1}^{K}  (e^{-\lambda k \diffp } a_k^t)^q \right)^\frac{1}{q}
			\left(\sum_{k=1}^{K}  ((a_k - a_{k-1}) e^{-(1 - \lambda)k \diffp})^p \right)^\frac{1}{p}}{\sum_{k=1}^{K} (a_k - a_{k-1}) e^{-k \diffp}} \\
		& \stackrel{(ii)}{\le} e
		\max_{1 \le k \le K} e^{-\lambda k \diffp } a_k^t
		\frac{\left(\sum_{k=1}^{K}  ((a_k - a_{k-1}) e^{-(1 - \lambda)k \diffp})^p \right)^\frac{1}{p}}{\sum_{k=1}^{K} (a_k - a_{k-1}) e^{-k \diffp}} \\
		& \stackrel{(iii)}{=} e \max_{1 \le k \le K} e^{-\lambda k \diffp } a_k^t
		\frac{\left(\sum_{k=1}^{K}  ((a_k - a_{k-1})^p e^{-k \diffp}) \right)^\frac{1}{p}}{
			\sum_{k=1}^{K} (a_k - a_{k-1}) e^{-k \diffp}}  
	\end{align*}
	where $(i)$ follows from H\"older's inequality where we set $q =
	\log{K}$ and $p=\frac{q}{q-1}$, $(ii)$ follows since $\norm{z}_q \le
	d^{1/q} \linf{z} \le e \linf{z}$ for $z \in \R^d$ when $q \ge \log d$,
	$(iii)$ follows since $\lambda = \frac{1}{q}$.  Let us now separately
	consider the last term. We have:
	\begin{align*}
		\frac{\left(\sum_{k=1}^{K}  ((a_k - a_{k-1})^p e^{-k \diffp}) \right)^\frac{1}{p}}
		{\sum_{k=1}^{K} (a_k - a_{k-1}) e^{-k \diffp}} 
		& = \frac{\left(\sum_{k=1}^{K}  ((\frac{a_k}{a_{K}} - \frac{a_{k-1}}{a_{K}} )^p e^{-k \diffp}) \right)^\frac{1}{p}} 
		{\sum_{k=1}^{K} (\frac{a_k}{a_{K}} - \frac{a_{k-1}}{a_{K}}) e^{-k \diffp}} \\
		&  \stackrel{(i)}{\le}  \left(\sum_{k=1}^{K}  (\frac{a_k}{a_{K}} - \frac{a_{k-1}}{a_{K}} ) e^{-k \diffp}
		\right)^{\frac{1}{p} - 1} \\
		& \stackrel{(ii)}{\le}  \left( \frac{1}{K} e^{-t \diffp} \right)^{-\frac{1}{q}} 
		= e^{\frac{t \diffp + \log{K}}{q} } 
		\le e^{\frac{K \diffp + \log{K}}{\log{K}} } 
		= e^{\frac{4 \log{\frac{\GS_\func}{\rho \diffp}}}{ \log{\frac{1}{\diffp}} + \log\log{\frac{\GS_\func}{\rho \diffp}} }}
	\end{align*}
	where $(i)$ follows since $\frac{a_k}{a_{K}} \le 1$ and $p > 1$,
	and $(ii)$ follows since there exists $1 \le t \le K$ such that $a_t - a_{t-1} \ge \frac{a_{K}}{K}$.
	
\end{proof}

Following the same notation as Theorem~\ref{thm:loss-monotone},
Lemma~\ref{lemma:montone-loss} implies that
\begin{align*}
\E\left[|\mech_\cont(\ds) - \func(\ds)|\right]
		& \le  2 \rho 
		+  \frac{\sum_{k=1}^{n}  \modcont_{k^+} (\modcont_{k^+} - \modcont_{{k-1}^+} ) e^{-k \diffp/2}}
		{\rho + \sum_{k=1}^{n} (\modcont_{k^+} - \modcont_{{k-1}^+})	e^{-k \diffp/2}} +  \frac{\sum_{k=1}^{n}  \modcont_{k^-} (\modcont_{k^-} - \modcont_{{k-1}^-} ) e^{-k \diffp/2}}
		{\rho + \sum_{k=1}^{n} (\modcont_{k^-} - \modcont_{{k-1}^-}) e^{-k \diffp/2}}.
\end{align*}
%Let $S_{\ds}^k = \{t : \smoothinvmodcont_\func(\ds;t) = k \}$
%and denote $\modcont_k = \modcont_\func(\ds;k)  $. Since $\func$ is monotone,
%we have that for $k \ge 1$
%\begin{equation*}
%	\P \left( \mech(\ds) \in S_{\ds}^k \right) \le
%	\frac{(\modcont_k - \modcont_{k-1}) e^{-k \diffp}}
%	{\rho + \sum_{k=1}^{n} (\modcont_k - \modcont_{k-1}) e^{-k \diffp}}
%\end{equation*}
%This is true since the maximum of $S_{\ds}^k$ is at most $\modcont_k + \rho$
%and the minimum is at least $\modcont_{k-1} + \rho$ for $k \ge 1$.
%We can therefore write the loss
%\begin{align}
%	\label{eq:main}
%	\E[ \left| \mech(\ds) - \func(\ds) \right| ]
%	& \le \sum_{k=0}^{n} \P \left( \mech(\ds) \in S_{\ds}^k \right) (\modcont_k + \rho) \nonumber \\
%	& = \rho + \sum_{k=1}^{n} \P \left( \mech(\ds) \in S_{\ds}^k \right) \modcont_k \nonumber \\
%	& = \rho + 	\frac{\sum_{k=1}^{n}  \modcont_k (\modcont_k - \modcont_{k-1} ) e^{-k \diffp}}
%	{\rho + \sum_{k=1}^{n} (\modcont_k - \modcont_{k-1}) e^{-k \diffp}} \nonumber \\
%	& \le \rho + \frac{\sum_{k=1}^{K_0}  \modcont_k (\modcont_k - \modcont_{k-1}) e^{-k \diffp}}
%	{\rho + \sum_{k=1}^{n} (\modcont_k - \modcont_{k-1}) e^{-k \diffp}}
%	+ \frac{1}{\rho} \sum_{k=K_0 + 1}^{n} \modcont_k (\modcont_k - \modcont_{k-1}) e^{-k \diffp}. 
%\end{align}
We only bound the second term as the third follows
using similar arguments. 
For $K >0$ to be chosen presently, we have the following upper
bound on the second term
\begin{equation}
\label{eq:intermediate-inequality}
\frac{\sum_{k=1}^{K}  \modcont_{k^+}
	(\modcont_{k^+} - \modcont_{{k-1}^+} ) e^{-k \diffp/2}}
	{\sum_{k=1}^{K}(\modcont_{k^+} - \modcont_{{k-1}^+}) e^{-k \diffp/2}} 
	+ 	\frac{1}{\rho} \sum_{k=K + 1}^{n} \modcont_{k^+} (\modcont_{k^+} - \modcont_{{k-1}^+} ) e^{-k \diffp/2}
\end{equation}
Since $\modcont_{k^+} \le \modcont_k $, Lemma~\ref{lemma:bound-loss-holder} gives an upper bound
on the first term in~\eqref{eq:intermediate-inequality} 
where $\lambda = \frac{1}{\log{K}}$,
\begin{equation*}
\frac{\sum_{k=1}^{K}  \modcont_{k^+} 
	(\modcont_{k^+} - \modcont_{{k-1}^+} ) e^{-k \diffp/2}}
{\sum_{k=1}^{K}(\modcont_{k^+} - \modcont_{{k-1}^+}) e^{-k \diffp/2}} 
		\le e \exp{\left(  \frac{4\log \dssize}{\log{\frac{2}{\diffp}} + \log \log \dssize} \right)} \max_{1 \le k \le K} e^{-\lambda k \diffp/2 } \modcont_k. 
\end{equation*}

We can also upper bound the second term in~\eqref{eq:intermediate-inequality},
\begin{align*}
	\frac{1}{\rho} \sum_{k=K + 1}^{n} \modcont_{k^+} (\modcont_{k^+} - \modcont_{{k-1}^+}) e^{-k \diffp/2}
	& \le \frac{\GS^2_\func e^{-K \diffp/2}}{\rho} \sum_{k=1}^{n - K} (k + K)  e^{-k \diffp/2} \\
	& \le O(1) \frac{\GS^2_\func}{\rho \diffp^2} e^{-K \diffp/2}
	+ O(1) \frac{K \GS^2_\func}{\rho \diffp} e^{-K \diffp/2} \\
	& \le O(1) \frac{\GS^2_\func}{\rho \diffp^2} e^{-K \diffp/2}
	+ O(1) \frac{\GS^2_\func}{\rho \diffp^2} e^{-K \diffp/4}
	\le O(1) \rho,
\end{align*}
where the third inequality follows since $K \diffp \ge 4$ and $z e^{-z} \le e^{-z/2}$
for $z \ge 4$, and the last inequality follows by setting $\rho = n^{-p}$, 	
$K = \frac{8(p+2)\log{n}}{\diffp}$ and the assumptions that 
$\GS_\func \le n$ and $\diffp \ge n^{-1}$.
Proposition~\ref{proposition:near-instance-opt-monotone} now follows.

% -*- Mode: latex -*- %

\subsection{Proof of Proposition~\ref{prop:worse-case-upper-bound}}
\label{sec:proof-prop-worse-case-upper-bound}

We prove inequalities~\eqref{eqn:as-good-as-laplace}
and~\eqref{eqn:as-good-as-smooth} in turn.
Our proofs in fact prove a stronger result
for higher moments $\E\left[|\mech(\ds) - \func(\ds)|^p \right] $
for any $p \in \N$.

\paragraph{Proof of inequality~\eqref{eqn:as-good-as-laplace}}

As in the proof of Theorem~\ref{thm:upper-bound-general-loss},
Eq.~\eqref{eqn:slice}, define the slices $S_\sample^k \defeq \{t \in \range
: \invmodcont_\func(\sample; t) = k\}$, and let $k_0$ be the minimal integer
satisfying $\modcont_\func(\ds;k ) \ge \frac{\GS_\func}{\diffp}$.  (If no
such $k_0$ exists, then inequality~\eqref{eqn:as-good-as-laplace} is
immediate.)  By monotonicity, there exists $t \in \R$ (w.l.o.g.\ we take $t
\ge f(\sample)$) satisfying $t - f(\sample) = \modcont_\func(\sample; k_0)$
and $\invmodcont_f(\sample; t) = k_0$. Then by monotonicity, we have
$\invmodcont_f(\sample; s) \le k_0$ for all $f(\sample) \le s \le t$, and
using the uniformity of $\basemeasure$ on $\range$ we obtain
%% we may define
%% \begin{equation*}
%%   B_0 \defeq \basemeasure([0, \modcont_\func(\sample; k_0)])
%%   = \half \basemeasure([-\modcont_\func(\sample; k_0),
%%     \modcont_\func(\sample; k_0)])
%%   \ge \basemeasure([0, \GS_\func / \diffp]),
%% \end{equation*}
\begin{equation*}
  \int_\range e^{-\invmodcont_f(\sample; t)} d\basemeasure(t)
  \ge
  \basemeasure([0, \modcont_\func(\sample; k_0)]) e^{-k_0 \diffp / 2}
  \ge
  \basemeasure([0, \GS_\func / \diffp]) e^{-k_0 \diffp / 2}
  \ge \frac{1}{\diffp}
  \basemeasure([0, \GS_\func]) e^{-k_0 \diffp / 2}.
\end{equation*}
At the same time, for each slice $S_\sample^k$ we have
by monotonicity (Def.~\ref{def:monotone}) that
there exist (at most) two points $t_-, t_+$ satisfying
$t_- \le \func(\sample) \le t_+$ and
\begin{equation*}
  S_\sample^k \subset [t_- - \GS_\func, t_+ + \GS_\func]
  \cup [t_+ - \GS_\func, t_+ + \GS_\func],
\end{equation*}
so that
$\basemeasure(S_\sample^k) \le
4 \basemeasure([0, \GS_\func])$ by uniformity.
This gives the probability bound
\begin{equation}
  \label{eqn:probability-ratio-gs}
  \P(\mech(\sample) \in S_\sample^k)
  = \frac{\basemeasure(S_\sample^k) e^{-k \diffp / 2}}{
    \sum_{l = 0}^n \basemeasure(S_\sample^l) e^{-l \diffp / 2}}
  \le \frac{\basemeasure(S_\sample^k) e^{-k \diffp / 2}}{
    \basemeasure([0, \GS_\func / \diffp]) e^{-k_0 \diffp / 2}}
  \le 4 \diffp e^{- (k - k_0) \diffp / 2}.
\end{equation}

Now, note that $\modcont_\func(\sample; k) \le
\modcont_\func(\sample; k_0) + \GS_\func (k - k_0)$ for all $k \ge k_0$
to obtain
\begin{align*}
  \E\left[|\mech(\ds) - \func(\ds)|^p \right] 
  %% & = \sum_{t \neq \func(\ds)} \P(\mech(\ds)=t) |t - \func(\ds)| \\
  & \le \sum_{k=1}^{n}
  \P\left(\mech(\ds) \in S_\ds^{k} \right) \modcont_\func(\ds;k)^p \\
  & \le  \modcont_\func(\ds; k_0)^p
  +  \sum_{k=k_0 + 1}^{\dssize} \P \left(\mech(\ds) \in S_\ds^{k} \right)
  \modcont_\func(\ds;k)^p \\
  & \stackrel{(i)}{\le}
  \modcont_\func(\ds;k_0)^p 
  + 4 \diffp \sum_{k=k_0 + 1}^{\dssize} e^{-(k - k_0) \diffp /2}
  (\modcont_\func(\ds;k_0)  + (k-k_0 )\GS_\func)^p \\
  & \stackrel{(ii)}{\le}
  \modcont_\func(\ds; k_0)^p
  + 2^{p+1} \diffp
  \sum_{k = k_0 + 1}^n e^{-(k - k_0) \diffp / 2}
  \left(\modcont_\func(\ds; k_0)^p + (k - k_0)^p \GS_\func^p\right) \\
  & \le
  C \cdot 2^p
  \cdot \bigg[
    \modcont_\func(\ds; k_0)^p
    + \Gamma(p + 1) \frac{\GS_\func^p}{\diffp^p}\bigg],
  %% & \le  \modcont_\func(\ds;k_0) 
  %% + 4 \diffp \sum_{k=1}^{\infty}
  %% e^{- k \diffp /2} (\modcont_\func(\ds;k_0)  + k \GS_\func) \\
  %%       		%& \le \modcont_\func(\ds;k_0)  
  %%       		%	+ 2 \diffp \frac{\modcont_\func(\ds;k_0)}{e^{\diffp/2} - 1 }
  %%       		%	+ 2 \diffp \frac{e^{\diffp/2} \GS_\func}{(e^{\diffp/2}-1)^2} \\
  %%       		%& \le 5 \modcont_\func(\ds;k_0)  + 8 e^{\diffp/2} \frac{\GS_\func}{\diffp} \\
  %%       		& \le  ( 5(1+\diffp) +  8 e^{\diffp/2} ) \frac{\GS_\func}{\diffp},
\end{align*}
where inequality~$(i)$ used the probability
bound~\eqref{eqn:probability-ratio-gs} and
inequality~$(ii)$ used that $(a + b)^p \le 2^{p-1} a^p + 2^{p-1} b^p$ for
$p \ge 1$.
Using that
$\modcont_\func(\sample; k_0)
\le \frac{\GS_\func}{\diffp} + \GS_\func
= \frac{\GS_\func(1 + \diffp)}{\diffp}$ by definition of $k_0$ then gives
inequality~\eqref{eqn:as-good-as-laplace}.

\paragraph{Proof of inequality~\eqref{eqn:as-good-as-smooth}}

Let $T \in \N$ be (for now) arbitrary,
and for shorthand define the quantity
$L \defeq \max_{\sample' : \dham(\ds, \ds') \le T} \LS_\func(\sample')$.
We claim that for $p \ge 1$,
\begin{equation}
  \label{eqn:settable-to-smooth}
  \E\left[| \mech(\ds) - \func(\ds) |^p \right]^{1/p}
  \le C
  \left[\frac{L}{\diffp} p
    + \GS_\func T e^{-\frac{T \diffp}{2p}}
    \left(\frac{\GS_\func}{\rho} \frac{T}{\hinge{T \diffp / (2p) - 1}}
    \right)^{1/p}
    \right].
\end{equation}
To see how inequality~\eqref{eqn:as-good-as-smooth} follows
from the bound~\eqref{eqn:settable-to-smooth}, set
\begin{equation*}
  T = \max\left\{\frac{4p}{\diffp},
  \frac{2p}{\diffp}\left(
  \log \GS_\func
  + \frac{1}{p} \log \frac{\GS_\func}{\rho \diffp}
  + \log \frac{1}{\gamma}\right)\right\}
\end{equation*}
and note that in this case, $\frac{T \diffp}{2p} - 1 \ge \frac{T \diffp}{p}$.

We thus prove inequality~\eqref{eqn:settable-to-smooth}.  Let $k_0 \le T$
denote the minimal integer such that $\modcont_\func(\ds;k_0) \ge
\frac{L}{\diffp}$, as inequality~\eqref{eqn:settable-to-smooth} is immediate
if no such $k_0$ exists.  As in inequality~\eqref{eqn:probability-ratio-gs},
we have $\P(\mech(\sample) \in S_\sample^k) \le 4 \diffp e^{-(k - k_0)
  \diffp / 2}$ for $k \ge k_0$,
and so
\begin{align}
  \lefteqn{\E\left[|\mech(\ds) - \func(\ds)|^p \right] 
    %% & = \sum_{t \neq \func(\ds)} \P(\mechdisc(\ds)=t) |t - \func(\ds)| \\
    \le \sum_{k=1}^{n} \P\left(\mech(\ds) \in S_\ds^{k} \right)
    \modcont_\func(\ds;k)^p} \nonumber \\
  & ~~~ \le  \modcont_\func(\ds; k_0)^p
  +  \sum_{k=k_0 + 1}^{T} \P \left(\mech(\ds) \in S_\ds^{k} \right)
  \modcont_\func(\ds;k)^p
  + \sum_{k= T + 1}^{\dssize} \P \left(\mech(\ds) \in S_\ds^{k} \right)
  \modcont_\func(\ds;k)^p.
  \label{eqn:sum-smooth-bound}
\end{align} 
For the second term in the sum~\eqref{eqn:sum-smooth-bound}, the same
derivation as in the proof of inequality~\eqref{eqn:as-good-as-laplace}
yields
\begin{equation*}
  \sum_{k=k_0 + 1}^{T} \P \left(\mech(\ds) \in S_\ds^{k} \right)
  \modcont_\func(\ds;k)^p
  \le
  C \cdot 2^p \cdot \bigg[
    \modcont_\func(\ds; k_0)^p
    + \Gamma(p+1) \frac{L^p}{\diffp^p}\bigg].
\end{equation*}
For the third term in the sum~\eqref{eqn:sum-smooth-bound}, we use
a utility inequality~\cite[Thm.~2.1]{BorweinCh09} for the incomplete
gamma function that
\begin{equation*}
  \Gamma(\alpha, \tau) \defeq \int_\tau^\infty e^{-u} u^{\alpha - 1}
  du \le \tau^\alpha \exp(-\tau) \cdot
  \frac{1}{\hinge{\tau + 1 - \alpha}}.
\end{equation*}
An analogue of inequalities~\eqref{eqn:probability-bound-dimension}
in the proof of Lemma~\ref{lemma:upper-bound-continuous} with $d = 1$
and \eqref{eqn:probability-ratio-gs} above
implies that
\begin{equation*}
  \P(\mech(\sample) \in S_\sample^k)
  \le C e^{-k \diffp / 2} \frac{\basemeasure([-\GS_\func, \GS_\func])}{
    \basemeasure([0, \rho])}
  \le C \frac{\GS_\func}{\rho} e^{-k \diffp / 2}.
\end{equation*}
As we always have $\modcont_\func(\sample; k)
\le k \GS_\func$, the third term of
inequality~\eqref{eqn:sum-smooth-bound} has bound
\begin{align*}
  \sum_{k= T + 1}^{\dssize}
  \P \left(\mech(\ds) \in S_\ds^{k} \right) \modcont_\func(\ds;k)^p
  & \le C \frac{\GS_\func^{p+1}}{\rho} \sum_{k = T + 1}^\dssize
  k^p e^{-k \diffp / 2} \\
  %% & \le 2 \GS_\func^p \sum_{k = T + 1}^\dssize
  %% k^p \left(k \frac{\GS_\func}{\rho} + 1 \right) e^{-k \diffp / 2} \\
  & \le C \cdot
  \frac{\GS_\func^{p+1}}{\rho}
  \left(\frac{2}{\diffp}\right)^{p+1}
  \int_{T \diffp / 2}^\infty e^{-u} u^{p + 1 - 1} du
  \\
  & \le C \cdot \frac{\GS_\func^{p+1}}{\rho}
  \cdot \frac{T^{p + 1}}{\hinge{T \diffp / 2 - p}} \exp(-T \diffp / 2)
\end{align*}
Substituting into inequality~\eqref{eqn:sum-smooth-bound}
and using $\Gamma(p + 1)^{1/p} \lesssim p$
yields
\begin{equation*}
  \E\left[|\mech(\ds) - \func(\ds)|^p \right]^{1/p}
  \lesssim
  \modcont_\func(\sample; k_0)
  + p \frac{L}{\diffp}
  + \GS_\func T e^{-\frac{T \diffp}{2p}}
  \left(\frac{\GS_\func}{\rho}
  \frac{T}{\hinge{T \diffp / 2 - p}}\right)^{1/p}.
\end{equation*}
The claim~\eqref{eqn:settable-to-smooth}
follows as $\modcont_\func(\ds;k_0) \le \frac{L(1+\diffp)}{\diffp}$ by
construction and because $p^{-1/p} \le 1$ for $p \ge 1$.

\section{Proofs associated with examples (Sections
  \ref{sec:avg-instance-optimal} and~\ref{sec:examples})}
\label{sec:example-proofs}

% -*- Mode: latex -*- %

\subsection{Partial minimization}
\label{sec:partial-minimization}

We prove that both
inequalities~\eqref{eqn:modulus-partial-min}
and~\eqref{eqn:compute-modulus-partial-sum} hold in
Example~\ref{example:smoothness-partial-min}.
We adopt standard empirical process notation for simplicity in the derivation,
so that for a sample $\ds = \{x_i\}_{i=1}^n$,
$P_n = n^{-1} \sum_{i = 1}^n 1_{x_i}$ denotes the empirical measure
on $\ds$, and $P_n f = n^{-1} \sum_{i = 1}^n f(x_i)$. Given a modified
sample $\ds' = \{x_i'\}_{i=1}^n$ with $\dham(\ds', \ds)$, we use the shorthand
\begin{equation*}
  Q_k = \frac{1}{n} \sum_{i=1}^n (1_{x_i'} - 1_{x_i})
  = \frac{1}{n} \sum_{i : x_i' \neq x_i} (1_{x_i'} - 1_{x_i'})
\end{equation*}
for the difference measure, so that $P_n' = n^{-1} \sum_{i = 1}^n 1_{x_i'} =
P_n + Q_k$.
With this, we have that statistic of interest is
$\theta(\ds) = \argmin_\theta \inf_\beta
P_n \ell(\theta, \beta; X)$,
and given a modified sample $\ds'$, we have $\theta(\ds') = \argmin_\theta
\inf_\beta (P_n + Q_k) \ell(\theta, \beta; X)$. For shorthand throughout the
proof, we recall that $\tau = (\theta, \beta)$
and $\riskn(\tau) = P_n \ell(\tau; X) = P_n \ell(\theta, \beta; X)$.
First, we recapitulate inequalities~\eqref{eqn:modulus-partial-min}
and~\eqref{eqn:compute-modulus-partial-sum} in the notation here.
Inequality~\eqref{eqn:modulus-partial-min} becomes
\begin{equation}
  \label{eqn:modulus-partial-min-redux}
  \modcont_\theta(\ds; k)
  %% = \sup_{Q_k \in \mc{Q}_k} \left|\left[\ddot{\risk}_n(\tau)^{-1}
  %%   Q_k \dot{\ell}(\tau; X)\right]_1 \right|
  %% \pm \frac{24 {\lipobj}^2 \liphess}{\lambda^3} \cdot \frac{k^2}{n^2}
  = \left(1 \pm \frac{24 {\lipobj}^2 \lipgrad \liphess}{\lipobjin\lambda^3}
  \frac{k}{n}\right)
  \cdot
  \sup_{Q_k \in \mc{Q}_k} \left|\left[\ddot{\risk}_n(\tau)^{-1}
    Q_k \dot{\ell}(\tau; X)\right]_1 \right|.
\end{equation}
while for $A = \ddot{\risk}_n(\tau)^{-1}$ having first column
$a$, inequality~\eqref{eqn:compute-modulus-partial-sum} becomes
\begin{equation}
  \label{eqn:compute-modulus-partial-sum-redux}
  \frac{k}{n} \ltwo{a} \cdot \lipobjin
  \le
  \sup_{Q_k \in \mc{Q}_k}
  \left|\left[\ddot{\risk}_n(\tau)^{-1} Q_k \dot{\ell}(\tau; X)\right]_1
  \right|
  \le \frac{2 k}{n} \ltwo{a} \cdot \lipobj.
\end{equation}

We prove each in turn.
We begin with a few arguments that provide quantitative bounds on the change
in solutions to the empirical risk minimization problem; these are more or
less standard arguments in stability and implicit function
theorems~\cite{DontchevRo14}.
The main result we use is the following, which provides an
expansion of $\tau$ in terms of the perturbation measure $Q_k$.
\begin{lemma}
  \label{lemma:when-convex-life-good}
  For $\ds \in \mc{X}^n$, define
  $\tau(\ds) = \argmin_\tau P_n \ell(\tau; X)$, and assume that
  $\ddot{\risk}_n(\tau) \succeq \lambda I$ for some $\lambda > 0$.
  Let $\ds' \in \mc{X}^n$ satisfy $\dham(\ds', \ds) = k$ and
  have corresponding difference measure $Q_k$, and
  assume that
  $k \le \frac{\lambda^2}{6 \lipobj \liphess} \cdot n$.
  Then
  \begin{equation*}
    \tau(\ds') - \tau(\ds)
    = -\ddot{\risk}_n(\tau)^{-1} Q_k \dot{\ell}(\tau; X)
    + e(\ds', \ds)
  \end{equation*}
  where
  \begin{equation*}
    \ltwo{\tau(\ds') - \tau(\ds)}
    \le \frac{6 \lipobj}{\lambda} \frac{k}{n}
    ~~ \mbox{and} ~~
    \ltwo{e(\ds', \ds)}
    \le \frac{24 {\lipobj}^2 \liphess}{\lambda^3} \cdot \frac{k^2}{n^2}.
  \end{equation*}
\end{lemma}
\noindent
The proof of the lemma is more or less standard convex analysis and
perturbation theory and is somewhat tedious; we defer it to
Section~\ref{sec:proof-when-convex-life-good}.

With Lemma~\ref{lemma:when-convex-life-good} in place,
inequalities~\eqref{eqn:modulus-partial-min-redux}
and~\eqref{eqn:compute-modulus-partial-sum-redux} are straightforward.
Under the assumptions of
Example~\ref{example:smoothness-partial-min} we know that
$A^{-1} = \ddot{\risk}_n(\tau) \preceq \lipgrad I$, so that $A \succeq
{\lipgrad}^{-1} I$. In particular, the first diagonal entry $A_{11}
\ge {\lipgrad}^{-1}$.  Thus, using the gradient containment guarantees of
Assumption~\ref{assumption:smoothness-losses}, by considering the sign of
the first coordinate $[\ddot{\risk}_n(\tau)^{-1} \sum_{i : x_i \neq x_i'}
  \dot{\ell}(\tau; x_i)]_1$ we know that at least one of $\pm\frac{k}{n}
{\lipgrad}^{-1} \lipobjin$ must be contained in $[\ddot{\risk}_n(\tau)^{-1}
  Q_k \dot{\ell}(\tau; X)]_1$ as $Q_k$ varies over $\mc{Q}_k$.  Thus, we
obtain that for each $k \le \frac{\lambda^2}{6 \lipobj \liphess} n$, there
exists $\ds'$ with $\dham(\ds, \ds') = k$ such that
\begin{align*}
  |\theta(\ds) - \theta(\ds')|
  & \ge \sup_{Q_k \in \mc{Q}_k} \left|\left[
    \ddot{\risk}_n(\tau)^{-1} Q_k \dot{\ell}(\tau; X)\right]_1\right|
  - \frac{24 {\lipobj}^2 \liphess}{\lambda^3} \cdot \frac{k^2}{n^2} \\
  & \ge \left(1 - \frac{24 {\lipobj}^2  \lipgrad \liphess}{\lambda^3
    \lipobjin} \cdot \frac{k}{n} \right)
  \sup_{Q_k \in \mc{Q}_k} \left|\left[
    \ddot{\risk}_n(\tau)^{-1} Q_k \dot{\ell}(\tau; X)
    \right]_1 \right|.
\end{align*}
We obtain the upper bound on $
\modcont_\theta(\ds; k) = \sup_{\dham(\ds', \ds) \le k}
|\theta(\ds) - \theta(\ds')|$ claimed
in inequality~\eqref{eqn:modulus-partial-min-redux} similarly.
The final claim~\eqref{eqn:compute-modulus-partial-sum-redux} is immediate
by Cauchy-Schwarz.

%% Similarly, using that $\ddot{\risk}_n(\tau) \succeq \lambda I$ by assumption,
%% we obtain
%% \begin{equation*}
%%   [\ddot{\risk}_n(\tau)^{-1} Q_k \dot{\ell}(\tau; X)]_1
%%   \subset \left[- \frac{\lipobj k}{\lambda n}, \frac{\lipobj k}{\lambda n}
%%     \right].
%% \end{equation*}
%% That is, for any $\ds'$ with $\dham(\ds, \ds') \le k$ we have
%% \begin{equation*}
%%   |\theta(\ds) - \theta(\ds')|
%%   \le \left(1 + \frac{24 
%% \end{equation*}

\subsubsection{Proof of Lemma~\ref{lemma:when-convex-life-good}}
\label{sec:proof-when-convex-life-good}

We prove Lemma~\ref{lemma:when-convex-life-good} via series of lemmas
on perturbations of solutions and inverses.
\begin{lemma}
  \label{lemma:minimizers-near}
  Let $L$ be convex, have $\liphess$-Lipschitz Hessian, and
  satisfy $\ddot{\risk}(\tau) \succeq \lambda I$.
  If $\ltwos{\dot{\risk}(\tau)} \le \epsilon$ and
  $\epsilon < \frac{\lambda^2}{3 \liphess}$, then
  the minimizer $\tau\opt$ of $L$ satisfies
  $\ltwos{\tau\opt - \tau} \le \frac{3\epsilon}{\lambda}$.
\end{lemma}
\begin{proof}
  Whenever $\ltwo{\tau' - \tau} \le \frac{\lambda}{\liphess}$, we have
  by convexity and smoothness that
  \begin{align*}
    L(\tau')
    & \ge L(\tau) + \<\dot{\risk}(\tau), \tau' - \tau\>
    + \frac{\lambda}{2} \ltwo{\tau' - \tau}^2
    - \frac{\liphess}{6} \ltwos{\tau' - \tau}^3
    \\
    & \ge L(\tau) + \<\dot{\risk}(\tau), \tau' - \tau\>
    + \frac{\lambda}{3} \ltwo{\tau' - \tau}^2.
  \end{align*}
  If $\ltwos{\dot{\risk}(\tau)} \le \epsilon$, we have for all
  $\tau'$ satisfying $\ltwo{\tau' - \tau} >
  \frac{3 \epsilon}{\lambda}$ that
  \begin{equation*}
    \<\dot{\risk}(\tau), \tau' - \tau\>
    + \frac{\lambda}{3} \ltwo{\tau' - \tau}^2
    \ge -\epsilon \ltwo{\tau' - \tau}
    + \frac{\lambda}{3} \ltwo{\tau' - \tau}^2
    %% = \frac{1}{4} \ltwo{\tau' - \tau}\left(\lambda \ltwo{\tau'
    %%   - \tau} - 4 \epsilon\right)
    > 0.
  \end{equation*}
  By convexity, $\alpha \mapsto L(\tau + \alpha(\tau'
  - \tau))$ is increasing in $\alpha \ge 1$ when $\ltwo{\tau' - \tau}
  \ge \frac{3 \epsilon}{\lambda}$, so
  $L(\tau') > L(\tau)$ whenever $\ltwo{\tau' - \tau} > 3 \epsilon / \lambda$.
  Verifying that $3 \epsilon / \lambda < \lambda / \liphess$
  completes the proof.
\end{proof}

  %% Let $\mc{Q}_k$ denote the collection of measures of the form $\frac{1}{n}
  %% \sum_{i = 1}^n (1_{x_i} - 1_{x'_i})$, where $\dham(x, x') \le k$. Let
  %% $\tau = \argmin_\tau \riskn(\tau)$ and assume that $\dot{\risk}_n(\tau) = 0$ and
  %% $\ddot{\risk}_n(\tau) \succeq \lambda I$, and assume that $k \le
  %% \frac{\lambda^2}{6 \lipobj \liphess} \cdot n$.  Then for any $Q_k \in
  %% \mc{Q}_k$ and $\tau' = \argmin_\tau (P_n + Q_k) \ell(\tau; X)$,

\begin{lemma}
  \label{lemma:expand-local-error}
  Let the conditions and notation of of
  Lemma~\ref{lemma:when-convex-life-good} hold, and let $\tau = \tau(\ds)$
  and $\tau' = \tau(\ds')$.  Then
  there is
  an error matrix $E$ satisfying $\ltwo{E} \le \frac{6 \lipobj
    \liphess}{\lambda} \frac{k}{n}$ such that
  \begin{equation*}
    \tau' - \tau = -\left(\ddot{\risk}_n(\tau) + E\right)^{-1} Q_k \dot{\ell}(\tau;
    X),
    ~~~ \mbox{and} ~~~
    \ltwo{\tau' - \tau} \le \frac{6 \lipobj}{\lambda} \frac{k}{n}.
  \end{equation*}
\end{lemma}
\begin{proof}
  As $\dot{\risk}_n(\tau) = 0$, we have
  \begin{equation*}
    \ltwos{(P_n + Q_k) \dot{\ell}(\tau; X)}
    = \ltwos{Q_k \dot{\ell}(\tau; X)}
    \le \frac{1}{n} \sum_{i : x_i \neq x_i'}
    \ltwo{\dot{\ell}(\tau; x_i) - \dot{\ell}(\tau; x_i')}
    \le \frac{2 k}{n} \lipobj.
  \end{equation*}
  Let
  $\tau' = \argmin_\tau (P_n + Q_k) \ell(\tau; X)$.
  Using the assumption $\ddot{\risk}(\tau) \succeq \lambda I$,
  Lemma~\ref{lemma:minimizers-near} implies that
  $\ltwo{\tau' - \tau} \le \frac{6 k \lipobj}{n \lambda}$.
  A Taylor
  expansion gives an error matrix $E : \R^d \times \R^d \times \mc{X}
  \to \R^{d \times d}$ with
  $\ltwo{E(\tau, \tau'; x)} \le \liphess \ltwo{\tau - \tau'}$ such that
  \begin{align*}
    0 & = (P_n + Q_k) \dot{\ell}(\tau'; X) \\
    & = (P_n + Q_k) \dot{\ell}(\tau; X)
    + (P_n + Q_k) \ddot{\ell}(\tau; X) (\tau' - \tau)
    + (P_n + Q_k) E(\tau', \tau; X) (\tau' - \tau) \\
    & = Q_k \dot{\ell}(\tau; X)
    + \left(\ddot{\risk}_n(\tau) + E_{n,k}\right)(\tau' - \tau),
  \end{align*}
  where we use the shorthand $E_{n,k}
  = (P_n + Q_k) E(\tau', \tau; X)$. Moreover,
  $\ltwo{E_{n,k}} \le \liphess \ltwos{\tau - \tau'}
  \le \frac{6 k \liphess \lipobj}{n \lambda}$ by our
  bounds on $\ltwos{\tau - \tau'}$; inverting the
  preceding equality gives the lemma.
  %% in particular,
  %% \begin{equation*}
  %%   \tau' - \tau
  %%   = \left(\ddot{\risk}_n(\tau) + E_{n,k}\right)^{-1}
  %%   Q_k \dot{\ell}(\tau; X),
  %% \end{equation*}
  %% and the error matrix $E_{n,k}$ satisfies the conditions
  %% specified in the lemma.
\end{proof}

To control the error in Lemma~\ref{lemma:expand-local-error},
we use a standard matrix perturbation result~\cite{StewartSu90}.
\begin{lemma}
  \label{lemma:inverse-perturbations}
  Let $A \succeq \lambda I$ and
  $\ltwo{E} \le \epsilon$. Then
  there is $\Delta$ with
  $\ltwos{\Delta} \le \frac{\epsilon^2}{\hinge{\lambda - \epsilon}}$
  satisfying
  \begin{equation*}
    (A + E)^{-1} = A^{-1} - A^{-1} E A^{-1}
    + \sum_{i = 2}^\infty (-1)^i (A^{-1} E)^i A^{-1}
    = A^{-1} - A^{-1} (E + \Delta) A^{-1}.
  \end{equation*}
\end{lemma}

Revisiting Lemma~\ref{lemma:expand-local-error}, we have
$\tau' - \tau = -\ddot{\risk}_n(\tau)^{-1} Q_k \dot{\ell}(\tau; X)
+ e_{n,k}$ where for a matrix
$\Delta$ satisfying $\ltwo{\Delta} \le \ltwo{E}^2 / (\lambda - \ltwo{E})$,
the error $e_{n,k}$ satisfies
\begin{align*}
  \ltwo{e_{n,k}}
  & \le \ltwo{\ddot{\risk}_n(\tau)^{-1} (E + \Delta) \ddot{\risk}_n(\tau)^{-1}}
  \ltwo{Q_k \dot{\ell}(\tau; X)} \\
  & \stackrel{(i)}{\le} 3 \frac{\ltwo{E}}{\lambda^2}
  \ltwo{Q_k \dot{\ell}(\tau; X)}
  \stackrel{(ii)}{\le}
  \frac{24 {\lipobj}^2 \liphess}{\lambda^3} \frac{k^2}{n^2}.
\end{align*}
In inequality~$(i)$ we use that $\ltwo{E} \le \lambda/2$
and in inequality~$(ii)$ that
$\ltwo{E} \le \frac{6 \lipobj \liphess}{\lambda} \frac{k}{n}$ and
$\ltwos{Q_k \dot{\ell}(\tau; X)} \le 2 \lipobj k / n$.
This gives Lemma~\ref{lemma:when-convex-life-good}.

% -*- Mode: latex -*- %

%\section{Proofs for section~\ref{sec:examples}}

\subsection{Proofs for the median example (Section~\ref{sec:median})}
\label{sec:median-proofs}

We provide proofs for the median example here.  We frequently use the
following standard Chernoff bound.
\begin{lemma}[\cite{MitzenmacherUp05}, Ch.~4.2.1]
  \label{lemma:chernoff}
  Let $X = \sum_{i=1}^n X_i$ for $X_i \simiid \bernoulli(p)$.
  Then for $\delta \in [0,1]$,
  \begin{align*}
    \P(X > (1+\delta)np ) \le e^{-np\delta^2 /3}
    ~~~ \mbox{and} ~~~
    \P(X < (1-\delta)np ) \le e^{-np\delta^2 /2}.
  \end{align*}
\end{lemma}

Throughout this section, we let $\gamma>0$,
$m = \mathrm{Median}(P)$, $\hat m = \mathrm{Median}(\ds)$ be the 
empirical median,
and $ p_{\min} =  \inf_{|t - m| \le 2 \gamma} \pdf_P(t)$. 
First, we start with the following lemma, which proves that 
the empirical median is close to the true median with 
high probability.
\begin{lemma}
	\label{lemma:empirical-median-loss}
	%Let $m = Median(P)$, $\hat m = Median(\ds)$ be the 
	%empirical median,
	%and $ p_{min} =  \min_{|t - m| \le 2 \delta} \pdf_P(t)$.
	Under the setting of Proposition~\ref{prop:mech-median-performance-improved},
	for any $0 < u \le \gamma$,
	\begin{equation*}
	\P \left(|\hat m - m | > u \right) 
	\le 2 e^{-  n u^2 p_{\min}^2}.
	\end{equation*}
\end{lemma}
\begin{proof}
	Let $B_i = \indic{x_i > m + u}$ and 
$B = \sum_{i=1}^\dssize B_i$ denote the number of 
elements larger than $m + u$.
The definition of $p_{\min}$ implies that 
$\hat p = \P(B_i=1) \le 1/2 - u p_{\min} $.
If $\hat m > m + u$ then $B \ge \dssize/2$,
therefore we get 
\begin{align*}
\P(\hat m > m + u) 
\le \P(B \ge \dssize/2) 
=  \P(B \ge \hat p \dssize(1 + \frac{1}{2 \hat p} - 1 )) 
\le e^{- n u^2 p_{\min}^2},
\end{align*}
where the last inequality follows by using Chernoff bound of Lemma~\ref{lemma:chernoff}
and $\hat p  \le  1/2 - u p_{\min}$.
The same steps give that $\P(\hat m < m - u)  \le e^{-  n u^2 p_{\min}^2}$,
which proves the claim.	
\end{proof}

\begin{comment}
\begin{lemma}
	\label{lemma:mech-median-performance}
	%Let $m = Median(P)$, $\hat m = Median(\ds)$ be the 
	%empirical median,
	%and $ p_{min} =  \min_{|t - m| \le 2 \delta} \pdf_P(t)$.
	Under the setting of Proposition~\ref{prop:median-performance}, the mechanism $\mechcont$
	with $\rho = 1/\dssize$ has
	\begin{equation*}
	\P \left(|\mechcont(\ds) - \hat m | > 3 \delta + \rho \mid |\hat m - m| \le \delta \right) 
	\le  \datarange n e^{-n \delta p_{\min} \diffp /4} +  2 e^{-n \delta p_{\min} /8} + 4 e^{-  n \delta^2 p_{\min}^2}.
	\end{equation*}
\end{lemma}
\end{comment}

\subsubsection{Proof of Proposition~\ref{prop:mech-median-performance-improved}}

Let us first divide the interval $[m - \gamma, m + \gamma]$ to blocks of
size $u$: $I_1,I_2,\dots,I_T$. Let $N_i$ denote the number of elements in
$I_i$ and $A$ denote the event that $N_i \ge n u p_{\min} /2$ for all $i$,
and $B$ denote the event that $| m - \hat m| \le \gamma/2$.
	\begin{lemma}
	\label{lemma:median-crowded-improved}
	Under the above setting,
	\begin{equation*}
	\P(A \mid B ) \ge 1 - \frac{2 \gamma}{u} e^{-n u p_{\min}/8} - 2 e^{-  n \gamma^2 p_{\min}^2 / 4}.
	\end{equation*}
\end{lemma}
\begin{proof}
  Let $Z_i = \indic{x_i \in I_j}$ and $N_j = \sum_{i=1}^n Z_i $
  be the number of elements inside block $I_j$. As $\P(Z_i=1) \ge u p_{\min}$,
  we use the Chernoff bound of Lemma~\ref{lemma:chernoff} to get
  \begin{equation*}
    \P(N_j < n u p_{\min}/2) \le e^{-n u p_{\min}/8}.
  \end{equation*}
  Thus using a union bound we have
  \begin{equation*}
    \P(A^c) \le  \frac{2 \gamma}{u} e^{-n u p_{\min}/8}.
  \end{equation*}
  Using that for any events $A, B$ we have
  $\P(A \mid B) \ge \P(A \mid B) \P(B)
  = \P(A) - \P(A, B^c)
  \ge \P(A) - \P(B^c)$, the preceding display
  and Lemma~\ref{lemma:empirical-median-loss}
  give $\P(A \mid B ) 
  \ge \P(A) - \P(B^c)
  \ge 1 - \frac{2 \gamma}{u} e^{-n u p_{\min}/8} - 2 e^{-  n \gamma^2 p_{\min}^2 / 4}$.
\end{proof}

Now, we complete the proof of
Proposition~\ref{prop:mech-median-performance-improved}.  First, if events
$A$ and $B$ occur, then we know that for any $t$ such that $|t - \hat m| > 2
u$, there are at least $n u p_{\min} /2$ elements between $\hat m$ and
$t$. Therefore $\invmodcont_\func(\ds;t) \ge n u p_{\min}/2 $ for any $t$
such that $| t - \hat m | > 2 u$, which implies that
$\smoothinvmodcont_\func(\ds;s) \ge n u p_{\min}/2$ for $s$ such that $|s -
\hat m| > 2 u + \rho$.  The definition~\eqref{mech:continuous}
of the mechanism implies then
that for $t$ such that $|t - \hat m| > 2 u + \rho$,
\begin{align*}
  \pdf_{\mechcont(\ds)}(t \mid A,B) 
  = \frac{e^{-\smoothinvmodcont_\func(\ds;t ) \diffp/2}}{ \int_{s \in \range} e^{-\smoothinvmodcont_\func(\ds;s) \diffp/2}} 
  \le \frac{e^{-n u p_{\min} \diffp /4}}{\rho}.
  %& = n e^{-n u p_{min} \diffp /4}.
\end{align*}
Using a union bound to gives
\begin{align*}
  \P \left(|\mechcont(\ds) - \hat m | > 2 u  + \rho  \right) 
  & \le \P \left(|\mechcont(\ds) - \hat m | > 2 u  + \rho \mid A,B \right) 
  + \P(\wb B) + \P(\wb A \mid B) \\
  & \le \frac{\datarange}{\rho} e^{-n u p_{\min} \diffp /4}
  +  4 e^{-  n \gamma^2 p_{\min}^2 / 4}
  + \frac{2 \gamma}{u} e^{-n u p_{\min}/8}.
\end{align*}

\subsubsection{Proof of Lemma~\ref{lemma:smooth-laplace-median-failure}}
\label{sec:proof-smooth-laplace-median-failure}

Recall the definition~\eqref{mech:smooth-laplace} of the
smooth Laplace mechanism as
$\mechsmlap(\ds)
= \func(\ds) + \frac{2 S(\ds)}{\diffp} \laplace(1)$,
where $S(\ds)$ satisfies $\LS(\ds) \le S(\ds)$ and $S(\ds) \le e^{\beta}
S(\ds')$ for neighboring instances $\ds,\ds' \in \domain^\dssize$ and $\beta
= \frac{\diffp}{2 \ln(2/ \delta)}$.  To prove the lemma, we
show that $S(\ds) \ge \frac{\log(n)}{2e p_{\max} n \diffp}$ with high
probability. The main idea is to show that there exists an instance $\ds'$
such that $\dham(\ds,\ds') \le 1/\beta$ and $\LS(\ds') \ge
\frac{\log(n)}{2p_{\max} n \diffp}$.  To find $\ds'$, we show that---with
high probability---there exist at most $1/\beta$ elements between
$\mathrm{Median}(\ds)$ and $\mathrm{Median}(\ds) + \frac{\log(n)}{2p_{\max}
  n \diffp}$, so we get our desired $\ds'$ by increasing the values
of elements $x_i$ in this range.

Let $c > 0$, to be chosen, and define $\gamma = \frac{c \log(n)}{ n \diffp}$,
$p_{\max} = \sup_{t} \pdf_P(t) $,
$p_{\min} = \inf_{t} \pdf_P(t) $,
and
\begin{equation*}
  Z_i \defeq \indic{|x_i - \mathrm{Median}(\ds) | \le \gamma}.
\end{equation*}
Then 
$\sum_{i=1}^n Z_i$ upper bounds the number of elements between $\mathrm{Median}(\ds)- \gamma$
and $\mathrm{Median}(\ds) + \gamma$.
Now, we show that on the event that $\sum_{i=1}^n Z_i \le \frac{2
  \log(n)}{\diffp} $, we have that $\LS(\ds) \ge \gamma$.  Indeed, assume $\ds
= (x_1\le x_2 \le \dots \le x_n)$ such that $x_i =
\mathrm{Median}(\ds)$. Since $\sum_{i=1}^n Z_i \le \frac{2
  \log(n)}{\diffp}$, there exists $x_j \ge x_i + \gamma $ such that $|i -
j| \le \sum_{i=1}^n Z_i$.  Consider the instance $\ds'$ with
\begin{equation*}
  x'_k = 
  \begin{cases}
    x_j	& \text{if } i < k < j \\
    x_k   & \text{otherwise}.
  \end{cases}
\end{equation*}
To prove a lower bound on $S(\ds)$, we need to show that $\ds'$ has large
local sensitivity and is not too far from $\ds$.  Indeed, we have $\LS(\ds')
\ge |x_i - x_j| \ge \gamma$ as $\mathrm{Median}(\ds') = x_i$ but we can change
it to $x_j$ by changing one user, i.e., $x'_i = x_j$.  Moreover,
$\dham(\ds,\ds') \le \sum_{i=1}^n Z_i \le \frac{2 \log(n)}{\diffp}$, so that
as $\beta = \frac{\diffp}{2 \ln(2/ \delta)} \le \frac{\diffp}{2 \log(n)} $,
we have
\begin{align*}
  S(\ds) 
  \ge e^{-\dham(\ds,\ds') \beta} S(\ds') 
  \ge e^{-\beta \sum_{i=1}^n Z_i } \gamma 
  \ge e^{-1} \gamma
  = \frac{c \log(n)}{ e n \diffp}.
\end{align*}
%Overall, we have that $S(\ds) \ge e^{-1} \gamma = \frac{\log(n)}{ 2e p_{\max} n \diffp} $ with probability at least
%$\P\left(\sum_{i=1}^n Z_i \le \frac{2 \log(n)}{\diffp} \right)$.
The following lemma thus completes the proof of
Lemma~\ref{lemma:smooth-laplace-median-failure}.

\begin{lemma}
  \label{lemma:large-LS-median}
  Under the above setting, if $c = \frac{1}{2 p_{\max}}$, 
  \begin{equation*}
    \P\left(\sum_{i=1}^n Z_i > \frac{2 \log(n)}{\diffp} \right) 
    \le \exp\left(-\frac{p_{\min}}{16 p_{\max}} \frac{\log(n)}{ \diffp}
    + \log(2 \datarange n \diffp p_{\max})\right).
  \end{equation*}
\end{lemma}
\begin{proof}
  Let us divide the interval $[0,\datarange]$ to $\datarange/\gamma$ intervals,
  each of length $\gamma$.
  Let $B_i$ denote the number of points inside the $i$th such interval.
  Then, as $n \gamma p_{\min} \le \E[B_i] \le n \gamma p_{\max}$,
  the Chernoff bound of Lemma~\ref{lemma:chernoff} yields
  \begin{equation*}
    \P\left(B_i >  \frac{\log(n)}{\diffp} \right)
    = \P\left(B_i >  2  n \gamma p_{\max} \right)
    \le e^{-n \gamma p_{\min}/8}.
  \end{equation*}
  As evidently $\sum_{i=1}^n Z_i \le 2 \max_{i} B_i$, we apply union bound over 
  the intervals to see that
  \begin{align*}
    \P\left(\sum_{i=1}^n Z_i> \frac{2  \log(n)}{\diffp} \right)
    \le \P\left(\max_{i} B_i  >  \frac{\log(n)}{\diffp} \right)
    \le \frac{\datarange}{\gamma} e^{-n \gamma p_{\min}/8}
    =
    \frac{2 \datarange n \diffp p_{\max} }{\log n} e^{-\frac{p_{\min} n}{
        16 p_{\max} \diffp}}
  \end{align*}
  as desired.
  %% \\
  %%   & \le e^{-\left( \frac{p_{\min}}{16 p_{\max}} \frac{\log(n)}{ \diffp} - \log(2 \datarange n \diffp p_{\max})   \right)}.
  %% \end{align*}
  %Letting $c = \frac{1}{2 p_{\max}}$ proves the lemma.
\end{proof}

% -*- Mode: latex -*- %

%\section{Proofs for section~\ref{sec:examples}}

\subsection{Proofs for robust
  regression (Section~\ref{sec:example-regression-new})}
\label{sec:proofs-regression-new}

We collect the proofs of results associated with the robust regression
examples in Section~\ref{sec:proofs-regression-new} here.  In our proofs, we
use the shorthand $\risk(\theta;\dsx,\dsy) = \sum_{i=1}^n h(\<\theta,x_i\> -
y_i)$ so that $\risk_n(\theta;\dsx,\dsy) = \frac{1}{n}
\risk(\theta;\dsx,\dsy) $

\subsubsection{Proof of Lemma~\ref{lemma:inv-add-sensitivity}}
\label{sec:proof-inv-add-sensitivity}
Building on Lemma~\ref{lemma:lipschitz-score}, it is enough to prove 
that $\invmodusradd$~\eqref{eq:inv-usr-add} is $1$-Lipschitz.
Let $(\dsx,\dsy)= \{(x_i,y_i)\}_{i=1}^n$ and 
$(\dsx',\dsy') =  (\dsx,\dsy) \cup \{(x'_{n+1},y'_{n+1})\}$.
Then for every $\theta \in \Theta$, we prove 
that $|\invmodusradd(\dsx,\dsy;\theta) - \invmodusradd(\dsx',\dsy';\theta) | \le 1.$
It is clear that $\invmodusradd(\dsx,\dsy;\theta) \le \invmodusradd(\dsx',\dsy';\theta) + 1 $ as we can add one user
to $(\dsx,\dsy)$ to get the dataset $(\dsx',\dsy')$. Hence we only
need to prove that $\invmodusradd(\dsx',\dsy';\theta) \le \invmodusradd(\dsx,\dsy;\theta) + 1$.
Let $k = \invmodusradd(\dsx,\dsy;\theta)$. Therefore there 
exist $\tilde \dsx, \tilde \dsy = \{(\tilde x_i, \tilde y_i)\}_{i=1}^k$ 
such that
$\theta$ becomes a minimizer of the loss function when
we add them to the dataset $(\dsx,\dsy)$, hence
\begin{equation*}
		\nabla \risk(\theta;\dsx,\dsy) + \nabla \risk(\theta;\tilde \dsx,\tilde  \dsy) =0.
\end{equation*}
Therefore we have that 
\begin{equation*}
	\nabla \risk(\theta;\dsx',\dsy') + \nabla \risk(\theta;\tilde \dsx,\tilde  \dsy) 
		= \nabla_{\theta} h(\<\theta,x'_{n+1}\> - y'_{n+1}).
\end{equation*}
Therefore we can add $\tilde x_{k+1}, \tilde y_{k+1}$ that makes
$\theta$ a minimizer. Indeed, we only need to guarantee
that $ \nabla_{\theta} h(\<\theta,x'_{n+1}\> - y_{n+1})
= - \nabla_{\theta} h(\<\theta,\tilde x_{k+1}\> - \tilde y_{k+1}) $. Setting $\tilde x_{k+1} = - x_{n+1}$ and 
$\tilde y_{k+1} = -2 \<\theta,x_{n+1}\> + y_{n+1} $ 
proves the lemma.

\subsubsection{Proof of Lemma~\ref{lemma:inv-add-calc}}
\label{sec:proof-inv-add-calc}

A vector $\theta \in \interior \Theta$ minimizes
$\risk(\theta; \dsx', \dsy')$ if and only if
%% Assume we add $k$ new data points to $\dsx,\dsy$ which
%% we denote $x'_1,\dots,x'_k$. Then to make $\theta$ a minimizer
%% we must have:
%% \hacomment{check this part again?}
\begin{equation*}
  \nabla \risk(\theta;\dsx',\dsy')
  = \nabla \risk(\theta;\dsx,\dsy) + 
  \sum_{i=1}^{k} \nabla_{\theta} h(\<x'_i,\theta \> - y'_i) = 0.
\end{equation*}
As $h$ is $1$-Lipschitz and $\norm{x}_2 \le \Bx$, we have that
$\norm{\nabla_{\theta} h(\<x'_i,\theta \> - y'_i) } \le \Bx$ for every $1
\le i \le k$, so $\invmodusradd(\dsx,\dsy;\theta) \ge \ceil{\ltwos{\nabla
    \risk(\theta;\dsx,\dsy)} / \Bx}$. Now we show that adding $k =
\ceil{\ltwos{\nabla \risk(\theta;\dsx,\dsy)} / \Bx}$ points is
enough. First, denote $g = \nabla
\risk(\theta;\dsx,\dsy)$, and let $x'_i = - \Bx \frac{g}{\norm{g}_2}$ for every $1
\le i \le k$ and choose $y'_i \in \R$ such that $h'(\<x'_i,\theta \> - y'_i)
= 1$, for $i \le k - 1$, which is possible as $h$ is $1$-Lipschitz and
symmetric. We now have
$\nabla \risk(\theta;\dsx',\dsy')
= \nabla \risk(\theta;\dsx,\dsy) - (k-1) \Bx  \frac{g}{\norm{g}_2} =
\gamma g / \ltwo{g}$ for some $\gamma \in [0, 1]$. Take
$y_k'$ such that $h'(\<x_k', \theta\> - y_k') = \gamma$.

\subsubsection{Sampling from gamma-like distributions}
\label{sec:sample-gamma-like}

We wish to sample a vector $T \in \R^d$ with density
$\pdf(t) = \exp(-\norm{A t})$ for a matrix $A \succ 0$.
The change of variables $u = At$ and then using rotational
symmetry gives that
\begin{align*}
	\int \pdf(t) dt
	= \frac{1}{\det(A)} \int \exp(-\norm{u}) du
	& = \frac{1}{\det(A)}
	\int_0^\infty \exp(-r) \vol_{d-1}(r \sphere^{d-1}) dr \\
	& =
	\frac{1}{\det(A)}
	\frac{d \pi^{d/2}}{\Gamma(\frac{d}{2} + 1)}
	\int_0^\infty r^{d - 1} e^{-r} dr
	= \frac{d \pi^{d/2} \Gamma(d)}{
		\det(A) \Gamma(\frac{d}{2} + 1)}.
\end{align*}
In particular, to sample
$T$ with the density $\pdf(t) = \exp(-\norm{At})$, we draw
$R \sim \gammadist(d, 1)$,
then $U \sim \uniform(\sphere^{d-1})$, and
set $T = R A^{-1} U$.

\subsubsection{Proof sketch of Lemma~\ref{lemma:fast-mixing-ratio}}
\label{sec:proof-bounded-proposals}

To simplify notation, we analyze the case
when $\norm{x_i}_2 \le 1$, and we follow the empirical
process notation of Sec.~\ref{sec:partial-minimization}.

By Taylor's theorem, the Lipschitz continuity of $\nabla \loss$ implies
there exists an error function $E_n : \R^d \to \R^{d \times d}$ satisfying
$\norm{E_n(\theta)} \le (P_n \lipgrad) \norm{\theta - \theta_n}$
for which
\begin{align*}
	\nabla \riskn(\theta)
	& = \nabla \riskn(\theta_n) + (\nabla^2 \riskn(\theta_n) + E_n(\theta))
	(\theta - \theta_n).
\end{align*}

Let $Z_\pi
= \int_\Theta \exp(-n \diffp \norm{\nabla \riskn(\theta)} / 2) d\theta$
and $Z_q = \int_\Theta \exp(-\frac{n \diffp}{2} (\norm{\nabla^2
	\riskn(\theta_n) (\theta - \theta_n)} \wedge r_n)) d\theta$ be
the normalizing constants for the densities $\pi$ and $q$.
We would like to demonstrate that $q(t) / \pi(t) \ge \beta > 0$ for
all $t$. We consider two cases.
\begin{enumerate}[(i)]
\item Assume that $\norm{\nabla^2 \riskn(\theta_n)(\theta - \theta_n)}
  \le r_n$. In this case, the triangle inequality implies that
  \begin{align*}
    \frac{q(\theta)}{\pi(\theta)}
    & \ge \frac{Z_\pi}{Z_q}
    \exp\left(-\frac{\diffp n}{2}
    \norm{E_n(\theta)(\theta - \theta_n)}\right) \\
    & \ge \frac{Z_\pi}{Z_q}
    \exp\left(-\frac{\diffp n}{2}
    (P_n \lipgrad) \norm{\theta - \theta_n}^2\right)
    \ge \frac{Z_\pi}{Z_q}
    \exp(-\diffp O(n r_n^2))
    = \frac{Z_\pi}{Z_q} (1 - o(1)).
  \end{align*}
\item Assume that $\norm{\nabla^2 \riskn(\theta_n)(\theta - \theta_n)} \ge
  r_n$.
  This case is a bit more subtle.
  We consider two regimes. In the first, we have $\norm{\nabla^2
    \riskn(\theta_n)(\theta - \theta_n)} \in [r_n, r_n^{2/3}]$.
  Then
  \begin{equation*}
    \norm{\nabla \riskn(\theta_n)}
    \ge \norm{\nabla^2 \riskn(\theta_n)(\theta - \theta_n)}
    -  P_n \lipgrad \norm{\theta - \theta_n}^2
    \ge r_n - O(r_n^{4/3}).
  \end{equation*}
  We consequently obtain
  \begin{equation*}
    \frac{q(\theta)}{\pi(\theta)}
    = \frac{Z_\pi}{Z_q} \exp\left(-\frac{n \diffp r_n}{2}
    + \frac{n \diffp}{2} \norm{\nabla^2 \riskn(\theta_n)(\theta - \theta_n)}
    - O(n r_n^{4/3})\right)
    \ge \frac{Z_\pi}{Z_q} \exp(-O(n r_n^{4/3})).
  \end{equation*}
  In the regime that $\norm{\nabla^2 \riskn(\theta_n)(\theta - \theta_n)}
  \ge r_n^{2/3}$,
  we require a bit more work.
  First, we define $f_n(\theta) = \<\nabla \riskn(\theta), \theta - \theta_n\>
  / \norm{\theta - \theta_n}$,
  noting by convexity that
  $r \mapsto f(\theta_n + r \frac{\theta - \theta_n}{\ltwos{\theta - \theta_n}})$
  is monotone increasing. As $r_n \le 1$,
  if we define
  $u = r_n \frac{\theta - \theta_n}{\norms{\theta - \theta_n}}$,
  this monotonicity implies
  \begin{align*}
    \norm{\nabla \riskn(\theta)}
    & \ge f_n(\theta)
    = f_n(\theta_n + (\theta - \theta_n))
    \ge f_n(\theta_n + u) \\
    & = \<\nabla \riskn(\theta_n + u), u / \norm{u}\>
    = \<\nabla^2 \riskn(\theta_n) u, u / \norm{u}\>
    \pm \<E_n(\theta_n + u) u, u / \norm{u}\> \\
    & \ge \<\nabla^2 \riskn(\theta_n) u, u / \norm{u}\>
    - P_n \lipgrad \norm{u}^2
    \ge \lambda \norm{u} - P_n\lipgrad \norm{u}^2
    = \lambda r_n^{2/3} - O(r_n^{4/3}).
  \end{align*}
  In particular,
  $q(\theta) / \pi(\theta) \ge \frac{Z_\pi}{Z_q}
  \exp(-\frac{n \diffp r_n}{2} + \frac{n \diffp r_n^{2/3}}{2}
  - O(n r_n^{4/3})) \gg 1$,
  as $n^{-1} \ll r_n \ll n^{-3/4}$.
\end{enumerate}
In either case, our condition that $n^{-1} \ll r_n \ll n^{-3/4}$
gives that
$\frac{q(\theta)}{\pi(\theta)} \ge \frac{Z_\pi}{Z_q} (1 - o(1))$ as
$n$ grows.
A similar calculation to the two cases above gives that
$Z_\pi / Z_q \gtrsim 1$.

\bibliography{bib}
\bibliographystyle{abbrvnat}

\end{document}